\documentclass[oneside,12pt]{book}
\usepackage[margin=2.5cm,bindingoffset=1cm]{geometry}
\usepackage{setspace}
\usepackage{microtype}

\usepackage[utf8]{inputenc}

\title{\bf Distributed Optimization for \\ Smart Cyber-Physical Networks}

\date{\today}                           %

\author{
  Notarstefano, Giuseppe\\
  Universit\`a di Bologna, Bologna (Italy)\\
  \texttt{giuseppe.notarstefano@unibo.it}
  \and
  Notarnicola, Ivano\\
  Universit\`a di Bologna, Bologna (Italy)\\
  \texttt{ivano.notarnicola@unibo.it}
  \and
  Camisa, Andrea\\
  Universit\`a di Bologna, Bologna (Italy)\\
  \texttt{a.camisa@unibo.it}
}

\usepackage{amsfonts,amssymb,amsmath,amsthm}
\usepackage{mathrsfs}
\allowdisplaybreaks
\usepackage{bookmark}
\usepackage{algorithm}
\usepackage[noend]{algpseudocode}
\makeatletter
\newcommand{\StatexIndent}[1][3]{%
  \setlength\@tempdima{\algorithmicindent}%
  \Statex\hskip\dimexpr#1\@tempdima\relax}
\makeatother

\newtheorem{theorem}{Theorem}
\newtheorem{assumption}[theorem]{Assumption}
\newtheorem{definition}[theorem]{Definition}
\newtheorem{lemma}[theorem]{Lemma}
\newtheorem{proposition}[theorem]{Proposition}
\newtheorem{remark}[theorem]{Remark}
\newcommand\oprocendsymbol{\hbox{$\square$}}
\newcommand\oprocend{\relax\ifmmode\else\unskip\hfill\fi\oprocendsymbol}
\def\eqoprocend{\tag*{$\square$}}

\usepackage{soul}
\usepackage{accents}
\newcommand \ubar[1]{%
  \underaccent{\bar}{#1}}

\newcommand{\CC}{\mathcal{C}}
\newcommand{\DD}{\mathcal{D}}
\newcommand{\EE}{\mathcal{E}}
\newcommand{\GG}{\mathcal{G}}
\newcommand{\KK}{\mathcal{K}}
\newcommand{\LL}{\mathcal{L}}
\newcommand{\PP}{\mathcal{P}}
\renewcommand{\SS}{\mathcal{S}}
\newcommand{\XX}{\mathcal{X}}

\newcommand{\avgx}{\bar{\mathbf{x}}}
\newcommand{\avgy}{\bar{\mathbf{y}}}
\newcommand{\avgz}{\bar{\mathbf{z}}}

\newcommand{\bX}{\mathbf{X}}
\newcommand{\bZ}{\mathbf{Z}}

\newcommand{\bg}{\mathbf{g}}
\newcommand{\bh}{\mathbf{h}}
\newcommand{\br}{\mathbf{r}}
\newcommand{\bs}{\mathbf{s}}
\newcommand{\bu}{\mathbf{u}}
\newcommand{\bv}{\mathbf{v}}
\newcommand{\bw}{\mathbf{w}}
\newcommand{\bx}{\mathbf{x}}
\newcommand{\by}{\mathbf{y}}
\newcommand{\bz}{\mathbf{z}}
\newcommand{\1}{\mathbf{1}}
\newcommand{\0}{\mathbf{0}}

\newcommand{\subgrad}{\widetilde{\nabla}}

\newcommand{\bepsilon}{\boldsymbol{\epsilon}}

\newcommand{\bmu}{\boldsymbol{\mu}}

\newcommand{\blambda}{\boldsymbol{\lambda}}
\newcommand{\bLambda}{\boldsymbol{\Lambda}}

\newcommand{\bxi}{\boldsymbol{\xi}}

\newcommand{\real}{{\mathbb{R}}}
\renewcommand{\natural}{{\mathbb{N}}}
\newcommand{\integer}{{\mathbb{Z}}}
\newcommand{\nbrs}{\mathcal{N}}

\newcommand{\until}[1]{\{1,\ldots,#1\}}
\newcommand{\fromto}[2]{\{#1,\ldots,#2\}}
\newcommand{\subj}{\text{subj.\ to}}
\newcommand{\map}[3]{#1: #2 \rightarrow #3}

\newcommand{\lexsucc}{\stackrel{L}{>}}
\newcommand{\lexprec}{\stackrel{L}{<}}
\renewcommand{\inf}{\operatornamewithlimits{inf\vphantom{p}}}
\renewcommand{\liminf}{\operatornamewithlimits{liminf\vphantom{p}}}

\renewcommand{\lim}{\operatornamewithlimits{lim\vphantom{p}}}

\newcommand{\relint}{\mathop{\rm relint}}
\newcommand{\argmin}{\mathop{\rm argmin}}

\newcommand{\lexmin}{\mathop{\rm lexmin}}
\newcommand{\diam}{\mathop{\rm diam}}

\newcommand{\gen}{\texttt{gen}}
\newcommand{\GEN}{\texttt{GEN}}
\newcommand{\stor}{\texttt{stor}}
\newcommand{\STOR}{\texttt{STOR}}
\newcommand{\cload}{\texttt{conl}}
\newcommand{\CLOAD}{\texttt{CONL}}
\newcommand{\des}{\texttt{des}}
\newcommand{\trade}{\texttt{tr}}

\newcommand{\smallsum}{\textstyle\sum\limits}

\usepackage[dvipsnames]{xcolor}

\usepackage{tikz}

\graphicspath{{figs/},{figs/tikz/},{figs/simulations/}}

\definecolor{blue@O4S}{RGB}{0, 41, 69}
\definecolor{emph@O4S}{RGB}{0, 93, 137}
\definecolor{red@O4S}{RGB}{127,0,0}
\definecolor{gray@O4S}{RGB}{112, 112, 112}

\def\DistrGradTrack/{Distributed Gradient Tracking}

\def\DistrSubgr/{Distributed Subgradient}

\def\ConstrCons/{Constraints Consensus}
\def\DistrSimplex/{Distributed Simplex}

\begin{document}
\pagestyle{plain}

\maketitle

\iftrue

\clearpage
\thispagestyle{empty}

\null\vfill

\begin{center}
  \textbf{Abstract}
\end{center}
  The presence of embedded electronics and communication capabilities as well as
  sensing and control in smart devices has given rise to the novel concept of
  cyber-physical networks, in which agents aim at cooperatively solving complex
  tasks by local computation and communication.
  Numerous estimation, learning, decision and control tasks in smart networks
  involve the solution of large-scale, structured optimization problems in which
  network agents have only a partial knowledge of the whole problem.
  Distributed optimization aims at designing local computation and communication
  rules for the network processors allowing them to cooperatively solve the
  global optimization problem without relying on any central unit.
  The purpose of this survey is to provide an introduction to distributed
  optimization methodologies.
  Principal approaches, namely (primal) consensus-based, duality-based and
  constraint exchange methods, are formalized. An analysis of the basic schemes
  is supplied, and state-of-the-art extensions are reviewed.

\vfill\null

\clearpage
\pagenumbering{roman} %
\onehalfspacing

\tableofcontents

\clearpage
\pagenumbering{arabic} %

\chapter*{Introduction}
\label{chap:introduction} 
\addcontentsline{toc}{chapter}{Introduction}

\section*{Motivation}
In recent years, the breakthroughs in embedded electronics
are giving the opportunity to include computation
and communication capabilities in almost any device of several domains as
factories, farms, buildings, grids and cities.
Communication among devices has enabled a number of new challenges along the
direction of turning smart devices into smart (cooperating) systems. The keyword
``cyber-physical networks'' is being adopted to refer to this permeating
reality, whose distinctive feature is that a great advantage can be obtained if
its interconnected, complex nature is exploited.
A novel peer-to-peer \emph{distributed} computational framework is emerging as
a new opportunity in which peer processors, communicating over a network,
cooperatively solve a task without resorting to a unique provider that knows and
owns all the data.

Several challenges arising in cyber-physical networks can be stated as
optimization problems. Examples are estimation, decision, learning and
control applications.
To solve optimization problems over cyber-physical networks, it is not
possible to apply the classical optimization algorithms (that we call
\emph{centralized}), which require the data to be managed by a single entity.
In fact, the problem data are spread over the network, and it is undesirable
(or even impossible) to collect them at a unique node.
To this end, parallel computing serves as a source of inspiration.
In order to speed up the solution of large-scale optimization problems,
several effort has been made in designing \emph{parallel} algorithms
by splitting the computational burden among several processors.
However, for typical parallel optimization algorithms, a central coordinating
node is required and the communication topology is designed ad hoc.
In distributed computation the communication topology cannot
be thought of as a design parameter. Rather, it is a given part of the problem.
Thus, in cyber-physical networks, the goal is to design algorithms, based
on the exchange of information among the processors, that take advantage
of the aggregated computational power. All the agents must be treated as peers and
each of them must perform the same tasks and no ``master'' node must be present.
Moreover, information privacy is often a requirement (i.e., private problem
data at each node must not be shared with the other nodes).
These challenges call for tailored strategies and have given rise to a novel,
growing research branch termed \emph{distributed optimization}.

\section*{Scope of the Monograph}
The purpose of this survey is to give a comprehensive overview of the most
common approaches used to design distributed optimization algorithms, together
with the theoretical analysis of the main schemes in their basic version.
We identify and formalize classes of problem set-ups that arise in motivating 
application scenarios. For each set-up, in order to give the main tools for analysis,
we review tailored distributed algorithms in simplified cases. Extensions and generalizations
of the basic schemes are also discussed at the end of each chapter.
The algorithms have been developed by combining mathematical tools from
optimization theory (e.g., duality) and network control theory (e.g., average
consensus).
For some of the discussed algorithms, we will present also parallel algorithms that serve
as a starting point for the development of distributed methods.

We focus on three main categories of distributed optimization approaches:
\emph{(i)} primal consensus-based methods, i.e., methods combining
classical gradient or subgradient steps with local averaging schemes;
\emph{(ii)} dual methods, i.e., methods which employ the Lagrangian dual of suitable
equivalent formulations of the target problem to obtain a distributed routine;
\emph{(iii)} constraint exchange methods, which are based on the exchange of
(active) constraints among agents to compute a solution of the considered
problem.

Survey papers on distributed optimization have been proposed in the
literature. An early survey paper presenting a broad class of relevant optimization
problems in control is~\cite{necoara2011parallel}.
It also discusses tailored, parallel and distributed optimization algorithms
based on decomposition techniques and including also 
the distributed subgradient method.
Recent surveys analyze thoroughly average consensus~\cite{nedic2015convergence} 
and the distributed subgradient method~\cite{nedic2015convergence,nedic2018distributed,nedic2018network},
with a literature review on other distributed optimization techniques.
The book~\cite{nedic2018multi} provides parallel and distributed asynchronous optimization
algorithms, including gradient tracking techniques.
Some latest advances in distributed optimization are collected in~\cite{giselsson2018large}.

\section*{Organization}

In Chapter~\ref{chap:framework}, we introduce the relevant problem set-ups,
that we call \emph{cost-coupled}, \emph{constraint-coupled} and \emph{common cost},
along with several motivating applications of interest arising in estimation,
learning, decision and control.
In Chapter~\ref{chap:primal} we provide an overview of primal approaches to solve 
cost-coupled problems, namely the distributed subgradient algorithm and the 
gradient tracking algorithm.
In Chapter~\ref{chap:dual}, a discussion on relevant duality forms for distributed
optimization is first provided, and then distributed algorithms relying on Lagrangian
approaches are reviewed. Namely, for cost-coupled
problems, distributed dual decomposition and distributed ADMM algorithms
are considered, while for constraint-coupled problems, a distributed dual subgradient
algorithm and a method based on relaxation and successive distributed
decomposition are presented.
In Chapter~\ref{chap:constraint_exchange}, we focus on constraint exchange
methods. We introduce the Constraints Consensus algorithm applied to
common-cost problems, along with its most relevant extensions.

We also provide illustrative numerical examples to highlight
significant properties of the considered distributed optimization methods.
Since the described algorithms are designed for different problem set-ups,
different, relevant simulation scenarios are considered in each chapter.
 
\fi

\iftrue

\chapter{Distributed Optimization Framework}
\label{chap:framework}

In this chapter we introduce the conceptual framework for distributed
optimization in peer-to-peer networks. First, we describe the network model we
will consider throughout the survey. Then we present and motivate the main
optimization set-ups that are of interest in smart networks.

In a distributed scenario, we consider $N$ units, called \emph{agents} or 
\emph{processors}, that have both communication and computation capabilities.
Communication among agents is modeled by means of graph theory.  Informally,
given a graph $\GG$ with $N$ nodes, one for each agent, an agent $i$ can send
(receive) data to (from) another agent $j$, %
when the graph $\GG$ contains an edge connecting $i$ to $j$ ($j$ to $i$). %
In a distributed algorithm, agents initialize their local states and then 
start an iterative procedure in which communication and computation steps are 
iteratively performed, with all the nodes performing the same actions.  %
In particular, local states are updated by using only information received by in-neighbors.

In this survey we consider a distributed framework in which agents cooperatively
solve an optimization problem.
The basic assumption we make is that each agent $i$ has only a partial knowledge
of the entire problem, e.g., only a portion of the cost and/or a portion of the
constraints is locally available.
In the rest of the chapter, depending on the specific optimization set-up, we
will clarify what do we mean by cooperation among agents for the solution of
a given optimization problem.

\begin{remark}
We point out that, regardless of the optimization problem structure, our standing assumption is 
that the distributed framework is made by cooperative agents. There is another strain of research 
on non-cooperative set-ups with applications to game theoretic problems.
A non-exhaustive list of early references is~\cite{stankovic2011distributed,li2013designing,yang2010distributed}.
\oprocend
\end{remark}

\section{Distributed Computation Model}
\label{sec:network_comm_models}

In this section we formally define the communication model for a distributed algorithm.
A network is modeled as a (possibly time-dependent) directed
graph $\GG^t = (\until{N}, \EE^t)$, where $t \in \natural$ is a universal
(slotted) time, $\until{N}$ is the (fixed) set of agent identifiers and
$\EE^t \subseteq \until{N} \times \until{N}$, for all $t \ge 0$, is the
(time-dependent) set of (directed) edges over the vertices $\until{N}$, which
represents the communication links.  A graphical representation of a time-varying
network is given in Figure~\ref{fig:directed_graph}.
\begin{figure}[htpb]
\centering
  \includegraphics[scale=1]{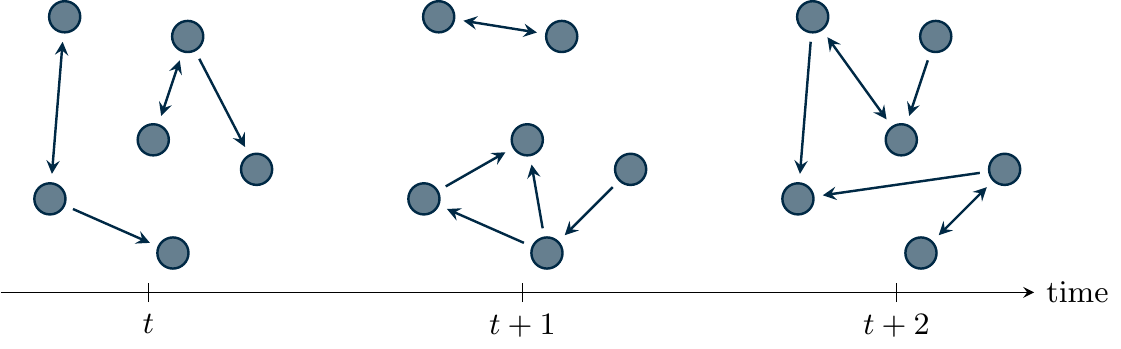}
  \caption{A directed time-varying graph of $N= 6$ nodes.}
  \label{fig:directed_graph}
\end{figure}

At each (universal) time instant $t$, a communication structure, i.e., a graph $\GG^t$, 
is active. The time-varying edge set $\EE^t$ models the communication in the sense that 
at time $t$ there is an edge from node $i$ to node $j$ in $\EE^t$ if and only if processor 
$i$ transmits information to processor $j$ at time $t$.
Given an edge $(i,j)\in \EE^t$, $i$ is called \emph{in-neighbor} of $j$ and $j$ is an 
\emph{out-neighbor} of $i$ at time $t$.
When the edge set $\EE^t$ does not depend on $t$, i.e., $\GG^t \equiv \GG$ for all $t$, we say  
that the network is fixed, otherwise the network is time-varying.
Moreover, when for every pair of nodes $i$ and $j$ in the network the edge $(i,j)$ and
the edge $(j,i)$ are in $\EE^t$, then the graph is undirected. An example of a directed and
of an undirected graph is depicted in Figure~\ref{fig:undirected_graph}.
\begin{figure}[htpb]
\centering
  \includegraphics[scale=1]{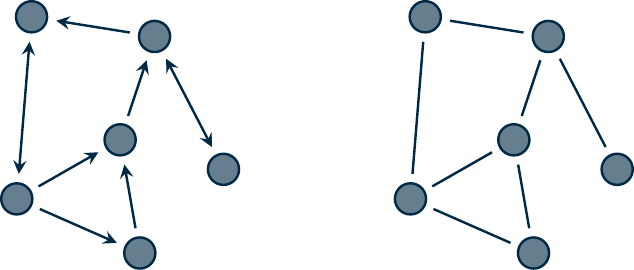}
  \caption{A directed (left) and an undirected (right) graph of $N= 6$ nodes.}
  \label{fig:undirected_graph}
\end{figure}

Given a fixed graph $\GG$, connectivity properties can be stated.
\begin{definition}
\label{def:fixed_graph_connectivity}
  A fixed directed graph $\GG$ is said to be \emph{strongly connected}
  if for every pair of nodes $(i,j)$ there exists a path of directed edges 
  that goes from $i$ to $j$. If $\GG$ is undirected, we say that $\GG$ is
  \emph{connected}.\oprocend
\end{definition}
Connectivity properties can be also stated for time-varying topologies
(we only consider directed graphs).
\begin{definition}
\label{def:time_varying_graph_connectivity}
  A time-varying directed graph $\GG^t$, $t \in \natural$, is said to be
  \begin{itemize}
    \item \emph{jointly strongly connected} if the graph $\GG_\infty^t \triangleq
      (\until{N}, \EE_\infty^t)$, with $\EE_\infty^t = \bigcup_{\tau=t}^\infty \: \EE^\tau$,
      is strongly connected for all $t \ge 0$.
    
    \item \emph{$T$-strongly connected} (or \emph{uniformly jointly strongly connected})
    if there exists a scalar $T > 0$ such that the graph $\GG_T^t \triangleq
    (\until{N}, \EE_T^t)$ with $\EE_T^t = \bigcup_{\tau=0}^{T-1} \EE^{t+\tau}$,
    is strongly connected for every $t\ge 0$.\oprocend
  \end{itemize}
\end{definition}

Given a network topology, agents can run distributed algorithms according to several
communication protocols.
When the steps of the algorithm explicitly depend on the value of $t$, we say
that the algorithm is \emph{synchronous}, i.e., agents must be aware of the current
value of $t$ and, thus, their local operations must be synchronized to a global clock.
We will also consider a communication protocol in which agents are not aware
of any global time information, i.e., their updates do not depend on $t$,
and we term these algorithms \emph{asynchronous}.
In fact, if a distributed algorithm is designed to run over a jointly strongly connected
graph, and the local computation steps do not depend on $t$, then
the algorithm can be also implemented in an asynchronous network.

\section{Optimization Set-ups}
\label{sec:setups}

In this section we describe three general optimization set-ups that comprise
several estimation, learning, decision and control application scenarios in
smart networks.
A distributed optimization algorithm for such classes of problems consists of
an iterative procedure based on the distributed computation model introduced in
Section~\ref{sec:network_comm_models}.
The goal for the agents is to eventually obtain a solution of the investigated
problem. In each considered optimization set-up, this goal translates to
different statements that will be formally specified next.

For an optimization algorithm, the aim is to minimize a scalar objective
function (or cost function), usually denoted as
$f(\bx)$, where $\bx \in \real^d$ is the decision variable. We may need to restrict
the minimizer of $f$ in a given constraint set $X \subseteq \real^d$ (or feasible set).
From now on, we use the symbol $\min$ to denote that we want to minimize $f(\bx)$
subject to the constraints, and we compactly write the overall optimization problem as
\begin{align*}
  \min_\bx \: & \: f(\bx)
  \\
  \subj \: & \: \bx \in X.
\end{align*}
The generic constraint set $X$ can also be expressed by means of equalities or
inequalities as, e.g., $h_j(\bx) = 0$ for $j \in \until{p}$, or $g_k(\bx) \le 0$
for $k \in \until{q}$, for some functions $h_j$ and $g_k$.
The equality and inequality constraints are usually compactly denoted as
$\bh(\bx) = \0$ or $\bg(\bx) \le \0$.
Centralized methods to approach this problem can be found
in~\cite{bertsekas1999nonlinear,beck2017first}.

In the remainder of this section, we introduce three structured versions of the above
general optimization problem.

\subsection{Cost-Coupled Optimization}
\label{sec:setups_cost_coupled}

We start by introducing an optimization set-up in which the cost function is expressed as 
the sum of local contributions $f_i$ and all of them depend on a common optimization variable $\bx$.
Formally, the set-up is
\begin{align}
\begin{split}
  \min_{\bx \in \real^d} \: & \: \smallsum_{i=1}^N f_i(\bx)
  \\
  \subj \: & \: \bx \in X,
\end{split}
\label{setups:cost-coupled_problem}
\end{align}
where $\bx \in \real^d$ and $X \subseteq \real^d$. The global constraint set $X$
is assumed to be common to each agent, while $\map{f_i}{\real^d}{\real}$
is assumed to be known by agent $i$ only, for all $i\in\until{N}$.
Figure~\ref{fig:setups:graph_info_cost_coupled} provides a graphical representation
of how problem information is spread over the network.
\begin{figure}[!htpb]
\centering
  \includegraphics[scale=1]{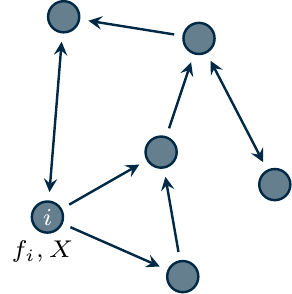}
  \caption{Cost-coupled set-up: each agent $i$ only knows $f_i$ and $X$
  and communicates according to the edges of the graph}.
  \label{fig:setups:graph_info_cost_coupled}
\end{figure}

\noindent More general versions of this optimization set-up assume that the constraint set is
more structured, e.g., $X = \bigcap_{i=1}^N X_i$, where each $X_i$
is known by agent $i$ only.

Let $\bx^\star$ denote an optimal solution of problem~\eqref{setups:cost-coupled_problem}.
For this optimization set-up, the goal is to design a distributed algorithm where
each agent updates a local estimate $\bx_i^t$ that converges (asymptotically
or in finite time) to $\bx^\star$, by means of local computation and neighboring 
communication only.
An illustrative scheme is depicted in Figure~\ref{fig:setups:cost_coupled}.
\begin{figure}[htpb]
\centering
  \includegraphics[scale=1]{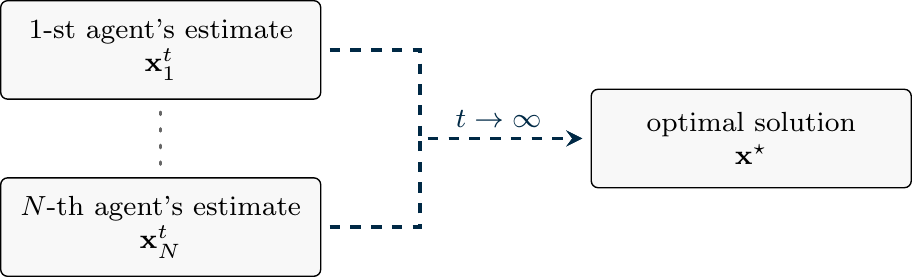}
  \caption{Illustrative scheme of the goal for the cost-coupled set-up and for the common cost set-up.}
  \label{fig:setups:cost_coupled}
\end{figure}

\begin{remark}
\label{setups:partitioned_optimization}
  An interesting optimization set-up arising in several applications is the
  so-called \emph{partitioned}, or \emph{partition-based}, set-up, first introduced in
  \cite{erseghe2012distributed}.
  The problem
  is in the form~\eqref{setups:cost-coupled_problem}, but the cost function and
  the constraints of each agent do not involve all the components of the
  decision variable, but rather they depend only on some of its components.
  This sparsity in the problem can be modeled using a graph.
  Formally, the partitioned optimization set-up is
  \begin{align*}
	  \min_{\bx} \: & \: \smallsum_{i=1}^N f_i \big( \bx_i, \{\bx_j\}_{j \in \nbrs_i} \big)
	  \\
	  \subj \: & \: \big( \bx_i, \{\bx_j\}_{j \in \nbrs_i} \big) \in X_i,
	    \hspace{1cm}
	    i \in \until{N},
  \end{align*}
  where $\bx$ denotes the vector stacking $(\bx_1,\ldots, \bx_N)$, while the notation 
  $f_i ( \bx_i, \{\bx_j\}_{j \in \nbrs_i} )$ highlights the fact that $f_i$ actually 
  depends only on the components of $\bx$ indexed by $\{i\}\cup\nbrs_i$.
  Distributed algorithms have been developed to solve partitioned problems.
  Remark~\ref{dual:partitioned_remark} discusses how to tailor
  algorithms based on dual decomposition in order to take into account the
  partitioned structure.\oprocend
\end{remark}

\subsection{Common Cost Optimization}
\label{sec:setups_common_cost}
Another important set-up arising in several applications is given by
\begin{align}
\begin{split}
  \min_{\bx \in \real^d } \:\: &\: f (\bx )
  \\
  \subj \: &\: \bx \in \bigcap_{i=1}^N X_i,
\end{split}
\label{setups:common_cost_problem}
\end{align}
where $\map{f}{\real^d}{\real}$ is known by all the agents while
each constraint $X_i \subseteq \real^d$ is known by agent $i$ only,
for all $i\in\until{N}$.
Figure~\ref{fig:setups:graph_info_common_cost} provides a graphical representation
of how information is spread over the network.
\begin{figure}[!htpb]
\centering
  \includegraphics[scale=1]{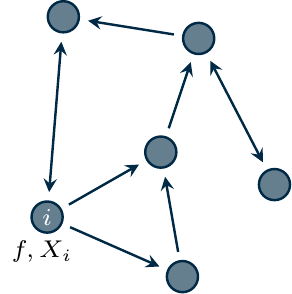}
  \caption{Common cost set-up: each agent $i$ only knows $f$ and $X_i$
  and communicates according to the edges of the graph.}
  \label{fig:setups:graph_info_common_cost}
\end{figure}

The common cost set-up~\eqref{setups:common_cost_problem} is somehow similar to the cost-coupled
set-up~\eqref{setups:cost-coupled_problem}, 
since in both cases the optimization variable is shared among the processors.
However, in the common cost set-up~\eqref{setups:common_cost_problem}, the cost
function is shared, and the coupling among the agents is due to the fact that the
optimization variable must belong to all the local constraint sets.
It is possible to think of problem~\eqref{setups:common_cost_problem}
as a special case of the cost-coupled set-up~\eqref{setups:cost-coupled_problem} (with $X = \bigcap_{i=1}^N X_i$)
by setting each $f_i(\bx) = 1/N \cdot f(\bx)$. However, notice that a commonly known cost
function explicitly allows for tailored distributed optimization algorithms such as, e.g.,
constraint exchange methods (cf. Chapter~\ref{chap:constraint_exchange}).

Let $\bx^\star$ denote an optimal solution of problem~\eqref{setups:common_cost_problem}.
For such optimization set-up, the goal is to design a distributed algorithm where
each agent updates a local estimate $\bx_i^t$ that converges (asymptotically
or in finite time) to $\bx^\star$, by means of local computation and neighboring communication only
(cf. Figure~\ref{fig:setups:cost_coupled}).

\subsection{Constraint-Coupled Optimization}
\label{sec:setups_constraint_coupled}

In this subsection, we present a different set-up %
which we call constraint-coupled.
Agents in a network want to minimize the sum of local cost functions, 
each one depending only on a local vector satisfying local constraints.
The decision vectors are then coupled to each other by means of
separable coupling constraints.
This feature leads easily to the so-called big-data problems having a very highly dimensional decision 
variable that grows with the network size. However, since agents are typically interested in computing 
only their (small) portion of an optimal solution, novel tailored methods need to be developed to 
address these challenges.

Formally, the constraint-coupled optimization problem is
\begin{align}
\begin{split}
  \min_{\bx_1,\ldots,\bx_N} \: & \: \smallsum_{i=1}^N f_i(\bx_i)
  \\
  \subj \: & \: \bx_i \in X_i, \hspace{1cm}  i \in \until{N}
  \\
  & \: \smallsum_{i=1}^N \bg_i (\bx_i) \le \0,
\end{split}
\label{setups:constraint-coupled_problem}
\end{align}
where $(\bx_1,\ldots,\bx_N)$ is the global optimization vector stacking 
all the local variables, $X_i \subseteq \real^{d_i}$, $\map{f_i}{\real^{d_i} }{\real}$ and 
$\map{\bg_i}{\real^{d_i}}{\real^S}$ are known by agent $i$ only, for all $i\in\until{N}$.
Notice that problem~\eqref{setups:constraint-coupled_problem} is challenging
because of the coupling constraints $\sum_{i=1}^N \bg_i (\bx_i) \le \0$.
If there were no coupling constraints, the optimization would trivially split into
$N$ independent problems.
Figure~\ref{fig:setups:graph_info_constraint_coupled} provides a graphical
representation of how information is spread over the network.
\begin{figure}[!htpb]
\centering
  \includegraphics[scale=1]{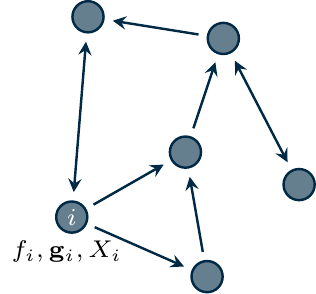}
  \caption{Constraint-coupled set-up: each agent $i$ only knows $f_i$, $X_i$ and $\bg_i$
  and communicates according to the edges of the graph.}
  \label{fig:setups:graph_info_constraint_coupled}
\end{figure}

Let $(\bx_1^\star, \ldots, \bx_N^\star)$ denote an optimal solution of
problem~\eqref{setups:constraint-coupled_problem}.
The goal is to design a distributed algorithm where each agent
updates a local estimate $\bx_i^t$ that converges (asymptotically or
in finite time) to $\bx_i^\star$,
the $i$-th portion of $(\bx_1^\star, \ldots, \bx_N^\star)$, by means 
of local computation and neighboring communication only.
An illustrative scheme is depicted in Figure~\ref{fig:setups:constraint_coupled}.
\begin{figure}[htpb]
\centering
  \includegraphics[scale=1]{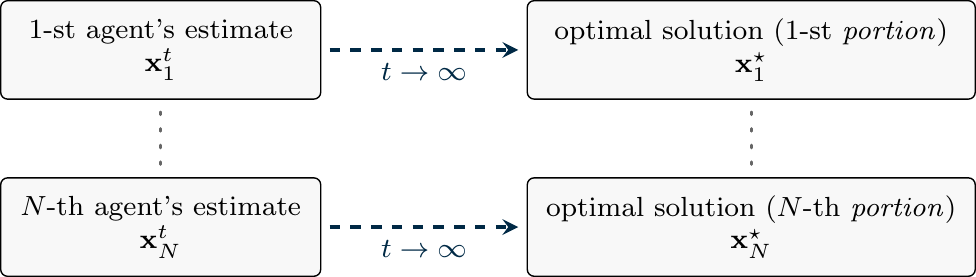}
  \caption{Illustrative scheme of the goal for the constraint-coupled set-up.}
  \label{fig:setups:constraint_coupled}
\end{figure}

A special instance of this set-up has been investigated in the context of resource
allocation, where the coupling constraint is linear, e.g., $\sum_{i=1}^N \bx_i = b$,
and there are no local constraints.
In this survey, we consider more general problems where the 
coupling may be nonlinear and local constraints are explicitly taken into account.

\begin{remark}[Comparison with the cost-coupled set-up]
  We notice that problem~\eqref{setups:cost-coupled_problem}
  can be cast as~\eqref{setups:constraint-coupled_problem} by introducing
  copies $\bx_1, \ldots, \bx_N$ of the decision vector $\bx$ and appropriate
  coherence (coupling) constraints, i.e.,
  \begin{align*}
	\begin{split}
	  \min_{\bx_1, \ldots, \bx_N} \: & \: \smallsum_{i=1}^N f_i(\bx_i)
	  \\
	  \subj \: & \: \bx_i \in X, \hspace{1cm} i \in \until{N}
	  \\
	  & \: \bx_1 = \bx_2
	  \\
	  & \hspace{0.5cm} \vdots
	  \\
	  & \: \bx_{N-1} = \bx_N
	\end{split}
	\end{align*}
  However, it is worth noticing that the coupling constraint of such reformulation enjoys
  a special, sparse structure while the constraints in~\eqref{setups:constraint-coupled_problem}
  are more general (since they involve all the agents in the network).
  \oprocend
\end{remark}

\section{Optimization Set-ups for Learning and Control}
\label{sec:setups_application}

In this section, we motivate the study of the optimization set-ups introduced in
Section~\eqref{sec:setups} by describing important application scenarios that
are of interest in control and robotics as well as communication and signal
processing.

\subsection{Regression for Data Analytics}
\label{sec:regression}

Let us consider an important task for several applications, namely the linear
\emph{regression} problem, in which we assume that a set of points in
a training dataset is used to estimate the parameters of a model
(assumed to be linear in the parameters).
The model can be exploited, e.g., to predict new generated samples.
Figure~\ref{setups:fig_regession_centralized} proposes a pictorial representation
of a simple scenario in $\real^2$.
\begin{figure}[!htpb]
\centering
  \includegraphics[scale=1]{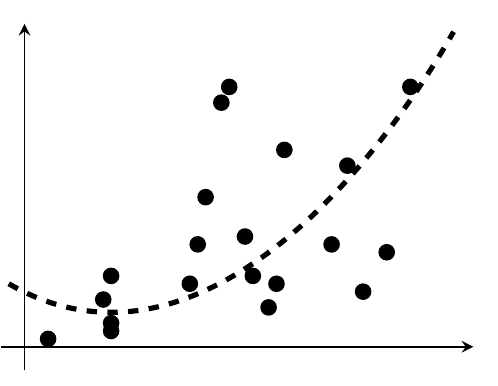}
  \caption{Set of data points in $\real^2$ that can be fit using a polynomial
	  model (i.e., linear in the parameters). The coefficients of the polynomial are
	  obtained with a regression approach.}
  \label{setups:fig_regession_centralized}
\end{figure}

Nowadays, especially in big-data contexts, a natural scenario is to assume that
the training data are not (or cannot be) gathered at a main collection center.
Rather, it is reasonable to assume that the samples are (spatially) distributed in a
network, as shown in Figure~\ref{setups:fig_regression_network}.
\begin{figure}[!htpb]
\centering
  \includegraphics[scale=1]{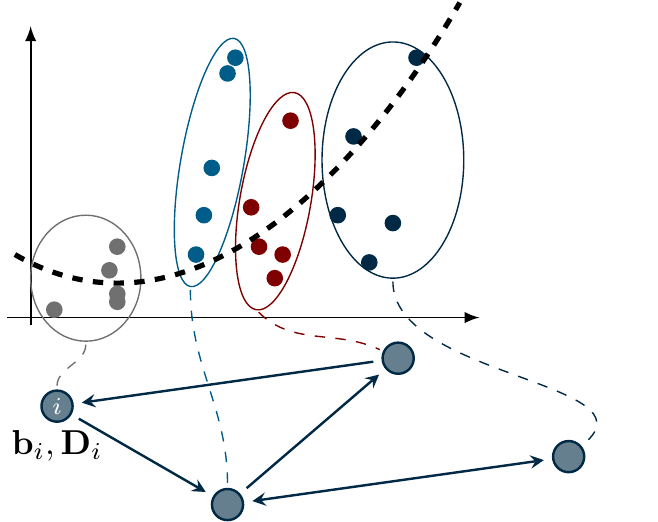}
  \caption{Regression problem over a network of $4$ agents.}
  \label{setups:fig_regression_network}
\end{figure}

\noindent Now, let us focus on Least Squares (LS), a popular regression approach.
Assume that $N$ processors in a network want to solve a
regression problem, where $\bx \in \real^d$ denotes the parameter vector that
has to be estimated, and each agent $i$ has $n_i$ observations. The (unweighted)
LS problem can be formulated as
\begin{align}
  \label{eq:LS}
  \min_{\bx} \:\: \smallsum_{i=1}^N\| \mathbf{D}_i \bx - \mathbf{b}_i \|^2
\end{align}
where, for all $i \in \until{N}$, $\mathbf{D}_i \in \real^{n_i \times d}$ is the
regression matrix and $\mathbf{b}_i \in \real^{n_i}$ is the label vector.

A typical challenge arising in regression problems is due to the fact that
problem~\eqref{eq:LS} may be ill-posed and can easily lead to
over-fitting phenomena.
A viable technique to prevent over-fitting consists in adding a suitable
\emph{regularization} term $r(\bx)$ in the cost function, leading to
\begin{align*}
  \min_{\bx} \:\: \smallsum_{i=1}^N\| \mathbf{D}_i \bx - \mathbf{b}_i \|^2 + r(\bx),
\end{align*}
where $\map{r}{\real^d}{\real}$ is assumed to be known by all the agents
in the network.
Several possibilities for the regularizer $r(\bx)$ can be chosen. For instance, by
using $\ell_1$-norm, we obtain the so-called LASSO (Least Absolute Shrinkage and 
Selection Operator) problem, i.e.,
\begin{align}
  \label{eq:LS_regularized}
  \min_{\bx} \:\: \smallsum_{i=1}^N \| \mathbf{D}_i \bx - \mathbf{b}_i \|^2 + \rho \|\bx\|_1
\end{align}
where $\rho$ is a positive scalar used to strengthen or weaken the 
effects of the regularizer.
Problem~\eqref{eq:LS_regularized} can be classified as cost-coupled, i.e.,
of the form~\eqref{setups:cost-coupled_problem}, with $X = \real^d$ and
local functions given by
$f_i(\bx) = \| \mathbf{D}_i \bx - \mathbf{b}_i \|^2 + \rho/N \cdot \|\bx\|_1$.

This problem will be used to test duality-based methods for cost-coupled problems
and a numerical example is shown in Section~\ref{sec:dual_simulations}.

\subsection{Classification via Logistic Regression}
\label{sec:setups_logistic}
Regression problems can be also set up for a classification scenario. We recall a
set-up in which linear models are trained by minimizing the so-called \emph{logistic
loss functions}.
Suppose each agent has $m_i$ points $p_{i,1}, \ldots, p_{i,m_i} \in \real^d$
(which represent training samples in a feature space) and  suppose they are
associated to binary labels, i.e., each point $p_{i,j}$ is labeled with
$\ell_{i,j} \in \{-1,1\}$, for all $j \in \until{m_i}$ and $i \in \until{N}$.
The problem consists of building a linear classification model from the training
samples by maximizing the a-posteriori probability of each class. 
In particular, we look for a separating hyperplane of the form
$\{ z \in \real^d \mid w^\top z + b = 0 \}$, whose parameters
($w$ and $b$) can be determined by solving the convex optimization problem
\begin{align}
\begin{split}
  \min_{w, b} \:  & \: \smallsum_{i=1}^N \: \smallsum_{j=1}^{m_i}
    \log \! \left[ 1 + e^{-(w^\top p_{i,j} + b) \ell_{i,j}} \right] + \dfrac{C}{2} \|w\|^2,
\end{split}
\label{eq:logistic_regression_problem}
\end{align}
where $C > 0$ is a parameter affecting regularization. We now make some
observations on problem~\eqref{eq:logistic_regression_problem}.
First, we see that it is an unconstrained optimization problem, so that an optimal
solution can always be found (even though it may be meaningless for the
classification problem). Second, we point out that the cost function is strictly
convex, so that the optimal solution is unique. Finally, notice that
the problem is cost-coupled, i.e., it is of the form~\eqref{setups:cost-coupled_problem},
with $X = \real^d$ and each $f_i$ is given by
\begin{align*}
  f_i(w, b) = \smallsum_{j=1}^{m_i} 
    \log \left[ 1 + e^{-(w^\top p_{i,j} + b) \ell_{i,j}} \right] + \dfrac{C}{2N} \|w\|^2,
  \hspace{0.5cm}
  i \in \until{N}.
\end{align*}

In a distributed setting, the goal is to make agents agree on a common
solution $(w^\star, b^\star)$, so that all of them can compute the separating
hyperplane as $\{ z \in \real^d \mid (w^\star)^\top z + b^\star = 0 \}$.

This problem is suited for the application of consensus-based primal
methods (cf. Section~\ref{chap:primal}) and a numerical example is
shown in Section~\ref{sec:primal:simulations}.

\subsection{Classification via Support Vector Machine (SVM)}
\label{sec:setups_SVM}
Support Vector Machines (SVMs) are a popular tool used in (supervised) learning
to build classification models.
Suppose we have $N$ points $p_1, \ldots, p_N \in \real^d$ (which represent training
samples in a feature space) and suppose they are associated to binary labels, i.e.,
each $p_i$ is labeled with $\ell_i \in \{-1,1\}$, for all $i \in \until{N}$.
For simplicity, we consider linear SVM (more complex
set-ups can be handled with appropriate transformations \cite{boser1992training}).
The problem consists of building a classification model from the training samples.
In particular, we look for a separating hyperplane of the form
$\{ z \in \real^d \mid w^\top z + b = 0 \}$ such that it separates all the points
with $\ell_i = -1$ from all the points with $\ell_i = 1$. In symbols:
\begin{align*}
  w^\top p_i + b &> 0, \hspace{0.5cm} \forall i \text{ such that } \ell_i = 1, \text{ and}
  \\
  w^\top p_i + b &< 0, \hspace{0.5cm} \forall i \text{ such that } \ell_i = -1.
\end{align*}
In Figure~\ref{fig:setups_SVM}, a classification example is shown.

\begin{figure}[!htbp]
  \centering
  \includegraphics[scale=1]{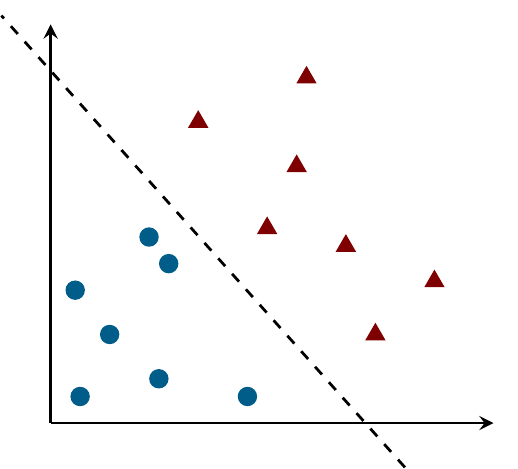}
  \caption{
    Graphical representation of a linear SVM problem in $\real^2$. The triangles and
    the dots represent points with different labels and the goal is to compute a
    separating hyperplane, denoted here by a dashed line.
  }
  \label{fig:setups_SVM}
\end{figure}

In order to maximize the distance of the separating
hyperplane from the training points, one can solve the following
(convex) quadratic program:
\begin{align}
\begin{split}
  \min_{w, b} \:  & \: \frac{1}{2} w^\top w
  \\
  \subj \: & \: \ell_i ( w^\top p_i + b ) \ge 1, \hspace{0.5cm} i \in\until{N}.
\end{split}
\label{eq:svm_problem}
\end{align}
Problem~\eqref{eq:svm_problem} is known in the literature as \emph{hard-margin} SVM problem,
and can be solved only if a separating hyperplane exists.
However, if problem~\eqref{eq:svm_problem} is infeasible (e.g., when there are outliers),
one can solve a \emph{soft-margin} SVM problem in which some of the
training samples are allowed to be on the ``wrong side'' of the hyperplane.
Formally, we consider the following relaxation of problem~\eqref{eq:svm_problem}:
\begin{align}
\begin{split}
  \min_{w, b, \bxi} \:  & \: \frac{1}{2} w^\top w + C \smallsum_{i=1}^N \xi_i
  \\
  \subj \: & \: \ell_i ( w^\top p_i + b ) \ge 1 - \xi_i, \hspace{0.5cm} i \in\until{N},
  \\
  & \: \bxi \ge 0,
\end{split}
\label{eq:soft_svm_problem}
\end{align}
where we denote by $\bxi$ the vector stacking the violations $\xi_1, \ldots, \xi_N$
and $C > 0$ weighs the effect of the relaxation.
Notice that problem~\eqref{eq:soft_svm_problem} can be viewed either as a
cost-coupled problem of the form~\eqref{setups:cost-coupled_problem},
or as a common cost problem of the form~\eqref{setups:common_cost_problem}.

In a distributed set-up, problem~\eqref{eq:soft_svm_problem} must be solved by
agents in a network. We suppose that each agent $i$ is assigned
exactly one training tuple $(p_i,\ell_i)$, so that each agent
knows one constraint of the optimization problem.
Agents eventually agree on an optimal solution
$(w^\star, b^\star, \bxi^\star)$,
so that the separating hyperplane can be computed as
$\{ z \in \real^d \mid (w^\star)^\top z + b^\star = 0 \}$.

This problem is suited, e.g., for the application of constraint exchange methods
(cf. Section~\ref{ce:sec:CP}) and a numerical example is shown
in Section~\ref{ce:sec:simulations}.

\subsection{Target Localization in Sensor Networks}
An interesting application in the field of sensor and robotic networks is the
problem of estimating the position of a target, while having information on the
position of sensors that can detect the unknown target within their field
of sensing. A representational example of the problem is given in
Figure~\ref{fig:setups_target_localization}.
\begin{figure}[!htbp]
  \centering
  \includegraphics[scale=1]{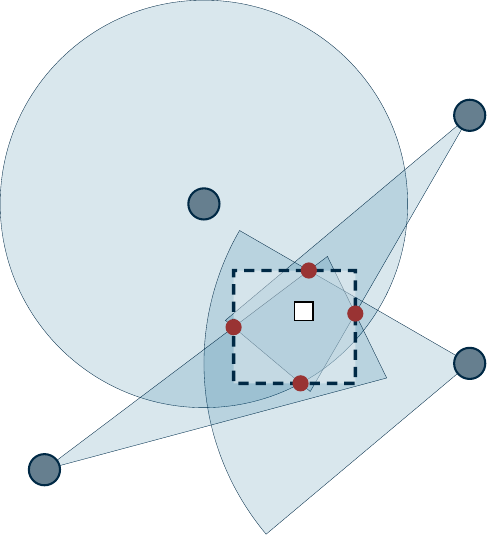}
  \caption{Localization of the target (white square) by set estimates of four sensors (blue nodes).
      The target is confined in the bounding box (dashed rectangle), which is determined
      by four extreme points (in red).
      The extreme points can be found by solving four instances of
      problem~\eqref{eq:convex_position_estimation} with different cost vector $c$.
  }
  \label{fig:setups_target_localization}
\end{figure}
Formally, we suppose that $N$ sensors are used to estimate in a distributed way
the position of a target. Each sensor $i$ knows
its position $\bv_i \in \real^2$ and the unknown target position
is denoted by $\bx \in \real^2$.
We assume that sensors in the network detect the presence of the unknown
target with two sensing mechanisms: \emph{(i)} laser transmitters
which scan through some angle, leading to a bounded cone set that can be
expressed by three linear constraints, two bounding the angle and one
bounding the distance, compactly written as $A_i \bx \le b_i$, with 
$A_i \in \real^{3 \times 2}$ and $b_i \in \real^3$,
and \emph{(ii)} the range of the RF transmitter, leading to circular
constraints of the form $\| \bx - \bv_i \|_2 \le r_i$, where
$r_i$ denotes the maximum sensing distance.
Depending on the sensing mechanisms that each sensor $i$ is equipped with,
it is possible to bound the position of the unknown target to be contained in
the intersection of convex sets $X_i$, each one known only by agent $i$, defined
as $X_i \triangleq \{ \bx \mid \| \bx - \bv_i \|_2 \le r_i \}$ if the constraint is a disk,
$X_i \triangleq \{ \bx \mid A_i \bx \le b_i \}$ if the constraint is a cone,
$X_i \triangleq \{ \bx \mid A_i \bx \le b_i,  \: \| \bx - \bv_i \|_2 \le r_i\}$
if the constraint is a quadrant.

Now, the goal for the agents is to compute the smallest bounding box
$\{ \bx \in \real^2 \mid \bx^L \le \bx \le \bx^U \}$,
for suitable $\bx^L, \bx^U \in  \real^2$, that is guaranteed to contain
the unknown position of the additional target. This can be addressed
by solving four optimization problems, one for each component of $\bx^L, \bx^U$.
For instance, to compute the first component of $\bx^L$, agents define the objective
vector $c = [1, 0]^\top$ and they cooperatively solve the optimization
problem
\begin{align}
\begin{split}
  \min_\bx \: & \: c^\top \bx
  \\
  \subj \: & \: \bx \in \bigcap_{i=1}^N X_i,
\end{split}
\label{eq:convex_position_estimation}
\end{align}
which is in the common cost form~\eqref{setups:common_cost_problem}.
After an optimal solution $\bx^\star$ is found, each agent computes
the first component of $\bx^L$ by using the first component of $\bx^\star$,
and similarly for the other coordinates.

\subsection{Task allocation/assignment}
\label{sec:task_assignment}
Task allocation is a building block for decision making problems in which
a certain number of agents must be assigned given tasks. The goal is to find the
best matching of agents and tasks according to a given performance criterion.
Here, we consider $N$ agents and $N$ tasks and we look for a one-to-one assignment.
Define the variable $x_{i\kappa}$, which is $1$ if agent $i$ is assigned to
task $\kappa$ and $0$ otherwise. Also, define the set $E_A$, which contains
the tuple $(i,\kappa)$ if agent $i$ can be assigned to task $\kappa$. Finally,
let $c_{i\kappa}$ be the cost occurring if agent $i$ is
assigned to task $\kappa$.  In Figure~\ref{fig:setups_task_assignment}, we show
an illustrative example of the set-up.
\begin{figure}[!htbp]
  \centering
  \includegraphics[scale=1]{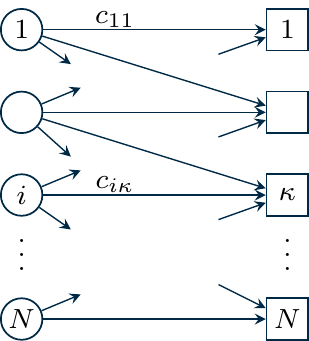}
  \caption{
    Graphical representation of the task assignment problem. Agents are represented by
    circles, while tasks are represented by squares.
    An arrow from agent $i$ to task $\kappa$ means that agent $i$ can perform task
    $\kappa$ (i.e., $(i,\kappa) \in E_A$) with corresponding cost equal to $c_{i\kappa}$.
  }
  \label{fig:setups_task_assignment}
\end{figure}

Since the objective is to minimize the total cost, the task allocation problem can be formulated
as an integer program. However, as pointed out in \cite{bertsekas1998network},
integrality constraints can be dropped to obtain the linear program
\begin{align}
\begin{split}
  \min_\bx \: & \: \smallsum_{(i,\kappa) \in E_A} c_{i\kappa} x_{i\kappa}
  \\
  \subj \: & \: \0 \leq \bx \leq \1,
  \\
  & \: \smallsum_{\{\kappa \mid (i,\kappa) \in E_A\}} x_{i\kappa} = 1 \hspace{0.5cm} \forall \: i \in \until{N},
  \\
  & \: \smallsum_{\{i \mid (i,\kappa) \in E_A\}} x_{i\kappa} = 1 \hspace{0.5cm} \forall \: \kappa \in \until{N},
\end{split}
\label{eq:assignment_problem}
\end{align}
where $\bx$ is the variable stacking all $x_{i\kappa}$.
If problem~\eqref{eq:assignment_problem} is feasible, it can be shown that
it admits an optimal solution such that $x_{i\kappa} \in \{0,1\}$ for all $(i,\kappa) \in E_A$
(see, e.g., \cite{bertsekas1998network}).
Moreover, all the optimal assignments belong to the optimal solution set of
problem~\eqref{eq:assignment_problem}.

Problem~\eqref{eq:assignment_problem} can be cast to the constraint-coupled
form~\eqref{setups:constraint-coupled_problem}.
To see this, let us define $K_i$ as the number of tasks that agent
$i$ can perform (i.e., $K_i = |\{\kappa \mid (i,\kappa) \in E_A\}|$).
We assume that agent $i$ deals with the variable $\bx_i \in \real^{K_i}$, stacking the $x_{i\kappa}$
for all $\kappa$ such that $(i,\kappa) \in E_A$. Then, the local sets $X_i$ can be written as
\begin{align}
  X_i = \{ \bx_i \in \real^{K_i} \mid \0 \leq \bx_i \leq \1 \text{ and } \bx_i^\top \1 = 1 \},
  \hspace{0.5cm}
  i \in \until{N}.
\end{align}
The coupling constraints can be written by defining, for all $i \in \until{N}$,
the matrix $H_i \in \real^{N \times K_i}$, obtained by extracting from the $N \times N$
identity matrix the subset of columns corresponding to the tasks that agent
$i$ can perform.
Problem~\eqref{eq:assignment_problem} becomes
\begin{align*}
\begin{split}
  \min_{\bx_1, \ldots, \bx_N} \: & \: \smallsum_{i=1}^N c_i^\top \bx_i
  \\
  \subj \: & \: \bx_i \in X_i, \hspace{0.5cm} i \in \until{N}
  \\
  & \: \smallsum_{i=1}^N H_i \bx_i = \1,
\end{split}
\end{align*}
where each $c_i$ stacks the costs $c_{i\kappa}$, for all $\kappa$ such that $(i,\kappa) \in E_A$.
Notice that problem~\eqref{eq:assignment_problem} can be also tackled by
resorting to its dual, which can be solved by using distributed optimization
algorithms for common-cost problems.

In a distributed context, the goal for the agents is to find an optimal solution
$\bx^\star$, but each agent $i$ is only interested in computing its portion $\bx_i^\star$
of optimal solution, which contains only one entry $x_{i\kappa} = 1$,
corresponding to the task $\kappa$ that agent $i$ is eventually assigned.

\subsection{Cooperative Distributed Model Predictive Control}
\label{sec:MPC}

Model Predictive Control (MPC) is a widely studied technique in the control
community, and is also used in distributed contexts.
The goal is to design an optimization-based feedback control law for a
(spatially distributed) network of dynamical systems.
The leading idea is the principle of \emph{receding horizon} control, which informally
speaking consists of solving at each time step an optimization problem (usually
termed \emph{optimal control problem}), in which the system model is used to
predict the system trajectory over a fixed time window.
After an optimal solution of the optimal control problem is found,
the input associated to the current time instant is applied
and the process is repeated (for a survey on MPC methods, see, e.g., \cite{rawlings2009model}).

Now, we describe a typical distributed MPC framework applied to a network of linear
systems with linear coupling constraints.
Formally, assume we have $N$ discrete-time linear dynamical systems with
independent dynamics of the form $\bz_i(s+1) = A_i \bz_i(s) + B_i \bu_i(s)$, where
$s \in \integer$ represents time; $\bz_i(s) \in \real^{q_i}$ is the system state at time $s$;
$\bu_i(s) \in \real^{m_i}$ is the input fed to the system at time $s$;
and $A_i, B_i$ are given matrices of appropriate dimensions, for all $i \in \until{N}$.
We suppose that the states and the inputs must satisfy local constraints $\bz_i(s) \in Z_i$
and $\bu_i(s) \in U_i$
for all $i \in \until{N}$, and that the agents' states are coupled to each other
by means of coupling constraints of the form $\sum_{i=1}^N H_i \bz_i(s) \leq h$,
for a given $h \in \real^P$.
Given the initial conditions of the systems $\bz_1^0, \ldots, \bz_N^0$,
the optimal control problem to be solved is
\begin{align}
\begin{split}
  \min_{\substack{\bz_1, \ldots, \bz_N \\ \bu_1, \ldots, \bu_N}} \:
    & \: \smallsum_{i=1}^N \bigg( \smallsum_{s=0}^{S-1} \ell_i( \bz_i(s), \bu_i(s) ) + V_i (\bz_i(S)) \bigg)
  \\
  \subj \: & \: \bz_i(s + 1) = A_i \bz_i(s) + B_i \bu_i(s), \: s \in \fromto{0}{S-1}, \forall \: i,
  \\
  & \: \bz_i(s) \in Z_i, \bu_i(s-1) \in U_i \hspace{1cm} s \in \until{S}, \hspace{0.65cm} \forall \: i,
  \\
  & \: \bz_i(0) = \bz_i^0, \hspace{6.47cm} \forall \: i,
  \\
  & \: \smallsum_{i=1}^N H_i \bz_i(s) \leq h, \hspace{2.45cm} s \in \until{S},
\end{split}
\label{eq:mpc_problem}
\end{align}
where $S$ is the \emph{prediction horizon}, $\bz_i = [ \bz_i(0)^\top, \ldots, \bz_i(S)^\top ]^\top$
and $\bu_i = [ \bu_i(0)^\top, \ldots, \bu_i(S-1)^\top ]^\top$ are the optimization
vectors,  $\map{\ell_i}{\real^{q_i + m_i}}{\real}$ is the \emph{stage cost} and
$\map{V_i}{\real^{q_i}}{\real}$ is the \emph{terminal cost}, for all $i \in \until{N}$.
Problem~\eqref{eq:mpc_problem} can be fit into the constraint-coupled
set-up~\eqref{setups:constraint-coupled_problem} by setting
\begin{align*}
  f_i(\bx_i) &= \smallsum_{s=0}^{S-1} \ell_i( \bz_i(s), \bu_i(s) ) + V_i (\bz_i(S)),
  \\
  \bg_i(\bx_i)
  &=
  \begin{bmatrix}
    H_i \bz_i(1) - \frac{h}{N}
    \\
    \vdots
    \\
    H_i \bz_i(S) - \frac{h}{N}
  \end{bmatrix},
\end{align*}
for all $i \in \until{N}$,
with the local optimization variables being $\bx_i = \big[\bz_i^\top, \bu_i^\top\big]^\top$
and the local constraint set $X_i$ being
\begin{align*}
  X_i
  \triangleq
  \Big\{
    (\bz_i, \bu_i) \in \real^{(S+1) q_i + Sm_i}
    \mid \:
    &\bz_i(s + 1) = A_i \bz_i(s) + B_i \bu_i(s),
  \\
    &\bz_i(s+1) \in Z_i,
    \bu_i(s) \in U_i,
    \:\: \forall \: s %
  \Big\},
\end{align*}
for all $i \in \until{N}$.
Next, we describe an example of microgrid control scenario that can
be fit into our distributed optimization framework.
A microgrid consists of several generators, controllable loads, storage devices
and a connection to the main grid. In the following, we use the notational convention
that energy generation corresponds to positive variables, while energy consumption
corresponds to negative variables. 
Generators are collected in the set $\GEN$. At each time instant $s$ in a
given horizon $[0,S]$, they generate power, denoted by $p_{\gen,i}^{s}$, that
must satisfy magnitude and rate bounds, i.e., for given positive scalars
$\ubar{p}$, $\bar{p}$, $\ubar{r}$ and $\bar{r}$, it must hold, for all
$i \in \GEN$, $\ubar{p} \le p_{\gen,i}^{s} \le \bar{p}$, with
$s \in [0,S]$, and
$\ubar{r} \le p_{\gen,i}^{s+1} - p_{\gen,i}^{s} \le \bar{r}$, with
$s \in [0,S-1]$.
The cost to produce power by a generator is modeled as a quadratic function 
$f_{ \gen,i }^s = \alpha_1 p_{\gen ,i}^s + \alpha_2 (p_{\gen ,i}^s)^2$ with 
$\alpha_1$ and $\alpha_2$ positive scalars.
Storage devices are collected in $\STOR$ and their power is denoted by
$p_{\stor,i}^{s}$ and satisfies bounds and a dynamical constraint given by
$-d_{\stor} \le p_{\stor,i}^{s} \le c_{\stor}$, $s \in [0,S]$,
$ q_{\stor,i}^{s+1} = q_{\stor,i}^{s} + p_{\stor ,i}^{s}$,
$s \in [0,S-1]$, and $0 \le q_{\stor,i}^{s} \le q_{\text{max}}$,
$s \in [0,S]$,
where the initial capacity $q_{\stor ,i}^{0}$ is given and 
$d_{\stor }$, $c_{\stor }$ and $q_{\text{max}}$ are positive scalars.
There are no costs associated with the stored power.
Controllable loads are collected in $\CLOAD$ and their power is denoted by
$p_{\cload,i}^{s}$. The power must satisfy box constraints, i.e.,
$-P \le p^s_{\cload,i} \le P, \:\:\: s \in [0,S]$.
A desired load profile $p_{\des,i}^s$ for $p_{\cload,i}^s$ is given and 
the controllable load incurs in a cost
$f_{ \cload,i }^s =\beta \max \{0, p_{\des,i}^s  - p_{\cload,i}^s \}$, $\beta \ge0$,
if the desired load is not satisfied.
Finally, the device $i=N$ is the connection node with the main grid; its power
is denoted as $p_{\trade}^{s}$ and must satisfy
$| p_{\trade }^{s} | \le E$, $s \in [0,S]$.
The power-trading cost is modeled as $f_{\trade}^s = -c_1 p_{\trade }^{s} + c_2 |p_{\trade }^{s} | $, 
with $c_1$ and $c_2$ positive scalars corresponding to the price and
a general transaction cost respectively.

The power network must provide at least a given power demand $D^s$,
which can be modeled by a \emph{coupling constraint} among the units
\begin{align}
  \smallsum_{i\in \GEN} p_{\gen,i}^{s} 
  +
  \smallsum_{i\in \STOR} p_{\stor,i}^{s}
  +
  \smallsum_{i\in \CLOAD} p_{\cload,i}^{s}
  +
  p_{\trade}^{s}
  \ge D^s,
\label{eq:microgrid_coupling}
\end{align}
for all $s\in[0,S]$. Reasonably, we assume $D^s$ to be known only by the
connection node $\trade$.

Notice that the microgrid control problem can be cast in the
constraint-coupled form~\eqref{setups:constraint-coupled_problem}.
To this end, we let each $\bx_{i}$ be the whole trajectory over the
prediction horizon $[0,S]$, i.e.,
\begin{align*}
  \bx_{i}
  \triangleq
  [p_{\gen,i}^{0},\ldots, p_{\gen,i}^{S}]^\top,
\end{align*}
for all the generators $i \in \GEN$ and, consistently, for the other device types.
As for the cost functions, we define
\begin{align*}
  f_i(\bx_{i}) \triangleq \smallsum_{s=0}^S f_{\gen, i}^s (p_{\gen,i}^s)
\end{align*}
for all the generators $i \in \GEN$ and, consistently, for the other device types.
The local constraint sets $X_i$ are given by
\begin{align*}
  X_i
  \triangleq
  \Big\{
    [p_{\gen,i}^{0},\ldots, p_{\gen,i}^{S}]^\top
    \mid \:\:
    & \ubar{p} \le p^s_{\gen,i} \le \bar{p}, \: \tau \in [0,S],
  \\
    & \ubar{r} \le p^{s+1}_{\gen,i} - p^s_{\gen,i} \le \bar{r}, \:\:\: \tau \in [0,S-1]
  \Big\},
\end{align*}
for all the generators $i \in \GEN$ and, consistently, for the other device types.
The coupling constraints are as in~\eqref{eq:microgrid_coupling}.

This problem is suited, e.g., for the application of duality-based methods
for constraint-coupled problems (cf. Section~\ref{sec:dual_constraint_coupled})
and a numerical example is shown in Section~\ref{sec:dual_simulations}.

\fi

\iftrue

\chapter{Consensus-Based Primal Methods}
\label{chap:primal}

In this chapter we focus on primal approaches to design distributed algorithms
for cost-coupled problems. We start by describing the so-called distributed
subgradient method, based on a combination of the average consensus protocol
with the subgradient method. Then, we present a recent improvement of such
consensus-based scheme, named gradient tracking, that relies on the novel
idea of tracking the gradient of the global cost function via a dynamic
consensus scheme. Then, we provide extensions to the basic schemes.
Finally, we show a numerical example to compare the two
presented algorithms.

\section{Distributed Subgradient Method}
\label{sec:distributed_gradient}

In this section we review the distributed subgradient method that has been 
proposed in the pioneering works~\cite{nedic2009distributed,nedic2010constrained}
(see also the tutorial papers~\cite{nedic2015convergence,nedic2018distributed,nedic2018network}). 
In this survey, we report a proof based on
the analysis proposed in the references above.

As already described in Section~\ref{sec:setups_cost_coupled}, we consider
a network of $N$ agents that aim to cooperatively solve the cost-coupled
problem
\begin{align}
\begin{split}
  \min_{\bx } \: & \: \smallsum_{i=1}^N f_i ( \bx)
  \\
  \subj \: & \: \bx \in \real^d,
\end{split}
\label{primal:cost_coupled_problem}
\end{align}
where each cost function $\map{f_i}{\real^d}{\real}$ is known by agent
$i$ only, for all $i\in\until{N}$.

A natural way to devise a distributed algorithm for problem~\eqref{primal:cost_coupled_problem}
is to study how it would be solved through a centralized gradient-based approach.
We recall that a subgradient method applied to~\eqref{primal:cost_coupled_problem} 
consists of an iterative procedure in which the current solution estimate, 
denoted by $\bx^t$, is updated according to
\begin{align*}
  \bx^{t+1} 
  &
  =
  \bx^t - \gamma^t \smallsum_{i=1}^N \subgrad f_i (\bx^t),
\end{align*}
where $\gamma^t$ is the step-size and $\sum_{i=1}^N \subgrad f_i (\bx^t)$
is a subgradient of the cost function at $\bx^t$.
The initial value $\bx^0$ can be set to any element of $\real^d$.

Next, we introduce the distributed subgradient algorithm 
proposed in~\cite{nedic2009distributed,nedic2010constrained}.
Each agent $i$ maintains its own estimate $\bx_i^t$ of the decision variable $\bx$,
initialized to any value in $\real^d$ and iteratively updated until it
eventually converges to an optimal solution of~\eqref{primal:cost_coupled_problem}.
The distributed subgradient algorithm is based on the combination of a
consensus protocol (cf. Appendix~\ref{sec:consensus_appendix}) with the
subgradient optimization method (cf. Appendix~\ref{app:gradient_method}) to
move each local solution estimate toward an optimal (common) solution of
problem~\eqref{primal:cost_coupled_problem}.
Algorithm~\ref{alg:distributed_subgradient} summarizes the distributed 
subgradient algorithm from the perspective of node $i$.
\begin{algorithm}%
  \begin{algorithmic}[0]
    \Statex \textbf{Initialization}: $\bx_i^0 \in \real^d$
    \smallskip
    
    \Statex \textbf{Evolution}: for $t=0,1,... $
    \smallskip
    
      \StatexIndent[0.5] 
      \textbf{Gather} $\bx_j^t$ from neighbors $j \in \nbrs_i$ 
      
      \StatexIndent[0.5] 
      \textbf{Compute}
      \begin{subequations}
      \label{primal:subg_algorithm}
      \begin{align}
      \label{primal:subg_alg_v_i}
        \bv_i^{t+1} = \smallsum_{j\in\nbrs_i} a_{ij} \, \bx_j^t
      \end{align}

      \StatexIndent[0.5] 
      \textbf{Update} 
      \begin{align}
        \label{primal:subg_alg_update}
			  \bx_i^{t+1} 
			  =
			  \bv_i^{t+1}
        - \gamma^t \, 
			  \widetilde{\nabla} f_i ( \bv_i^{t+1} )
      \end{align}
      \end{subequations}

  \end{algorithmic}
  \caption{Distributed Subgradient}
  \label{alg:distributed_subgradient}
\end{algorithm}

For presentation purposes, in this survey we consider a simplified network configuration,
so that the core idea of the scheme can be easily caught. That is, the network is modeled as 
a fixed, connected and undirected graph $\GG = (\until{N},\EE)$. %
The weights $a_{ij}$ in~\eqref{primal:subg_alg_v_i} inherit the typical assumptions
on consensus protocols, formally reported next.
\begin{assumption}
\label{primal:assumption_network_subg}
	Let the weights $a_{ij}$, $i,j \in \until{N}$ be nonnegative entries of
	$A \in \real^{N \times N}$ that match the graph $\GG$, 
	i.e., $a_{ij} \neq 0$ for all $(i,j) \in \EE$ and $a_{ij} = 0$ otherwise.
	Moreover, they satisfy
	\begin{itemize}
	  \item $\sum_{j=1}^N a_{ij} = 1$, for all $i\in\until{N}$;
	  \item $\sum_{i=1}^N a_{ij} = 1$, for all $j\in\until{N}$;
	  \item for all $i\in\until{N}$, $a_{ii} > 0$.\oprocend
  \end{itemize}
\end{assumption}

We point out that one may also consider strongly connected directed graphs
that admits a doubly-stochastic weighted adjacency matrix.
Detailed convergence analyses of distributed subgradient schemes have been
provided, e.g., in~\cite{nedic2009distributed,nedic2010constrained,
  nedic2015convergence,nedic2018network,nedic2018distributed}.  For the sake of
completeness, this survey provides a proof for the convergence of
Algorithm~\ref{alg:distributed_subgradient} that is mainly inspired by the
references above.

We start by stating the condition on the step-size $\gamma^t$ used in
the update~\eqref{primal:subg_alg_update}. 
As in the standard (centralized) subgradient method, it must satisfy a diminishing property.
\begin{assumption}
\label{primal:stepsize_diminishing}
  The step-size sequence $\{ \gamma^t \}_{t\ge0}$, with $\gamma^t \ge 0$, satisfies the
  conditions $\sum_{t=0}^{\infty} \gamma^t = \infty$,
  $\sum_{t=0}^{\infty} \big( \gamma^t \big)^2 < \infty$.
  \oprocend
\end{assumption}
As a consequence of the square summability in Assumption~\ref{primal:stepsize_diminishing}, 
the step-size vanishes as the algorithm proceeds, i.e., $\lim_{t \to \infty} \gamma^t = 0$.

Next, we state regularity requirements for problem~\eqref{primal:cost_coupled_problem}.
\begin{assumption}
\label{primal:regularity_assumption}
  Let the following conditions hold:
  \begin{itemize}
    \item[\it (i)] each $\map{f_i}{\real^d}{\real}$ is convex and has bounded subgradients,
    i.e., there exists a scalar $C_i > 0$ such that
    $\| \subgrad f_i(\bx) \| \le C_i$ for any subgradient 
    $\subgrad f_i(\bx)$ of $f_i$ at any $\bx \in \real^d$;  %

    \item[\it (ii)] problem~\eqref{primal:cost_coupled_problem} has at least one optimal
    solution, i.e., the optimal solution set 
    $X^\star = \{ \bx \in \real^d \mid  f(\bx) = f^\star \}$
    is nonempty, where $f^\star$ denotes the optimal value of
    problem~\eqref{primal:cost_coupled_problem}.~\oprocend
  \end{itemize}
\end{assumption}

Usually, in the analysis of consensus-based algorithms, it is useful to introduce
the average of the quantities that are required to be asymptotically consensual.
Here, we introduce the average of the current solution estimates, i.e., for all
$t\ge 0$ we define
\begin{align}
  \label{primal:distr_subgr_avg_definition}
  \avgx^t \triangleq \frac{1}{N} \smallsum_{i=1}^N \bx_i^t.
\end{align}
We point out that $\avgx^t \in \real^d$ has the same dimension of the local solution
estimates $\bx_i^t$ and is introduced only for the sake of analysis. Of course,
it cannot be computed by any agent and, nevertheless, it does not need to be known.
We observe that $\avgx^t$ evolves according to its own dynamics, which can be 
obtained by combining the local updates of the agents. Formally, it holds
\begin{align}
\label{primal:distr_subgr_avg_evolution}
\begin{split}
  \avgx^{t+1} 
  & = 
  \dfrac{1}{N} \smallsum_{i=1}^N \bx_i^{t+1}
  =
  \dfrac{1}{N} \smallsum_{i=1}^N
  \left (
  \bv_i^{t+1}
  -
  \subgrad f_i ( \bv_i^{t+1} )
  \right )
  \\
  & 
  =
  \dfrac{1}{N} \smallsum_{i=1}^N
  \smallsum_{j=1}^N
  a_{ij} \, \bx_j^t
  -
  \gamma^t
  \dfrac{1}{N} \smallsum_{i=1}^N
  \subgrad f_i ( \bv_i^{t+1} )
  \\
  & 
  = 
  \avgx^t
  -
  \gamma^t
  \dfrac{1}{N} \smallsum_{i=1}^N
  \subgrad f_i ( \bv_i^{t+1} ),
  \end{split}
\end{align}
where we used the (row) stochasticity of the weights $a_{ij}$.

The following result (see~\cite{bertsekas2000gradient} for a proof) is an
important building block for the forthcoming proof of the convergence of
Algorithm~\ref{alg:distributed_subgradient}. %
\begin{lemma}
\label{lem:supermartingale}
  Let $\{ Y^t\}_{t\ge0}$, $\{ W^t\}_{t\ge0}$, and $\{Z^t\}_{t\ge0}$ be three 
  scalar sequences such that $W^tt$ is nonnegative for all $t$. 
  Assume the following
  \begin{align*}
    & Y^{t+1} \le Y^t - W^t +Z^t, \qquad t \ge 0,
    \\
    & \smallsum_{t=0}^\infty Z^t < \infty.
  \end{align*}
  Then either $\lim_{t\to\infty} Y^t = -\infty$ or else $\{ Y^t \}_{t\ge0}$ converges to a finite 
  value and $\sum_{t=0}^\infty W^t < \infty$.\oprocend
\end{lemma}

The following theorem, also provided, e.g., in~\cite{nedic2009distributed,nedic2010constrained,
nedic2015convergence,nedic2018distributed,nedic2018network},
formally states the convergence properties of
Algorithm~\ref{alg:distributed_subgradient}.
For ease of notation, we consider a scalar optimization problem, i.e.,
$d = 1$.

\begin{theorem} %
  Let Assumptions~\ref{primal:assumption_network_subg},~\ref{primal:stepsize_diminishing}
  and~\ref{primal:regularity_assumption} hold and let the communication graph
  be undirected and connected.
	Then, the sequences of local solution estimates $\{ \bx_i^t\}_{t\ge0}$, $i \in\until{N}$, 
	generated by 	Algorithm~\ref{alg:distributed_subgradient}, converge to a 
	(common) solution of problem~\eqref{primal:cost_coupled_problem}, i.e., 
  for all $i\in\until{N}$, it holds
	\begin{align}
	  \lim_{t \to \infty} \| \bx_i^t - \bx^\star \| = 0,
	\end{align}
 for some $\bx^\star \in X^\star$.
\end{theorem}
\begin{proof} 
The proof provided in this manuscript is mainly based on the ones
given in \cite{nedic2009distributed,nedic2010constrained,
nedic2015convergence,nedic2018distributed,nedic2018network}.
It is based on showing the following three steps:
\begin{enumerate}
  \item asymptotic consensus of the local solution estimates to their average, i.e., 
  \begin{align}
    \lim_{t \to \infty} \| \bx_i^t - \avgx^t \| = 0,
  \end{align}
  for all $i\in\until{N}$;

  \item summability of the consensus error weighted by the step-size, i.e.,
  \begin{align}
    \lim_{T \to \infty} \smallsum_{t=0}^T \gamma^t \| \bx_i^t - \avgx^t \| < \infty;
  \end{align}

  \item convergence of the average sequence $\{ \avgx^t  \}_{t\ge 0}$ to an 
  optimal solution of problem~\eqref{primal:cost_coupled_problem}, i.e.,
  \begin{align}
    \lim_{t \to \infty} \| \avgx^t - \bx^\star \| = 0,
  \end{align}
  for some $\bx^\star \in X^\star$.
\end{enumerate}

Let $\bx^t$ be the vector stacking the local solution estimates $\bx_i^t$.
Then, the consensus error evolution is given by
  \begin{align*}
     \bx^{t+1} - \avgx^{t+1} \1
     & 
     =
     A \bx^t + \bepsilon^t 
     - \avgx^t \1 - \frac{1}{N} \smallsum_{j=1}^N \bepsilon_j^t
     \\
     & 
     =
     (A - \1\1^\top /N) \bx^t
     +
     \bepsilon^t - \frac{1}{N} \smallsum_{j=1}^N \bepsilon_j^t
     \\
     & 
     =
     (A - \1\1^\top /N) (\bx^t - \avgx^t \1)
     +
     \bepsilon^t - \frac{1}{N} \smallsum_{j=1}^N \bepsilon_j^t
     \\
     & 
     =
     (A - \1\1^\top /N) (\bx^t - \avgx^t \1)
     +
     (I - \1\1^\top /N) \bepsilon^t,
  \end{align*}
  where $\bepsilon^t$ denotes the vector stacking all the $\bepsilon_i^t$ with the short-hand 
  $\bepsilon_i^t$ for $-\gamma^t \subgrad f_i ( \bv_i^{t+1} )$, for all $i\in\until{N}$.
  
  Taking the norm of both sides in the last equation and applying the triangle inequality
  leads to
  \begin{align*}
    \| \bx^{t+1} - \avgx^{t+1} \1 \|
    & 
    \le
    \| A - \1\1^\top /N\| \| \bx^t - \avgx^t \1\|
     +
     \| I - \1\1^\top /N\| \| \bepsilon^t \|
    \\
    & 
    \le
    \sigma_A \| \bx^t - \avgx^t \1\|
    +
    \| \bepsilon^t \|,
  \end{align*}
  where we used the sub-multiplicative property of the $2$-norm, 
  we set $\sigma_A = \| A - \1\1^\top /N\|$ (i.e., the contraction
	factor associated to the matrix $A$,
  cf. Appendix~\ref{sec:consensus_static_network}), and 
  we used the bound $\| I - \1\1^\top /N\| \le 1$.

  By using the explicit solution for the free evolution and the forced evolution 
  of a linear time-invariant system, the term $\| \bx^t - \avgx^t \1 \|$ can be
  bounded as follows
  \begin{align*}
    \| \bx^t - \avgx^t \1 \|
    & 
    \le
    \sigma_A^t \| \bx^0 - \avgx^0 \1\|
    +
    \smallsum_{\tau=0}^{t-1}
    \sigma_A^{t-1-\tau}
    \| \bepsilon^\tau \|.
  \end{align*}
  Since by assumption $\| \bepsilon^\tau \| \to 0$ (cf. Assumption~\ref{primal:stepsize_diminishing}
  and \ref{primal:regularity_assumption}(i)) and $\sigma_A \in (0,1)$,
  it can be proven that
  \begin{align}
  \label{primal:summable_term}
  \lim_{t \to \infty}
  \smallsum_{\tau=0}^{t-1}
    \sigma_A^{t-1-\tau}
    \| \bepsilon^\tau \|
    = 0,
  \end{align}
  which, in turns, implies that 
  \begin{align}
  \label{primal:distr_subg_consensus}
    \lim_{t \to \infty}  
    \| \bx^t - \avgx^t \1 \|
    = 0.
  \end{align}

Next we show the summability condition. It holds
\begin{align}
\label{primal:distr_subg_consensus_summability}
  \begin{split}
  \lim_{T\to \infty}
  \smallsum_{t =0}^T
  \gamma^t \| \bx^t - \avgx^t \1 \|
  & \le
  \lim\limits_{T\to \infty}
  \smallsum_{t =0}^T
  \gamma^t \sigma_A^t \| \bx^0 - \avgx^0 \1\|
  \\
  & \quad
  +
  \lim\limits_{T\to \infty}
  \smallsum_{t =0}^T
  \gamma^t
  \smallsum_{\tau=0}^{t-1}
  \sigma_A^{t-1-\tau}
  \| \bepsilon^\tau \|
  \\
  & 
  \stackrel{(a)}{=}
  \lim\limits_{T\to \infty}
  \smallsum_{t =0}^T
  \gamma^t \sigma_A^t \| \bx^0 - \avgx^0 \1\|
  \\
  & \quad
  +
  \lim\limits_{T\to \infty}
  \smallsum_{t =0}^T
  \gamma^t
  \smallsum_{\tau=0}^{t-1}
  \sigma_A^{t-1-\tau}  
  \gamma^\tau C
  \\
  & 
  \stackrel{(b)}{\le}
  \kappa,
  \end{split}
\end{align}
for some finite $\kappa$, where in (a) we rearranged terms;
in (b) the first term is bounded due to geometric
series properties (cf. Assumption~\ref{primal:stepsize_diminishing}
and recall $\sigma_A \in (0,1)$), 
while the second one can be shown to be bounded by using
the Young's inequality%
\footnote{For all $a \ge 0$, and $b \ge 0$, it holds $2ab \le a^2 + b^2$.}
to write
$$
  \smallsum_{t =0}^T
  \gamma^t
  \smallsum_{\tau=0}^{t-1}
  \sigma_A^{t-1-\tau}  
  \gamma^\tau 
  \le
  \smallsum_{t =0}^T
  \smallsum_{\tau=0}^{t-1}
  \sigma_A^{t-1-\tau}  
  \Big( (\gamma^t)^2 +
  (\gamma^\tau)^2 \Big)/2,
$$
and, then, exploiting subgradient boundedness (cf. 
Assumption~\ref{primal:regularity_assumption}), geometric series properties 
(recall $\sigma_A \in (0,1)$) and the step-size properties (cf. square 
summability of $\gamma^t$ in Assumption~\ref{primal:stepsize_diminishing}).

Finally, we study convergence to the optimum by showing that a proper candidate 
function, say $V^t$, decreases along the algorithmic evolution.
Let $V^t$ be a measure of the distance between the local solution estimates 
$\bx_i^t$, $i\in\until{N}$, and an optimal solution to 
problem~\eqref{primal:cost_coupled_problem}, i.e.,
\begin{align}
\label{primal:Lyapunov_distr_subg}
  V^t 
  \triangleq 
  \smallsum_{i=1}^N \| \bx_i^t - \bx^\star \|^2,
\end{align}
where $\bx^\star \in X^\star$.
Due to convexity of each $f_i$ %
and subgradient boundedness (cf. Assumption~\ref{primal:regularity_assumption}), 
it follows that
  \begin{align*}
  \begin{split}
    V^{t+1} 
    =
    \smallsum_{i=1}^N \| \bx_i^{t+1} - \bx^\star \|^2
    & 
    =
    \smallsum_{i=1}^N 
    \| \bv_i^{t+1}
      - \gamma^t \subgrad f_i ( \bv_i^{t+1} )
    - \bx^\star \|^2
    \\
    & 
    =
    \smallsum_{i=1}^N 
    \| \bv_i^{t+1} - \bx^\star \|^2
    +
    \smallsum_{i=1}^N 
    \Big \| \gamma^t \subgrad f_i ( \bv_i^{t+1} ) \Big \|^2
    \\    
    & 
    \quad 
    -
    2 \gamma^t
    \smallsum_{i=1}^N 
    \subgrad f_i ( \bv_i^{t+1} )^\top 
    (\bv_i^{t+1} - \bx^\star )
    \\
    & 
    \stackrel{(a)}{\le}
    \smallsum_{j=1}^N \left ( \smallsum_{i=1}^N a_{ij}\right ) \| \bx_j^t - \bx^\star \|^2
    + 
    (\gamma^t)^2 \smallsum_{i=1}^N C_i^2
    \\
    & 
    \quad - 
    2\gamma^t \smallsum_{i=1}^N 
    \Big ( f_i ( \bv_i^{t+1} ) - f_i(\bx^\star) \Big ),
  \end{split}
  \end{align*}
  where in (a) we exploited convexity of the square norm $\| \cdot \|^2$
  and weights properties (cf. Assumption~\ref{primal:assumption_network_subg}) 
  to write $ \sum_{i=1}^N \|  \bv_i^{t+1} - \bx^\star \|^2 
  =
  \sum_{i=1}^N \|  \sum_{j\in\nbrs_i} a_{ij} (\bx_j^t - \bx^\star) \|^2 
  \le 
  \sum_{j=1}^N (\sum_{i=1}^N  a_{ij} ) \| \bx_j^t - \bx^\star \|^2$; subgradient 
  boundedness (cf. Assumption~\ref{primal:regularity_assumption});
  and the subgradient definition (cf. Appendix \ref{app:subgradient_method}.
  Compactly it holds
  \begin{align}
  \label{primal:descent_1}
    V^{t+1} 
    \le 
    V^t 
    -
    2\, \gamma^t \smallsum_{i=1}^N \Big ( f_i ( \bv_i^{t+1} ) - f_i(\bx^\star) \Big ) 
    + 
    (\gamma^t)^2
    C^2,
  \end{align}
  where $C^2 \triangleq \sum_{i=1}^N C_i^2$.
  Adding and subtracting $2\, \gamma^t \sum_{i=1}^N f_i( \avgx^{t + 1} )$ yields
  \begin{align}
  \begin{split}
    V^{t+1} 
    & 
    \le 
    V^t 
    -
    2 \gamma^t \smallsum_{i=1}^N \Big ( f_i( \avgx^{t + 1} ) - f_i(\bx^\star) \Big) 
    \\
    &
    \quad
    +
    2 \gamma^t \smallsum_{i=1}^N 
    \Big ( f_i( \avgx^{t + 1} ) - f_i (\bv_i^{t+1} ) \Big) 
    + 
    (\gamma^t)^2
    C^2
    \\
    & 
    \stackrel{(a)}{\le}
    V^t 
    -
    2 \gamma^t \smallsum_{i=1}^N \Big( f_i( \avgx^{t + 1} ) - f_i(\bx^\star) \Big ) 
    \\
    &
    \quad
    +
    2 C \gamma^t \smallsum_{i=1}^N \| \avgx^{t+1} - \bx_i^t \|
    + 
    (\gamma^t)^2
    C^2,
    \end{split}
  \end{align} 
  where in (a) we used subgradient boundedness to write
  \begin{align*}
    \smallsum_{i=1}^N \Big | f_i( \avgx^{t + 1} ) - f_i ( \bv_i^{t+1} ) \Big |
    & \le 
    C
    \smallsum_{i=1}^N \| \avgx^{t + 1} - \smallsum_{j\in\nbrs_i} a_{ij} \bx_j^t \|
    \\
    & \le 
    C
    \smallsum_{i=1}^N \| \avgx^{t + 1} - \bx_i^t \|.
  \end{align*}
  Notice that $\sum_{i=1}^N ( f_i( \avgx^{t + 1} ) - f_i ( \bx^\star) ) > 0$
  since $\bx^\star$ is a minimum of~\eqref{primal:cost_coupled_problem}.
  Using Lemma~\ref{lem:supermartingale} we can conclude that: 
  \begin{itemize}
    \item the sequence $\{ V^t \}_{t\ge 0}$
    converges to a finite value, say $\bar{V}$, for every $\bx^\star \in X^*$, and 
  
    \item the average sequence $\{ \avgx^t \}_{t\ge 0}$ satisfies
  \begin{align}
  \label{primal:distr_subg_cost_convergence}
    \liminf_{ t\to\infty } 
    \smallsum_{i=1}^N f_i (\avgx^{t + 1}) 
    = 
    \smallsum_{i=1}^N f_i (\bx^\star) 
    = 
    f^\star.
  \end{align}
\end{itemize}  

Since the sequence $\{ V^t \}_{t\ge 0}$ (cf. its definition in~\eqref{primal:Lyapunov_distr_subg}) 
converges, then also the sequence $\{\sum_{i=1}^N\| \bx_i^t - \bx^\star \|\}_{t\ge0}$ 
converges, for every $\bx^\star \in X^\star$. 
Moreover, recall that by consensus achievement~\eqref{primal:distr_subg_consensus}, 
it holds $\lim_{t \to\infty}\| \bx_i^t - \avgx^t\|=0$. 
Therefore, also $\{ \| \avgx^t - \bx^\star \| \}_{t\ge0}$ must converge.

In view of~\eqref{primal:distr_subg_cost_convergence} and of continuity of 
$f$ (due to its convexity), one of the limit points of $\{ \avgx^t \}_{t\ge0}$ must 
belong to $X^\star$; thus, consider a subsequence $\{ \avgx^{t_k} \}_{k\ge0}$ of 
$\{ \avgx^t \}_{t\ge0}$ converging to an optimum, i.e., such that 
$\lim_{k\to\infty} \| \avgx^{t_k} - \bx^\infty\| = 0$,
with $\bx^\infty \in X^\star$.
Convergence of $\avgx^{t_k} $ with the asymptotic consensus property (cf. 
eq.~\eqref{primal:distr_subg_consensus}) implies that also 
$\lim_{k\to\infty} \| \bx_i^{t_k} - \bx^\infty\| = 0$, for all $i\in\until{N}$.
But in view of convergence of $V^t = \sum_{i=1}^N \| \bx_i^t - \bx^\star \|^2$, 
it must be that the (entire) sequence 
$\{ \bx_i^t \}_{\ge 0}$ converges to $\bx^\star \in X^\star$.
\end{proof}

It is worth mentioning that convergence of the distributed subgradient algorithm
to an optimum can only be guaranteed with a diminishing step-size.
This is mainly due to the
fact that at each iteration, each agent $i$ considers an update direction
depending only on its local objective function $f_i$, rather than on the entire
cost function $\sum_{i=1}^N f_i$.

Convergence rates have been proven for the distributed subgradient method
and its variants. In~\cite{nedic2015distributed}, a convergence rate of $\mathcal{O}(\ln{t}/\sqrt{t})$
is proved for an extension of the distributed subgradient algorithm for directed graphs.

\section{Gradient Tracking Algorithm}
\label{sec:distributed_gradient_tracking}

In this section we review a recent method for cost-coupled problems
(cf. Section~\ref{sec:setups_cost_coupled}) that exhibits a faster convergence
rate because it allows for the use of a constant step-size.
The underlying idea of this novel approach is to implement a distributed
consensus-based mechanism to track the gradient of the whole cost function. 
Thanks to this tracking mechanism, a linear convergence rate has been shown for
this scheme, matching the rate of the centralized gradient method.

Formally, we consider a cost-coupled problem in the form~\eqref{primal:cost_coupled_problem},
where the cost functions $f_i$ satisfy suitable regularity properties that will be 
specified next.

In order to understand the concept
underlying the gradient tracking algorithm,
let us consider the (centralized) gradient method applied 
to~\eqref{primal:cost_coupled_problem}. If we denote by
$\bx^t$ the (centralized) solution estimate, the method reads
\begin{align}
\label{primal:centralized_gradient}
  \bx^{t+1} & = \bx^t - \gamma \smallsum_{h=1}^N \nabla f_h (\bx^t).
\end{align}

In a distributed context, each agent $i$ has its own version $\bx_i^t$ of the current solution 
estimate $\bx^t$. Thus, the gradient scheme~\eqref{primal:centralized_gradient} 
can be adapted as follows
\begin{align*}
  \bx_i^{t+1} 
  & = \smallsum_{j \in \nbrs_i} a_{ij} \bx_j^t 
  - \gamma \smallsum_{h=1}^N \nabla f_h (\bx_h^t),
\end{align*}
where the consensus iteration $\sum_{j \in \nbrs_i} a_{ij} \bx_j^t $ is meant to
enforce an agreement among the agents. However, still the descent direction
$\sum_{h=1}^N \nabla f_h (\bx_h^t)$ requires a global knowledge that is not
locally available.  To overcome this issue, the exact (centralized) descent direction is
replaced by a local descent direction, say $\by_i^t$, which is updated through
a \emph{dynamic} average consensus iteration to eventually track
$\sum_{h=1}^N \nabla f_h (\bx_h^t)$.
Informally, the dynamic average consensus is a distributed algorithm in which 
each agent $i$ has access only to its local (possibly time-varying) signal, say $\br_i^t$,
and wants to track the (time-varying) average signal $1/N \cdot \sum_{h=1}^N \br_h^t$ 
by exchanging information only with neighbors. See Appendix~\ref{sec:dynamic_average_consensus} 
for further details.
In the context of gradient tracking,
each agent's signal is the local gradient at the current estimate, i.e., $\br_i^t = \nabla f_i(\bx_i^t)$.
The following table (Algorithm~\ref{alg:distributed_grad_tracking}) 
formally summarizes the gradient tracking algorithm from the 
perspective of agent $i$, where 
eq.~\eqref{alg:distributed_gradient_tracking_y} describes the dynamic average 
consensus iteration for the tracking of $\sum_{h=1}^N \nabla f_h (\bx_h^t)$,
the local solution estimate $\bx_i^t$ is initialized to any vector in $\real^d$
and the gradient tracker $\by_i^t$ is initialized to $\nabla f_i ( \bx_i^0 )$.
\begin{algorithm}[H]

  \begin{algorithmic}[0]
    \Statex \textbf{Initialization}: $\bx_i^0$ and $\by_i^0 = \nabla f_i ( \bx_i^0 )$
    \smallskip
    
    \Statex \textbf{Evolution}: for $t=0,1,... $
    \smallskip
    
      \StatexIndent[0.5] 
      \textbf{Gather} $\bx_j^t$ from neighbors $j \in \nbrs_i$
      
      \StatexIndent[0.5] 
      \textbf{Update}
      \begin{align}
			  \label{alg:distributed_gradient_tracking_x}
			  \bx_i^{t+1} & = \smallsum_{j\in\nbrs_i} a_{ij} \, \bx_j^t - \gamma \, \by_i^t
      \end{align}
    
      \StatexIndent[0.5] 
      \textbf{Gather} $\by_j^t$ from neighbors $j \in \nbrs_i$

      \StatexIndent[0.5]
      \textbf{Update}
      \begin{align}
			 \label{alg:distributed_gradient_tracking_y}
			  \by_i^{t+1} & = \smallsum_{j\in\nbrs_i} a_{ij} \, \by_j^t 
			  + 
			  \big ( \nabla f_i ( \bx_i^{t+1} ) - \nabla f_i ( \bx_i^t ) \big)
      \end{align}
      
  \end{algorithmic}
  \caption{Gradient Tracking}
  \label{alg:distributed_grad_tracking}
\end{algorithm}

Gradient tracking algorithms have been proposed with several 
names and versions in the literature, but with a common underlying idea.
Early works~\cite{zanella2011newton,zanella2012asynchronous,varagnolo2016newton}
propose the novel idea of distributively tracking a Newton-Raphson direction by 
means of suitable average consensus ratios. In~\cite{carli2015analysis} the same 
approach has been extended to deal with directed, asynchronous networks with lossy 
communication.
More recently, the idea of gradient tracking has been independently proposed by 
several research groups.
In~\cite{dilorenzo2015distributed,dilorenzo2016next} the authors
consider constrained nonsmooth and nonconvex problems, 
while in~\cite{xu2015augmented,xu2018convergence} strongly 
convex, unconstrained, smooth optimization problems are addressed
with agent-specific stepsizes.
Works \cite{sun2016distributed,sun2017distributed} extend the algorithms to
(possibly) time-varying digraphs (still in a nonconvex setting). A convergence 
rate analysis of the scheme was later developed in
\cite{nedic2016cdc,nedic2017achieving,qu2016cdc,qu2017harnessing,xu2018convergence}, 
where~\cite{nedic2016cdc,nedic2017achieving} consider time-varying 
(directed) graphs.
Several other recent works investigate the same scheme under numerous 
variants, such as~\cite{nedic2017geometrically,qu2017accelerated,xi2018addopt,xin2018linear}.

In order to highlight the key tools needed for the analysis of such class of algorithms,
in this survey we investigate a simplified scenario that is characterized afterwards.

\begin{assumption}
\label{primal:assumption_lipschitz_gradient}
  For all $i\in\until{N}$, each cost function $\map{f_i}{\real^d}{\real}$ satisfies the following conditions
  \begin{itemize}
  \item it 
  is $\alpha$-strongly convex, i.e.,
  \begin{align*}
    f_i ( \bw)  \ge f_i ( \bz) + \nabla f_i ( \bz)^\top (\bw - \bz ) + \frac{\alpha}{2} \| \bw - \bz \|^2,
  \end{align*}
  for all $\bw,\bz \in \real^d$ and $\alpha >0$;
  \item it has Lipschitz continuous gradient with constant $L>0$, i.e.,
  \begin{align*}
    \| \nabla f_i ( \bw) - \nabla f_i ( \bz) \| \le L \| \bw - \bz \|,
  \end{align*}
  for all $\bw,\bz \in \real^d$.~\oprocend
  \end{itemize}
\end{assumption}
Since each $f_i$ is a strongly convex function, then also their sum is strongly convex.
Thus under Assumption~\ref{primal:assumption_lipschitz_gradient}, 
problem~\eqref{primal:cost_coupled_problem} has a unique optimal 
solution, denoted by $\bx^\star$. Notice that it holds $\alpha \le L$. 
We point out that one can also consider a more general case in which 
each $f_i$ has $L_i$-Lipschitz continuous gradient. The results proved next
still hold by setting in the analysis $L = \sum_{i=1}^N L_i$.

Similarly to the distributed subgradient algorithm in Section~\ref{sec:distributed_gradient},
we consider a simple network scenario modeled as a fixed, connected and 
undirected graph $\GG = (\until{N},\EE)$. We assume the weights $a_{ij}$ satisfy
a double stochasticity property as formalized
in Assumption~\ref{primal:assumption_network_subg}.

The gradient tracking scheme has been proposed in~\cite{dilorenzo2016next,
sun2017distributed} with a diminishing step-size $\gamma^t$. 
As in the distributed subgradient algorithm (cf. Section~\ref{sec:distributed_gradient}), 
this choice allows one to decouple the convergence analysis in two independent parts,
i.e., consensus achievement and asymptotic convergence of the consensual value to 
the optimum.
When a constant step-size $\gamma$ is used, as done in this survey, the
proof cannot be split in two parts anymore, but consensus and optimality need
to be handled simultaneously.

Since the gradient tracking algorithm is a consensus-based scheme, 
it is convenient to introduce average quantities of local agent variables. 
Namely, we define the average of the solution estimates and the average
of the trackers as
\begin{align*}
  \avgx^t & = 
  \dfrac{1}{N} \smallsum_{i=1}^N \bx_i^t
  \\
  \avgy^t & =
  \dfrac{1}{N} \smallsum_{i=1}^N \by_i^t,
\end{align*}
for all $t\ge 0$.
Using simple algebraic manipulations, it can be shown that the average 
quantities evolve as the following linear dynamical system
\begin{align}
  \label{primal:avgx_evolution}
  \avgx^{t+1} & =
  \avgx^t - \gamma \, \avgy^t
  \\
  \label{primal:avgy_evolution}
  \avgy^{t+1} & = 
  \avgy^t + \dfrac{1}{N} \smallsum_{i=1}^N
    \Big( \nabla f_i (\bx_i^{t+1}) - \nabla f_i (\bx_i^t) \Big).
\end{align}

By exploiting the (column) stochasticity of consensus weights (cf. 
Assumption~\ref{primal:assumption_network_subg}) and the initialization of 
the trackers, i.e., $\by_i^0 = \nabla f_i (\bx_i^0)$, one can show that 
a conservation property for the tracker average $\avgy^t$ holds. 
That is 
\begin{align}
\begin{split}
  \avgy^{t+1} - \dfrac{1}{N} \smallsum_{i=1}^N
  \nabla f_i (\bx_i^{t+1})
  & = 
  \avgy^t - \dfrac{1}{N} \smallsum_{i=1}^N
  \nabla f_i (\bx_i^t)
  \\
  & = 
  \avgy^0 - \dfrac{1}{N} \smallsum_{i=1}^N \nabla f_i (\bx_i^0)
  = 0,
\end{split}
\label{primal:conservation_y}
\end{align}
which implies $\avgy^t = 1/N \cdot \sum_{i=1}^N \nabla f_i (\bx_i^t)$,
for all $t\ge 0$.

The analysis we propose is mainly a detailed version of the proof provided 
in~\cite{xin2018linear} for the above simplified scenario.
In addition, we consider a scalar optimization problem, i.e., we set $d=1$.

The proof starts by characterizing the interconnection among the following quantities:
\begin{itemize}
	\item consensus error $\| \bx^t - \avgx^t \1 \|$, where $\bx^t$ stacks all the $\bx_i^t$;
	\item gradient tracking error $\| \by^t  - \avgy^t  \1 \|$, where $\by^t$ stacks all the $\by_i^t$;
	\item distance from optimality of the average $  \| \avgx^t - \bx^\star \| $, where $\bx^\star$
	is the optimal solution of problem~\eqref{primal:cost_coupled_problem}.
\end{itemize}

We first recall a preliminary result which relies on Lipschitz continuity of the cost gradients.
\begin{lemma}
  Let $\nabla f(\bx^t)$ denote the vector stacking all the gradients $\nabla f_i (\bx_i^t)$, $i\in\until{N}$.
	Under Assumptions~\ref{primal:assumption_lipschitz_gradient} and~\ref{primal:assumption_network_subg}, 
	it holds that
	\begin{align*}
	  \| \nabla f(\bx^{t+1}) - \nabla f (\bx^t) \| 
	  & \le 
	  L  \| \bx^{t+1} - \bx^t \|,
	  \\
	  \bigg \| \dfrac{1}{N} \smallsum_{i=1}^N
	    \Big( \nabla f_i (\bx_i^{t+1}) - \nabla f_i (\bx_i^t) \Big) 
	  \bigg \| 
	  & \le \dfrac{L}{\sqrt{N}}  \| \bx^{t+1} - \bx^t \|,
	  \\
	  \bigg \| \dfrac{1}{N} \smallsum_{i=1}^N
	  \Big( \nabla f_i ( \bx_i^{t+1} ) - \nabla f_i (\avgx^t) \Big) 
	  \bigg \| 
	    & \le \dfrac{L}{\sqrt{N}}  \| \bx^{t+1} - \avgx^t \1 \|,
	\end{align*}
	where $L$ is the Lipschitz constant of $\nabla f_i$, $i\in\until{N}$.
	\oprocend
\end{lemma}
The previous lemma can be easily shown by exploiting the basic algebraic property 
$\sum_{i=1}^N \| \theta_i \|_2 \le \sqrt{N} \| [\theta_1, \ldots, \theta_N]^\top \|_2$,
which follows by concavity of the square root function.

Next, we provide a list of intermediate results that will be used in the convergence theorem.
They explicitly provide linear upper bounds for the three quantities introduced above.
The following lemma characterizes the consensus error.
\begin{lemma}
\label{primal:lem:bound_consensus_error}
Under Assumption~\ref{primal:assumption_lipschitz_gradient}, it holds
\begin{align*}
  \| \bx^{t+1} - \avgx^{t+1} \1 \| 
  & \le
  \sigma_A \| \bx^t - \avgx^t \1 \| + \gamma  \| \by^t  - \avgy^t  \1 \|,
\end{align*}
for all $t\ge0$, where $\sigma_A \in (0,1)$.
\end{lemma}
\begin{proof}
	From~\eqref{alg:distributed_gradient_tracking_x} and~\eqref{primal:avgx_evolution}, we 
	can write
	\begin{align*}
	  \| \bx^{t+1} - \avgx^{t+1} \1 \|
	  &
	  =
	  \| A \bx^t - \gamma \by^t - (\avgx^t - \gamma  \avgy^t ) \1 \|
	  \\
	  & 
	  \le
	  \| A \bx^t - \avgx^t \1 \| + \gamma  \| \by^t  - \avgy^t \1 \|
	  \\
	  & \le
	  \sigma_A \| \bx^t - \avgx^t \1 \| + \gamma  \| \by^t  - \avgy^t \1 \|,
	\end{align*}
	where we used the triangle inequality and $\sigma_A$ is the contraction
	factor associated to the consensus matrix $A$ (cf. Appendix~\ref{app:average_consensus}).
\end{proof}

Next, we bound the distance of the average $\avgx^t$ from $\bx^\star$, optimal 
solution of problem~\eqref{primal:cost_coupled_problem}.
\begin{lemma}
\label{primal:lem:bound_optimality_error}
  Under Assumptions~\ref{primal:assumption_network_subg}
  and~\ref{primal:assumption_lipschitz_gradient}, it holds that
	\begin{align}
	  \| \avgx^{t+1} - \bx^\star \| 
	  & \le
	  \theta \| \avgx^t - \bx^\star \| + \gamma \frac{L}{\sqrt{N}} \| \bx^t - \avgx^t \1 \|,
	\end{align}
	where $\theta = \max(|1 - \alpha \gamma / N |, |1 - L \gamma / N |)$, with $L$ and $\alpha$
	being	the 	Lipschitz constant of $\nabla f_i$ and the strong convexity parameter of $f_i$,
	respectively,	$i \in \until{N}$.	
\end{lemma}
\begin{proof}
Using~\eqref{primal:avgx_evolution}, we can write
\begin{align*}
  \| \avgx^{t+1} - \bx^\star \| 
  & 
  =
  \| \avgx^t - \gamma\, \avgy^t - \bx^\star \|
  \\
  & 
  \stackrel{(a)}{=}
  \left \| \avgx^t - \gamma \dfrac{1}{N} \smallsum_{i=1}^N \nabla f_i (\avgx^t)  - \bx^\star
  + \gamma \dfrac{1}{N} \smallsum_{i=1}^N \nabla f_i (\avgx^t)  
  - \gamma\, \avgy^t \right  \|
  \\
  & 
  \stackrel{(b)}{\le}
  \left \| \avgx^t - \gamma \dfrac{1}{N} \smallsum_{i=1}^N \nabla f_i (\avgx^t)  - \bx^\star
  \right \|
  + \gamma 
  \left \| \dfrac{1}{N} \smallsum_{i=1}^N \nabla f_i (\avgx^t)  
  - \avgy^t \right  \|
  \\
  &
  \stackrel{(c)}{\le}
  \theta
  \| \avgx^t - \bx^\star \|
  +
  \gamma 
  \left \| \dfrac{1}{N} \smallsum_{i=1}^N \nabla f_i (\avgx^t)  
  - \avgy^t \right  \|
  \\
  &
  \stackrel{(d)}{\le}
  \theta
  \| \avgx^t - \bx^\star \|
  +
  \gamma \frac{L}{\sqrt{N}} \| \bx^t  - \avgx^t \1\|,
\end{align*}
where in (a) we added and subtracted $\gamma/N \cdot \sum_{i=1}^N \nabla f_i (\avgx^t)$,
in (b) we used the triangle inequality, in (c) we exploited the convergence rate 
result for a gradient iteration applied to a smooth and strongly convex function%
\footnote{
We recall that a (centralized) gradient iteration applied to the minimization
of a $L_\varphi$-smooth and $\alpha_\varphi$-strongly function $\varphi(z)$ satisfies
(for a sufficiently small $\gamma > 0$)
$\|z - \gamma \nabla \varphi(z) - z^\star \| \le \theta_\varphi \|z-z^\star\|$,
where $\theta_\varphi = \max(|1 - \alpha_\varphi \gamma |, |1 - L_\varphi \gamma |)$
and $z^\star$ is the minimizer of $\varphi$.
} and (d) follows by the conservation property of the tracker (cf. \eqref{primal:conservation_y}), 
the Lipschitz continuity of each $\nabla f_i$ and
the algebraic property 
$\sum_{i=1}^N \| \xi_i \|_2 \le \sqrt{N} \| [\xi_1, \ldots, \xi_N]^\top \|_2$.
\end{proof}

Finally, we provide an upper bound for the tracking error.
\begin{lemma}
\label{primal:lem:bound_tracking_error}
Under Assumptions~\ref{primal:assumption_network_subg} and~\ref{primal:assumption_lipschitz_gradient}, 
it holds
\begin{align}
\begin{split}
  \| \by^{t+1} - \avgy^{t+1} \1 \| 
  & 
  \le
  (\sigma_A + \gamma L) \| \by^t - \avgy^t \1 \|
  \\
  & 
  \quad
  +
  ( L \| A - I \| + \gamma L^2 \sqrt{N} ) \| \bx^t - \avgx^t \1 \|
  \\
  & 
  \quad 
  + \gamma L^2 \sqrt{N} \| \avgx^t - \bx^\star  \|,
\end{split}
\end{align}
for all $t\ge0$, where $\sigma_A$ is the contraction factor associated to the consensus matrix $A$,
$I$ is the identity matrix and $L$ is the Lipschitz constant of $\nabla f_i$, $i \in \until{N}$.
\end{lemma}
\begin{proof}
Under Lipschitz continuity of $\nabla f$, and using~\eqref{primal:avgy_evolution}, it follows
\begin{align*}
  \| \by^{t+1} - \avgy^{t+1} \1 \| 
  &
  \stackrel{(a)}{=}
  \Big \| 
  A \by^t + \nabla f(\bx^{t+1}) - \nabla f(\bx^t) 
  \\
  & 
  \qquad
  - 
  \big(
  \avgy^t +
  \dfrac{1}{N} \smallsum_{i=1}^N
  \big( \nabla f_i (\bx_i^{t+1}) - \nabla f_i (\bx_i^t) \big)
  \big)
  \1
  \Big \|
  \\
  &
  \stackrel{(b)}{\le}
  \| A \by^t - \avgy^t \1 \|
  +
  \left \| \left ( I - \dfrac{1}{N} \1 \1^\top \right )
  ( \nabla f(\bx^{t+1}) - \nabla f(\bx^t) ) \right \|
  \\
    &
  \stackrel{(c)}{\le}
  \sigma_A \| \by^t - \avgy^t \1 \|
  +
  \left \| I - \dfrac{1}{N} \1 \1^\top \right \|
  \| \nabla f(\bx^{t+1}) - \nabla f(\bx^t)  \|
  \\
  &
  \stackrel{(d)}{\le}
  \sigma_A \| \by^t - \avgy^t \1 \|
  +
  L \| A \bx^t - \gamma \by^t - \bx^t \|,
\end{align*}
where in (a) we used~\eqref{alg:distributed_gradient_tracking_y}
and~\eqref{primal:avgy_evolution}, in (b) we rearranged the terms and we used the
triangle inequality, in (c) we used the contraction property of the consensus matrix $A$
(cf. Appendix~\ref{app:average_consensus}) and the sub-multiplicativity of $2$-norm
and finally in (d) we used the fact that $\| I - \1 \1^\top /N \| \le 1$ and
the Lipschitz continuity of $\nabla f$ together with
the update law~\eqref{alg:distributed_gradient_tracking_x}.

Next, we make further modifications on the terms as follows
\begin{align*}
  \| \by^{t+1} - \avgy^{t+1} \1 \| 
  &\le
  \sigma_A \| \by^t - \avgy^t \1 \|
  +
  L \| A \bx^t - \gamma \by^t - \bx^t \|
  \\
  & 
  \stackrel{(a)}{=}
  \sigma_A \| \by^t - \avgy^t \1 \|
  +
  L \| (A - I ) ( \bx^t - \avgx^t \1 ) - \by^t  \|
  \\
  & 
  \stackrel{(b)}{\le}
  \sigma_A \| \by^t - \avgy^t \1 \|
  +
  L \| A - I \| \cdot \| \bx^t - \avgx^t \1 \|
  +
  \gamma L \| \by^t  \|,
\end{align*}
where in (a) we added and subtracted $\avgx^t$ and we exploited row stochasticity of $A$,
and in (b) we used the sub-multiplicativity of $2$-norm and the triangle inequality.
Adding and subtracting $\avgy^t$ and using the triangle inequality we can write 
$\| \by^t \|  \le   \| \by^t - \avgy^t \1 \| + \| \avgy^t \1 \|$, which plugged
into the last equation yields
\begin{align}
\begin{split}
  \| \by^{t+1} - \avgy^{t+1} \1 \| 
  & 
  \le
  (\sigma_A + \gamma L) \| \by^t - \avgy^t \1 \|
  \\
  & 
  +
  L \| A - I \| \cdot \| \bx^t - \avgx^t \1 \|
  +
  \gamma L \| \avgy^t \1 \|.
\end{split}
\label{primal:proof_tracking_error}
\end{align}

Finally, let us manipulate the last term in~\eqref{primal:proof_tracking_error} as
\begin{align}
  \begin{split}
  \| \avgy^t  \1 \|
  & 
  = 
  N \| \avgy^t \|
  = 
  \left \| \dfrac{1}{N} \smallsum_{i=1}^N \nabla f_i (\bx_i^t)  \right \|
  \\
  &
  \stackrel{(a)}{=}
  N \left \| \dfrac{1}{N} \smallsum_{i=1}^N \big( \nabla f_i (\bx_i^t) - \nabla f_i (\bx_i^\star) \big) \right \|
  \\
  &
  \stackrel{(b)}{\le}
  L \smallsum_{i=1}^N \| \bx_i^t - \bx_i^\star  \|
  \\
  & 
  \stackrel{(c)}{\le}
  L \sqrt{N} \| \bx^t - \bx^\star \1  \|
  \\
  &
  \stackrel{(d)}{\le}
  L \sqrt{N} \| \bx^t - \avgx^t \1 \| 
  + L \sqrt{N} \| \avgx^t - \bx^\star  \|,
  \end{split}
\label{primal:proof_tracking_error_2}
\end{align}
where in (a) we added $\sum_{i=1}^N \nabla f_i (\bx_i^\star) = 0$, 
in (b) we exploited the Lipschitz continuity of each $\nabla f_i$, 
in (c) we used the algebraic property 
$\sum_{i=1}^N \| \xi_i \|_2 \le \sqrt{N} \| [\xi_1, \ldots, \xi_N]^\top \|_2$,
and in (d) we added and subtracted $\avgx^t \1 $ and used the triangle inequality.
Combining~\eqref{primal:proof_tracking_error} with \eqref{primal:proof_tracking_error_2}
the proof follows.
\end{proof}

The following theorem states the convergence result for
Algorithm~\ref{alg:distributed_grad_tracking}.

\begin{theorem}
  Let Assumptions~\ref{primal:assumption_network_subg}
  and~\ref{primal:assumption_lipschitz_gradient} hold and let the communication graph
  be undirected and connected.
  Then, there exists a constant $\bar{\gamma} \in (0, N/L)$ such 
  that for all $\gamma \in (0,\bar{\gamma})$ 
  the sequences of local solution estimates $\{ \bx_i^t\}_{t\ge0}$, $i \in\until{N}$,   
  generated by Algorithm~\ref{alg:distributed_grad_tracking}
  are asymptotically consensual to the optimal solution $\bx^\star$
  of problem~\eqref{primal:cost_coupled_problem}, 
  i.e.,
  \begin{align}
    \lim_{t\to \infty} \| \bx_i^{t} - \bx^\star \| = 0,
  \end{align}
  for all $i\in\until{N}$. Moreover, the convergence rate is linear.%
  \footnote{
    A (convergent) sequence $\{ \bz^t \}_{t\ge0}$ is said to converge linearly (or geometrically)
    to $\bz^\star$ if there exists $\rho \in (0,1)$ such that
    $\| \bz^{t+1} - \bz^\star \| \le \rho \| \bz^t - \bz^\star \|$,
    for all $t \ge 0$.}
\end{theorem}
\begin{proof}
  The proof is based on showing a (strict) contraction property along the algorithmic evolution.
  Let us introduce the following vector
  \begin{align*}
    \bv^t \triangleq
    \begin{bmatrix}
      \| \bx^t - \avgx^t \1 \|
      \\
      \| \by^t - \avgy^t \1 \|
      \\
      \| \avgx^t - \bx^\star \|   
    \end{bmatrix}.
  \end{align*}
  Then, combining the results given in Lemma~\ref{primal:lem:bound_consensus_error}, 
  \ref{primal:lem:bound_optimality_error} and~\ref{primal:lem:bound_tracking_error} it holds 
  \begin{align}
  \label{primal:contraction}
    \bv^{t+1}
    \le
    J(\gamma) \bv^t,
  \end{align}
  where the matrix $J(\gamma)$ is defined as
  \begin{align*}
    J(\gamma)
    \triangleq
    \begin{bmatrix}
      \sigma_A & \gamma & 0
      \\
      \left ( L \| A - I \| + \gamma L^2 \sqrt{N} \right ) & \sigma_A +\gamma L & \gamma L^2 \sqrt{N}
      \\
      \gamma \frac{L}{\sqrt{N}} & 0 & \theta
    \end{bmatrix}.
  \end{align*}
  
  Recall that $\theta = \max(|1 - \alpha \gamma / N |, |1 - L \gamma / N |)$.
  Since $\alpha \le L$ and $\gamma \le N/L$, it follows that $\theta= 1 -\alpha\gamma/N$.
  Thus, we can express $J(\gamma)$ as the sum of two structured matrices
  as follows
  \begin{align*}
    J(\gamma) = 
    \begin{bmatrix}
      \sigma_A & 0 & 0
      \\
      L \| A - I \| & \sigma_A &   0
      \\
      0 & 0 & 1 
    \end{bmatrix}
    + \gamma
    \begin{bmatrix}
      0 & 1 & 0
      \\
      L^2 \sqrt{N} & 1 & L^2 \sqrt{N}
      \\
      \frac{L}{\sqrt{N}} & 0 & -\alpha/N
    \end{bmatrix}.
  \end{align*}
  Being $\sigma_A <1$ and due to the triangular structure of the left matrix,
  we can conclude that it has spectral radius equal to $1$.
  Since the eigenvalues of a matrix are a continuous function of its entries,
  we can use a continuity argument to assert that for positive $\gamma$ 
  the spectral radius of $J(\gamma)$ becomes strictly less than $1$ 
  (see~\cite[Theorem~1]{xin2018linear} for a more comprehensive discussion).
  Hence, we have $\bv^{t+1} \le \rho \, \bv^t$ with $\rho \in (0,1)$.
  Thus, $\| \bv^t  - [0,0,0]^\top\| \to 0$ as $t\to \infty$ with linear rate,
  and the proof follows.
\end{proof}

\section{Variants and Extensions of the Basic Gradient Tracking}

Several extensions of the gradient tracking scheme (described in
Section~\ref{sec:distributed_gradient_tracking}) have been proposed in the literature. 
We present some of them without following their historical development but
following a pure conceptual flow.

A first enhancement deals with optimization problems including both composite cost 
functions (i.e., with regularizers) and a common convex constraint.
The main idea is to compute a feasible descent direction
rather than a pure descent direction. 
Thus, let us consider a constrained cost-coupled optimization problem
\begin{align}
  \min_{\bx \in X} \: & \: \smallsum_{i=1}^N f_i (\bx) + r (\bx),
  \label{primal:cost-coupled_regularized}
\end{align}
with $r$ being a convex \emph{regularizer} and $X$ a convex constraint set.
The modified algorithm reads as follows
\begin{align*}
  \Delta \bx_i^t 
  &= 
  \argmin_{\bx_i \in X} \: \: (\bx_i - \bx_i^t)^\top (N\by_i^t)
  + 
  \frac{\tau}{2} \| \bx_i - \bx_i^t \|^2 + \frac{r(\bx)}{N},
  \\  
  \bx_i^{t+1} 
  &= 
  \smallsum_{j\in\nbrs_i} a_{ij} \Big ( (1 - \beta) \bx_j^t + \beta \Delta \bx_j^t \Big ),
  \\
  \by_i^{t+1} & = \smallsum_{j\in\nbrs_i} a_{ij} \, \by_j^t 
  + 
  \big ( \nabla f_i ( \bx_i^{t+1} ) - \nabla f_i ( \bx_i^t ) \big),
\end{align*}
where $\tau>0$ and $\beta \in (0,1]$ are parameters to be suitably tuned.
Notice that $N\by_i^t$ represents a local estimate of $\sum_{j=1}^N \nabla f_j (\bx_j^t)$
that is used to build a linear approximation of $\sum_{j=1}^N f_j (\bx_j^t)$
about the current iterate.
Moreover, notice that $\Delta \bx_i^t \in X$, so that, provided that $\bx_j^t \in X$,
then $\bx_i^{t+1}$ stays feasible.
This constrained version of the gradient tracking 
has been proposed and analyzed in~\cite{dilorenzo2015distributed,
dilorenzo2016next,sun2016distributed,
sun2017distributed,scutari2018distributed,nedic2018multi}. We notice that in these
works, the authors consider a more general nonconvex optimization setting
and propose more general approximation schemes than a simple linearization.
Indeed, using successive convex approximations,
the proposed distributed algorithms are able to solve also nonconvex 
instances of problem~\eqref{primal:cost-coupled_regularized}, which are 
of great interest in learning and estimation applications.

The gradient tracking has been extended also to time-varying and directed 
networks by means of the push-sum protocol (cf. Appendix~\ref{sec:consensus_push-sum})
in both the consensus and the tracking iterations. Formally, the algorithm
reads
\begin{align*}
  \phi_i^{t+1} & = \smallsum_{j\in\nbrs_i} b_{ij}^t \, \phi_i^t
  \\
  \bx_i^{t+1} & = \frac{1}{\phi_i^{t+1}} 
  \left ( \smallsum_{j\in\nbrs_i} b_{ij}^t \, \phi_i^t \bx_j^t - \gamma^t \, \by_i^t \right )
  \\
  \by_i^{t+1} & = \frac{1}{\phi_i^{t+1}}  
  \left (
    \smallsum_{j\in\nbrs_i} b_{ij}^t \, \phi_i^t \by_j^t 
    + 
    \nabla f_i ( \bx_i^{t+1} ) - \nabla f_i ( \bx_i^t ) 
    \right ),
\end{align*}
with $\phi_i^0 = 1$, for all $i\in\until{N}$, and where the time-varying weights 
$b_{ij}^t$ are entries of column stochastic matrices $B^t \in \real^{N\times N}$, 
for all $t \ge 0$.
This extension has been studied in~\cite{dilorenzo2015distributed,
dilorenzo2016next,sun2016distributed,
sun2017distributed,scutari2018distributed,
nedic2017achieving,xin2018linear,xi2018linear,
nedic2017geometrically,nedic2016cdc}.
Notice that the previous extensions have been combined in some of the mentioned 
works to design time-varying gradient algorithm for convex and nonconvex problems.
Recently, a block-wise implementation of the gradient tracking algorithm
has been proposed in~\cite{notarnicola2017camsap,notarnicola2017cdc,notarnicola2018tac}.

\section{Discussion and References}

Early consensus-based algorithms for distributed optimization and estimation
have been proposed and analyzed 
in~\cite{schizas2008consensus,johansson2009randomized,nedic2009distributed,
cattivelli2010diffusion,nedic2010constrained,lobel2011distributed,
chen2012diffusion,jakovetic2014fast}. 
A push-sum version of the subgradient algorithm has been proposed in~\cite{nedic2015distributed}
to deal with time-varying networks.
Extensions to the stochastic set-up are provided in~\cite{ram2010distributed,nedic2016stochastic}
A distributed algorithm using a constant step-size has been proposed 
in~\cite{shi2015extra}, with proved convergence rate $\mathcal{O}(1/t)$
(which can be strenghtened to linear for strongly convex problems).
The algorithmic framework has been extended 
to regularized problems in~\cite{shi2015proximal}, and a detailed convergence 
rate analysis has been proposed in~ \cite{yuan2016convergence}.
Its extension to directed graphs is proposed in~\cite{xi2017dextra}.
Distributed schemes to solve nonconvex optimization problems are proposed in
\cite{bianchi2013convergence,tatarenko2017nonconvex,zeng2018nonconvex}.

As regards gradient tracking algorithms, the interested reader can find
relevant up-to-date references in Section~\ref{sec:distributed_gradient_tracking}.
Second-order approaches have been investigated in
\cite{wei2013distributed_I,wei2013distributed_II,
mokhtari2017network,eisen2017decentralized}.
Netwon-Raphson distributed approaches have been proposed and analyzed
in \cite{zanella2011newton,varagnolo2016newton}. An extension
to networks with packet loss is given in \cite{bof2018multi}.

Distributed schemes working under asynchronous communication
protocols are studied in \cite{nedic2011asynchronous,srivastava2011distributed,
lee2013distributed,lee2016asynchronous,
lin2016distributed,liu2017convergence}. A randomized block-coordinate descent algorithm for 
convex optimization problems with linear constraints is proposed in~\cite{necoara2013random}.
In~\cite{wu2017decentralized} an asynchronous distributed algorithm working also with communication
delays is proposed.

As regards continous-time optimization, a purely primal approach is designed in~\cite{lu2012zero}.
A prediction-correction approach for online distributed optimization has been proposed in \cite{simonetto2017decentralized}.
It is also worth mentioning the works in
\cite{wang2011control,kia2017distributed,sundararajan2017robust,
hatanaka2018passivity,rawat2018distributed}, where a control perspective
is employed to analyze distributed optimization algorithms.
A distributed scenario with a variable number of working nodes is proposed in
\cite{jakovetic2016distributed}.
A novel methodology to design continuous-time distributed optimization algorithms 
using techniques from geometric control theory is investigated 
in~\cite{ebenbauer2017distributed,michalowsky2018lie}.

Among the most recent contributions, a Frank-Wolfe decomposition
approach for convex and nonconvex problems is analyzed in \cite{wai2017decentralized}.
A distributed algorithm based on the proximal minimization is proposed 
in~\cite{margellos2018distributed} to solve convex constrained problems.
In \cite{xu2018bregman}, a distributed scheme using a Bregman penalization has been proposed. 
A distributed optimization algorithm for convex optimization with local inequality constraints
has been studied in \cite{xie2018distributed}. An asynchronous distributed algorithm 
with heterogeneous regularizations and normalizations is proposed in~\cite{hochhaus2018asynchronous}.
A specialized version of the distributed subgradient algorithm for convex feasibility
problems, which allows for an infinite number of constraint sets, is presented in~\cite{farina2019setmembership}.

\section{Numerical Example}
\label{sec:primal:simulations}

In this section we provide a numerical study to show the behavior of the distributed
optimization algorithms presented in this chapter. 

We consider a network of $N = 30$ agents communicating over a fixed, undirected,
connected graph generated according to an Erd\H{o}s-R\'enyi random
model with parameter $p = 0.2$. 
Agents are equipped with a doubly stochastic matrix built according to the 
Metropolis-Hastings rule \cite{xiao2004fast}, i.e.,
\begin{align*}
	a_{ij} = 
	\begin{cases}
	\frac{1}{\max\{ d_i, d_j \} + 1}, & \text{ if } j \neq i \text{ and } (i,j) \in \EE,
	\\
	1-\sum_{j \in \nbrs_i} \frac{1}{\max\{ d_i, d_j \} + 1},  &  \text{ if } j = i,
	\\
	0, & \text{ otherwise. }
	\end{cases}
\end{align*}

We focus on the logistic regression problem introduced in Section~\ref{sec:setups_logistic},
where we suppose that each agent has $m_1 = \ldots = m_N = 10$ samples
with feature space dimension $d = 5$. We generate the points $p_{i,j}$ according to a
normal distribution with zero mean and variance equal to $2$ and we generate the
binary labels $\ell_{i,j}$ from a standard Bernoulli distribution.
Agents must agree on the optimal solution of
problem~\eqref{eq:logistic_regression_problem}, recalled here
\begin{align*}
  \begin{split}
  \min_{w, b} \: 
  & \:
    \smallsum_{i=1}^N
    \underbrace{
    \left (
      \smallsum_{j=1}^{m_i}
      \log \! \left[ 1 + e^{-(w^\top p_{i,j} + b) \ell_{i,j}} \right]
      + \dfrac{C}{2N} \|w\|^2
    \right )
    }_{f_i(w, b)}.
  \end{split}
\end{align*}
The regularization parameter
$C$ is assumed to be equal to $0.01$.
We compare the distributed subgradient method
(cf. Section~\ref{sec:distributed_gradient}), with diminishing step-size
$\gamma^t = (1/t)^{0.8}$, and the gradient tracking
algorithm (cf. Section~\ref{sec:distributed_gradient_tracking}), with
constant step-size $\gamma = 10^{-3}$.

In Figure~\ref{primal:fig_cost} we compare the convergence rate of 
Algorithm~\ref{alg:distributed_subgradient} and Algorithm~\ref{alg:distributed_grad_tracking}.
That is, we plot the absolute value of the difference between the optimal cost $f^\star$ and
the sum of local costs $\sum_{i=1}^N f_i(\bx_i^t)$.
From the theoretical analysis, the cost error of both algorithms is known to
asymptotically converge to zero. However, the gradient tracking algorithm
has a linear convergence rate and converges more quickly
than the distributed subgradient method (see Figure~\ref{primal:fig_cost}).
\begin{figure}[!htpb]
\centering
  \includegraphics[scale=1]{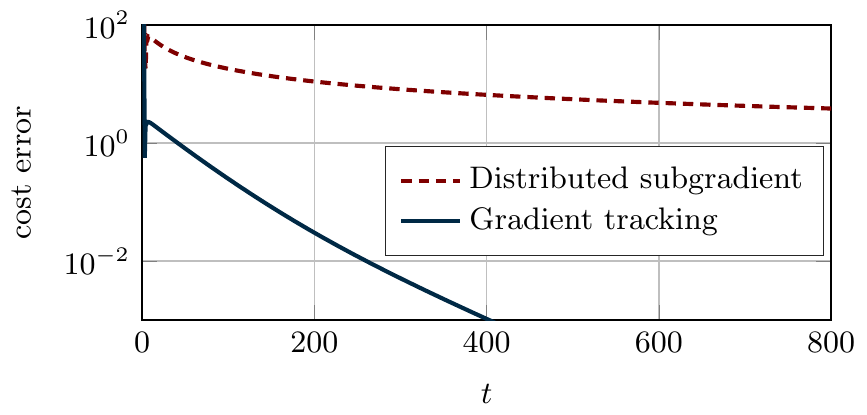}
\caption{
  Evolution of the cost error for the distributed subgradient method
  and for the gradient tracking algorithm.}
  \label{primal:fig_cost}
\end{figure}

In Figure~\ref{primal:fig_dissensus} and~\ref{primal:fig_tracking_error}, we show the
total consensus error of the local solution estimates (for both algorithms)
and of the gradient trackers (for the gradient tracking algorithm), respectively.
\begin{figure}[!htpb]
\centering
  \includegraphics[scale=1]{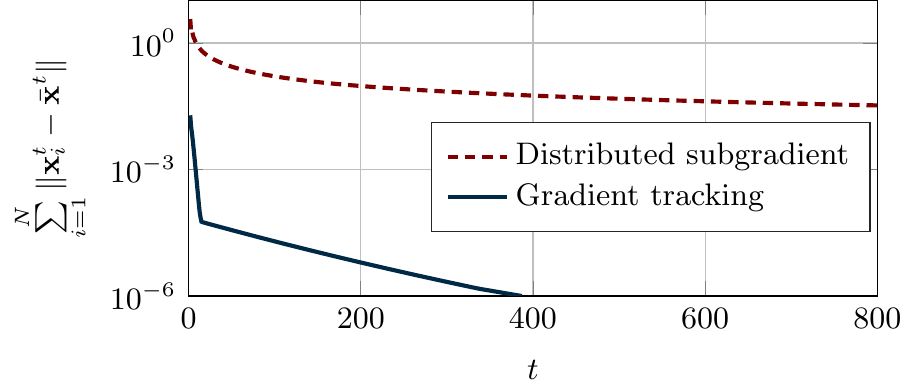}
\caption{
  Evolution of the total consensus error of the local solution estimates $\bx_i^t$ for
  the distributed subgradient method and for the gradient tracking algorithm.}
  \label{primal:fig_dissensus}
\end{figure}

\begin{figure}[!htpb]
\centering
  \includegraphics[scale=1]{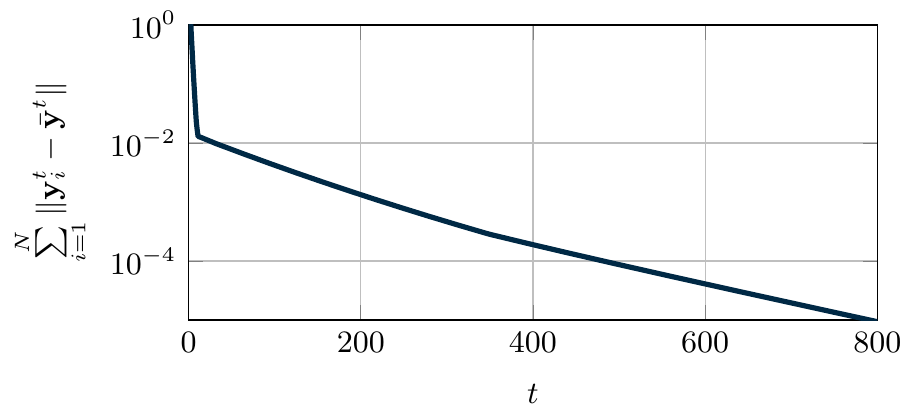}
\caption{
  Evolution of the total consensus error of agents' gradient trackers in
  the gradient tracking algorithm.}
  \label{primal:fig_tracking_error}
\end{figure}

\fi

\iftrue

\chapter{Distributed Dual Methods}
\label{chap:dual}

In this chapter we describe distributed optimization methods based on Lagrangian approaches.
We start by discussing an illustrative example and then we present two 
relevant duality forms to show how duality can be exploited to reformulate
cost-coupled problems as constraint-coupled problems and vice versa.
We describe algorithms for cost-coupled problems based on a decomposition technique
known as \emph{dual decomposition} and on the
Alternating Direction Method of Multipliers (ADMM).
Then, we illustrate duality-based approaches to solve constraint-coupled
problems. To conclude, we give an up-to-date set of references
and we provide numerical examples to highlight the main
features of the discussed algorithms.

\section{Fenchel Duality and Graph Duality}
\label{sec:Fenchel_Graph_duality}

In this section we show how a cost-coupled optimization problem
can be manipulated to obtain alternative (decoupled) problem formulations
that are amenable for distributed computation. 
First, we present a simplified scenario with two agents to illustrate
how duality can be exploited in designing a distributed optimization 
algorithm. Then, we recall a classical duality form known as Fenchel duality
(see~\cite{bertsekas1989parallel}), that paved the way for a number of parallel 
algorithms. Finally, we introduce an alternative and effective approach, that we 
term graph duality, tailored for the distributed framework.

Consider a cost-coupled problem 
\begin{align}
\begin{split}
  \min_{\bx \in\real^d} \: & \: \smallsum_{i=1}^N f_i( \bx )
  \\
  \subj \: & \: \bx \in\bigcap_{i=1}^N X_i,
\end{split}
\label{dual:primal_cost_coupled}
\end{align}
where, for all $i\in\until{N}$, the cost function $f_i$ is convex and the constraint set $X_i$ is 
convex and bounded. These regularity assumptions are standard and guarantee that dual methods apply, i.e.,
that strong duality holds (cf. Appendix~\ref{sec:appendix_duality}).
We will denote by $f^\star$ the optimal cost of problem~\eqref{dual:primal_cost_coupled}.

\subsection{Two-Agent Example}
\label{sec:dual_decomposition_2_agents}

We start by considering a simple ``network'' of $2$ agents and informally discuss how 
duality allows for a suitable decomposition of a cost-coupled problem. 
All the technical details will be provided in the forthcoming sections.

We assume that both agents cooperate to solve the cost-coupled optimization problem
\begin{align} 
\begin{split}
  \min_{\bx \in \real^d}\: & \:  f_1 (\bx) + f_2 (\bx)
  \\  
  \subj \: & \: \bx \in X_1 \cap X_2,
\end{split}
\label{dual:ddec_two_agents_primal_problem}
\end{align}
where $f_1, f_2 : \real^{d} \rightarrow \real$ and $X_1, X_2 \subseteq \real^d$.
Recall that for such cost-coupled set-up, each agent is assumed to know only its own cost
function and constraint (e.g., agent $1$ knows only $f_1$ and $X_1$).

The aim is to decompose problem~\eqref{dual:ddec_two_agents_primal_problem}
by exploiting Lagrangian duality. Specifically, we would like to obtain two
symmetric subproblems so that each agent can solve its subproblem independently.
To this end, we recast problem~\eqref{dual:ddec_two_agents_primal_problem} into
an equivalent formulation by introducing two copies, say $\bx_1$ and $\bx_2$, of the 
decision variable $\bx$ and a coherence constraint to obtain
\begin{align} 
\begin{split}
  \min_{\bx_1, \bx_2}\: & \: f_1(\bx_1) + f_2(\bx_2)
  \\  
  \subj \: & \: \bx_1\in X_1, \: \bx_2\in X_2,
  \\
  & \: \bx_1 = \bx_2.
\end{split}
\label{dual:ddec_two_agents_one_copy}
\end{align}
This reformulation exhibits a convenient structure since the cost function of
each agent depends only on its copy of the decision variable, while the coupling
in the problem is due only to the coherence constraint $\bx_1 = \bx_2$.
Now we write the dual of problem~\eqref{dual:ddec_two_agents_one_copy}
(cf. Appendix~\ref{sec:appendix_duality}).  Let us introduce the Lagrangian
of~\eqref{dual:ddec_two_agents_one_copy}, i.e.,
\begin{align*}
  \LL ( \bx_1, \bx_2, \blambda) 
  = 
  f_1 ( \bx_1 ) + f_2 ( \bx_2 ) + \blambda^\top (\bx_1 - \bx_2),
\end{align*}
where $\blambda \in\real^d$ is the multiplier associated to the constraint $\bx_1 = \bx_2$.
As it will be clear from the forthcoming discussion, the presence of a single
$\blambda$ in $\LL$ does not allow for a symmetric decomposition.
Thus, let us follow an alternative approach, more suited
for distributed computation. Formally, we add another, redundant
constraint and rewrite~\eqref{dual:ddec_two_agents_one_copy} as
\begin{align} 
\begin{split}
  \min_{ \bx_1, \bx_2} \: & \:  f_1( \bx_1) + f_2( \bx_2)
  \\  
  \subj \: & \: \bx_1 \in X_1, \: \bx_2 \in X_2,
  \\
  & \: \bx_1 = \bx_2,
  \\
  & \: \bx_2 = \bx_1,
\end{split}
\label{dual:ddec_two_agents_two_copies}
\end{align}
which is trivially equivalent to problem~\eqref{dual:ddec_two_agents_one_copy}.
For this problem, the Lagrangian becomes
\begin{align}
\label{dual:ddec_two_agents_lagrangian}
\begin{split}
  \LL( \bx_1, \bx_2, \blambda_{12}, \blambda_{21}) 
  & 
  = 
  f_1 (\bx_1) + f_2(\bx_2)
  +
  \blambda_{12}^\top  (\bx_1 - \bx_2) 
  +
  \blambda_{21}^\top  (\bx_2 - \bx_1) 
  \\
  & 
  \stackrel{(a)}{=} 
  f_1(\bx_1) + (\blambda_{12} - \blambda_{21} )^\top \bx_1
  \\
  & 
  \quad
  + 
  f_2 ( \bx_2) + (\blambda_{21} - \blambda_{12})^\top \bx_2,
\end{split}
\end{align}
where $\blambda_{12}$ and $\blambda_{21}$ are the multipliers
associated to the constraints $\bx_1 = \bx_2$ and $\bx_2 = \bx_1$
respectively, and in (a) we use the problem symmetry to rearrange $\LL$ in two similar terms,
each one depending only on a single primal variable, 
i.e., on $\bx_1$ and $\bx_2$ respectively.
The dual function of problem~\eqref{dual:ddec_two_agents_two_copies} is obtained 
by minimizing the Lagrangian~\eqref{dual:ddec_two_agents_lagrangian} with respect 
to the primal variables. Formally,
\begin{align*}
  &
  q ( \blambda_{12}, \blambda_{21})
  = \inf_{ \bx_1 \in X_1, \bx_2 \in X_2} \: \LL( \bx_1, \bx_2, \blambda_{12}, \blambda_{21})
  \\
  &
  =
  \underbrace{
  \min_{ \bx_1 \in X_1} \big( f ( \bx_1) + (\blambda_{12} - \blambda_{21} )^\top \bx_1\big)
  }_{q_1 (\blambda_{12}, \blambda_{21}) }
  +
  \underbrace{
  \min_{ \bx_2\in X_2} \big( f ( \bx_2)+( \blambda_{21} - \blambda_{12})^\top \bx_2 \big)
  }_{q_2 (\blambda_{12}, \blambda_{21}) } \! .
\end{align*}
Finally, we can pose the dual problem as 
\begin{align}
\label{dual:ddec_two_agents_dual_problem}
  \max_{ \blambda_{12}, \blambda_{21}} \: & \: 
  q( \blambda_{12}, \blambda_{21})
  =
  \max_{ \blambda_{12}, \blambda_{21} }\:
  q_1( \blambda_{12}, \blambda_{21} ) + q_2( \blambda_{12}, \blambda_{21}).
\end{align}
Under suitable regularity assumption on the primal problem~\eqref{dual:ddec_two_agents_primal_problem},
problem~\eqref{dual:ddec_two_agents_dual_problem} has the same optimal cost.
Thus, by solving~\eqref{dual:ddec_two_agents_dual_problem}, a dual optimal solution can be exploited to
recover a primal optimal solution.
In Section~\ref{sec:graph_duality}, we described the extended approach for a
general set-up with $N$ agents.

The \emph{distributed dual decomposition} algorithm consists of an iterative
procedure to solve problem~\eqref{dual:ddec_two_agents_dual_problem}
by means of a subgradient algorithm (cf. Appendix~\ref{app:gradient_method}),
and to obtain ultimately a solution of the original problem~\eqref{dual:ddec_two_agents_primal_problem}.
The choice of solving~\eqref{dual:ddec_two_agents_dual_problem} with such algorithm
is convenient since a subgradient of the dual function\footnote{
Notice that here we are slightly abusing terminology. Indeed, subgradients
are defined for convex functions, while the dual function $q$ is concave.
Here, the notation $\subgrad q$ stands for the opposite of a subgradient of $-q$.
}
at a given $( \bar{\blambda}_{12}, \bar{\blambda}_{21} )$ can be computed, in a
distributed way, as
\begin{align*}
  \subgrad q( \bar{\blambda}_{12}, \bar{\blambda}_{21} )
  =
  \begin{bmatrix}
    \subgrad_{\blambda_{12}}  q ( \bar{\blambda}_{12}, \bar{\blambda}_{21} )
    \\  
    \subgrad_{\blambda_{21}}  q( \bar{\blambda}_{12}, \bar{\blambda}_{21} )
  \end{bmatrix}
  = 
  \begin{bmatrix}
    \bar{\bx}_1 - \bar{\bx}_2
    \\
    \bar{\bx}_2 - \bar{\bx}_1
  \end{bmatrix},
\end{align*}
where
\begin{align*}
  \bar{\bx}_1 
  & \in 
  \argmin_{ \bx_1\in X_1} \, f_1(\bx_1) + ( \bar{\blambda}_{12} - \bar{\blambda}_{21} )^\top \bx_1,
  \\  
  \bar{\bx}_2 
  & \in 
  \argmin_{ \bx_2\in X_2} \, f_2(\bx_2) + ( \bar{\blambda}_{21} - \bar{\blambda}_{12} )^\top \bx_2 .
\end{align*}

We assume that agent $1$ maintains and updates $\bx_1$ and $\blambda_{12}$, while 
agent $2$ maintains and updates $\bx_2$ and $\blambda_{21}$.
At the beginning, they initialize $\blambda_{12}^0$ and $\blambda_{21}^0$
to arbitrary values. Then, at each iteration $t \ge 0$ of the algorithm, agents exchange 
their current value of $\blambda_{12}^t$ and $\blambda_{21}^t$
and compute a local estimate of the solution as
\begin{align}
\label{dual:ddec_2ag_primal_min}
\begin{split}
  \bx_1^{t+1} 
  &
  \in 
  \argmin_{ \bx_1\in X_1 } \, f_1 ( \bx_1) + (\blambda_{12}^t - \blambda_{21}^t )^\top \bx_1,
  \\  
  \bx_2^{t+1} 
  & 
  \in
  \argmin_{ \bx_2\in X_2} \, f_2 ( \bx_2) + ( \blambda_{21}^t - \blambda_{12}^t )^\top \bx_2.
\end{split}  
\end{align}

Then, they exchange the updated value of $\bx_1^{t+1}$ and $\bx_2^{t+1}$
to adjust their local dual variable as
\begin{align}
\label{dual:ddec_2ag_dual_update}
\begin{split}
  \blambda_{12}^{t+1} 
  & 
  =
  \blambda_{12}^t 
  + \gamma^t \,
  \subgrad_{\blambda_{12}} q(\blambda_{12}^t, \blambda_{21}^t)
  =
  \blambda_{12}^t 
  + \gamma^t \,
  ( \bx_1^{t+1} - \bx_2^{t+1} ),
  \\[1.2ex]
  \blambda_{21}^{t+1} 
  & 
  =
  \blambda_{21}^t
  + \gamma^t \,
  \subgrad_{\blambda_{21} } q(\blambda_{12}^t, \blambda_{21}^t)
  =
  \blambda_{21}^t 
  + \gamma^t \,
  ( \bx_2^{t+1} - \bx_1^{t+1} ),
\end{split}
\end{align}
where $\gamma^t$ denotes the step-size of the gradient method.
An illustration of how communication and computation interleave is shown in
Figure~\ref{fig:ddec_two_agents_algorithm}.
\begin{figure}[!htbp]
  \centering
  \includegraphics[scale=1]{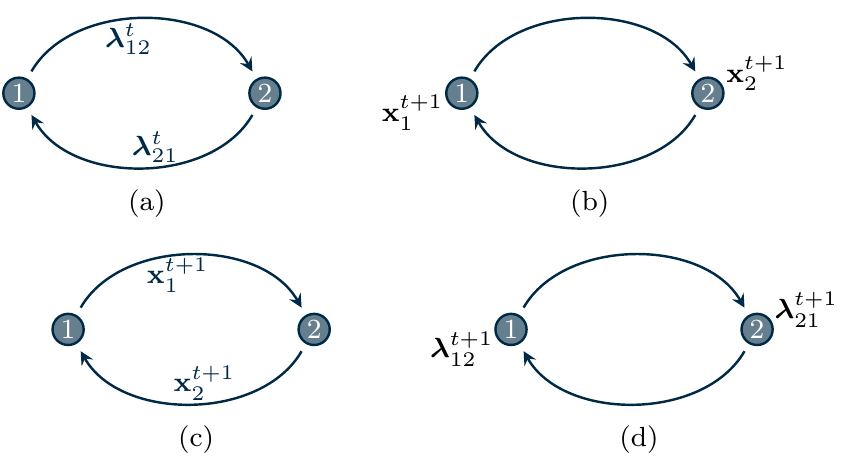}
  \caption{Distributed dual decomposition algorithm for the $2$-agent case. In (a) agents exchange their dual variables 
    to update in (b) the primal variables, cf.~\eqref{dual:ddec_2ag_primal_min}.
    Then, the updated primal variables are communicated in (c) to perform the dual update in (d),
    cf.~\eqref{dual:ddec_2ag_dual_update}.}
  \label{fig:ddec_two_agents_algorithm}
\end{figure}

In Section~\ref{sec:dual_decomposition} we will present and analyze
the general case with $N$ agents and prove that the local solution estimates 
are asymptotically consensual and converge to an optimal solution of the primal 
problem.

\subsection{Fenchel Duality}
\label{sec:Fenchel_duality}
A classical approach to manipulate problem~\eqref{dual:primal_cost_coupled} 
consists in writing its Fenchel dual \cite{bertsekas1999nonlinear}.
To this end, let us introduce copies $\bx_i \in \real^d$ of the optimization variable $\bx$ 
and an auxiliary variable $\bz \in \real^d$ needed to enforce coherence among all the 
copies.
Then, problem~\eqref{dual:primal_cost_coupled} can be equivalently
recast as
\begin{align}
\label{dual:primal_problem_fenchel_copies}
\begin{split}
  \min_{\bx_1,\ldots,\bx_N, \bz} \: & \: \smallsum_{i=1}^N f_i(\bx_i)
  \\
  \subj \: & \: \bx_i \in X_i, \hspace{1cm} i \in \until{N}
  \\
  & \: \bx_i = \bz, \hspace{1.2cm} i \in \until{N}.
\end{split}
\end{align}
The Fenchel-dual problem of~\eqref{dual:primal_cost_coupled} is defined as 
the (standard) dual of~\eqref{dual:primal_problem_fenchel_copies}.
To this end, consider the Lagrangian function of~\eqref{dual:primal_problem_fenchel_copies}, i.e.,
\begin{align*}
 \LL(\bx_1,\ldots,\bx_N,\bz, \blambda_1,\ldots, \blambda_N)
 = \smallsum_{i=1}^N ( f_i(\bx_i) + \blambda_i^\top (\bx_i - \bz) ).
\end{align*}
The minimization of $\LL$ with respect to the primal variables 
gives the dual function
\begin{align*}
  q(\blambda_1, \ldots, \blambda_N)
  & =
  \inf_{\bx_1 \in X_1, \ldots, \bx_N\in X_N, \bz} 
  \smallsum_{i=1}^N \left ( f_i(\bx_i) + \blambda_i^\top (\bx_i - \bz) \right )
  \\
  & =
  \smallsum_{i=1}^N \inf\limits_{\bx_i\in X_i} \left ( f_i(\bx_i) + \blambda_i^\top \bx_i \right )
  +
  \inf\limits_{\bz} \Big (\smallsum_{i=1}^N \blambda_i \Big)^\top \bz
  \\[0.2cm]
  & =
  \begin{cases}
  \smallsum_{i=1}^N \inf\limits_{\bx_i\in X_i} \left ( f_i(\bx_i) + \blambda_i^\top \bx_i \right )  
  & \text{ if } \smallsum_{i=1}^N \blambda_i = \0
  \\
  -\infty & \text{ otherwise}.
  \end{cases}
\end{align*}
Then, the Fenchel-dual problem of~\eqref{dual:primal_cost_coupled} is given by 
the maximization of $q$ over its domain, i.e.,
\begin{align}
\label{dual:Fenchel_dual}
\begin{split}
  \max_{ \blambda_1,\ldots,\blambda_N } \: & \: \smallsum_{i=1}^N q_i (\blambda_i)
  \\
  \subj \: & \: \smallsum_{i=1}^N \blambda_i = \0,
\end{split}
\end{align}
where each $q_i$ is defined as
\begin{align*}
  q_i( \blambda_i) \triangleq \min_{\bx_i \in X_i } \: f_i ( \bx_i ) + \blambda_i^\top \bx_i.
\end{align*}
Problems in the form~\eqref{dual:Fenchel_dual} are often referred to as
\emph{resource allocation} problems.
We point out that~\eqref{dual:Fenchel_dual} has a constraint-coupled structure,
similar to problem~\eqref{setups:constraint-coupled_problem} in 
Section~\ref{sec:setups_constraint_coupled}.
A (centralized) projected gradient method applied to~\eqref{dual:Fenchel_dual}
reads as follows,
\begin{subequations}
\label{dual:parallel_ddec}
\begin{align}
\label{dual:parallel_ddec_primal}
  \bx_i^{t+1}
  & 
  \in \argmin_{\bx_i \in X_i } \: f_i ( \bx_i ) + \big( \blambda_i^t \big)^\top \bx_i, \hspace{0.5cm} i\in\until{N},
  \\[1.2ex]
  \label{dual:parallel_ddec_dual}
  \begin{bmatrix}
    \blambda_1^{t+1} 
    \\
    \vdots
    \\
    \blambda_N^{t+1} 
  \end{bmatrix}
  & = 
  \PP_{\DD} 
   \left ( 
     \begin{bmatrix}
    \blambda_1^t + \gamma\, \bx_1^{t+1}
    \\
    \vdots
    \\
    \blambda_N^t + \gamma\, \bx_N^{t+1}
  \end{bmatrix}
 \right ),
\end{align}
where $\DD = \{ (\blambda_1, \ldots \blambda_N) \mid \sum_{i=1}^N \blambda_i = \0 \}$
and $\PP_{\DD}$ denotes the Euclidean projection onto $\DD$. We assume that 
the algorithm is initialized such that $(\blambda_1^0, \ldots \blambda_N^0) \in \DD$,
e.g., $\blambda_i = \0$ for all $i\in\until{N}$. 
The projection step~\eqref{dual:parallel_ddec_dual} admits the following 
explicit expression
\begin{align*}
  \begin{split}
  \begin{bmatrix}
    \blambda_1^{t+1} 
    \\
    \vdots
    \\
    \blambda_N^{t+1} 
  \end{bmatrix}
  & = 
  \begin{bmatrix}
    \blambda_1^t + \gamma\, \bx_1^{t+1}
    \\[2.3ex]
    \vdots
    \\[2.3ex]
    \blambda_N^t + \gamma\, \bx_N^{t+1}
  \end{bmatrix}
  -
  \begin{bmatrix}
    \dfrac{1}{N} \smallsum_{i=1}^N \big( \blambda_i^t + \gamma\, \bx_i^{t+1} \big)
    \\
    \vdots
    \\
    \dfrac{1}{N} \smallsum_{i=1}^N \big( \blambda_i^t + \gamma\, \bx_i^{t+1} \big)
  \end{bmatrix}
  \\
  & 
  \stackrel{(a)}{=}
  \begin{bmatrix}
    \blambda_1^t + \gamma\, \left ( \bx_1^{t+1} - \dfrac{1}{N}
    \smallsum_{i=1}^N \bx_i^{t+1} \right )
    \\
    \vdots
    \\
    \blambda_N^t + \gamma\, \left ( \bx_N^{t+1} - \dfrac{1}{N}
    \smallsum_{i=1}^N \bx_i^{t+1} \right )
  \end{bmatrix},
  \end{split}
\end{align*}
\end{subequations}
where in (a) we exploited the (recursive) feasibility of the previous iterate
$( \blambda_1^t, \ldots, \blambda_N^t)$.

Algorithm~\eqref{dual:parallel_ddec} is also known as \emph{parallel dual 
decomposition}. Notice that we used properties of dual subgradients involving 
the local primal minimizers to write the dual update (cf. Appendix~\ref{sec:appendix_duality}).
Figure~\ref{dual:fig_parallel_architecture} shows the algorithmic flow of parallel dual 
decomposition.
\begin{figure}[htpb]
	\centering
	\includegraphics[scale=0.92]{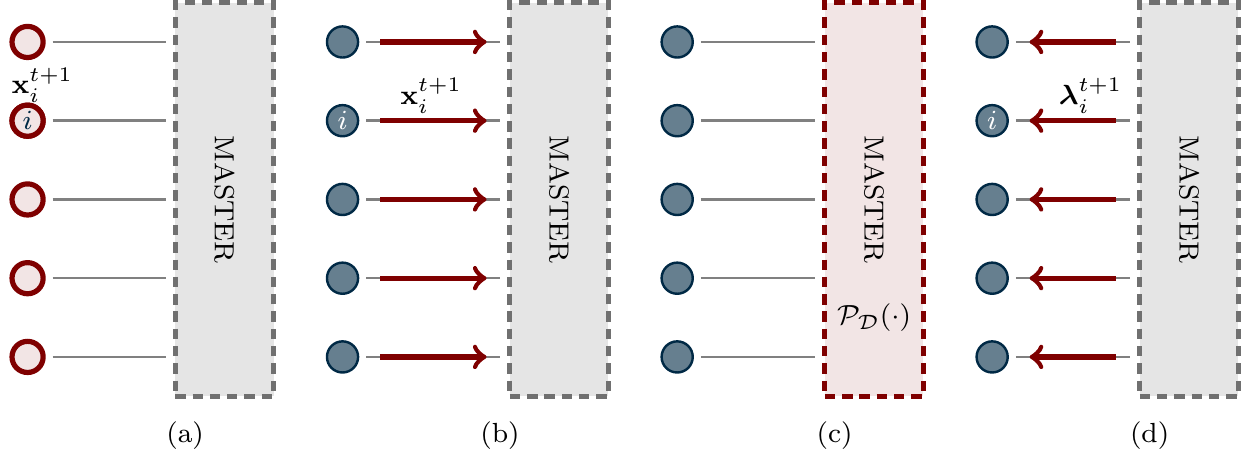}
	\caption{
	  Algorithmic evolution of parallel dual decomposition: in (a) each node
	  updates its primal variable according to~\eqref{dual:parallel_ddec_primal}
	  and sends it to the master node (b). Then, in (c) the master node performs 
	  the projection of the dual variables as in~\eqref{dual:parallel_ddec_dual} and 
	  sends them back to the nodes in (d).
	}
\label{dual:fig_parallel_architecture}
\end{figure}

Notice that problem~\eqref{dual:primal_problem_fenchel_copies}
can also be solved using ADMM (cf. Appendix~\ref{app:ADMM}).
The formal updates can be derived as done for the parallel dual decomposition
by considering the so-called \emph{augmented} Lagrangian. %
It can be shown (see~\cite{boyd2011distributed}) that the resulting algorithm is
\begin{subequations}
\label{dual:parallel_admm}
\begin{align}
\label{dual:parallel_admm_x}
  \bx_i^{t+1}
  & 
  = \argmin_{\bx_i \in X_i } \: f_i ( \bx_i ) + \big( \blambda_i^t \big)^\top \bx_i
  +\frac{\rho}{2} \|\bx_i - \bz^t\|^2, 
  \hspace{0.5cm} \forall\, i
  \\
  \label{dual:parallel_admm_z}
  \bz^{t+1}
  & 
  = \frac{1}{\rho} \smallsum_{i=1}^N \blambda_i^t
  + \smallsum_{i=1}^N\bx_i^{t+1}
  \\
  \label{dual:parallel_admm_dual}
  \blambda_i^{t+1} 
  & = 
  \blambda_i^t + \rho \, ( \bx_i^{t+1} - \bz^{t+1} ),
  \hspace{0.5cm} \forall\, i,
\end{align}
\end{subequations}
where $\rho$ is the positive penalty parameter of the augmented Lagrangian.  It
is worth noting that algorithm~\eqref{dual:parallel_admm} enjoys a parallel
structure similarly to the dual decomposition case.

\subsection{Graph Duality}
\label{sec:graph_duality}

A powerful method to decouple a cost-coupled problem~\eqref{dual:primal_cost_coupled}
into a convenient structure, amenable to distributed computation, is to introduce suitable graph-induced
constraints, that result into an appropriate dual problem. 
We term this methodology \emph{graph duality} to stress that
it combines the classical duality theory with the network structure. Indeed, the resulting 
dual problem heavily depends on the specific network as will be detailed next.
The method that we now formalize is the general form of the approach used in
Section~\ref{sec:dual_decomposition_2_agents}

Let a fixed, undirected and connected graph $\GG = (\until{N},\EE)$ be given,
then we define the $\GG$-dual of~\eqref{dual:primal_cost_coupled} as follows. 
Introduce $N$ copies, say $\bx_1, \ldots, \bx_N$, of the decision variable $\bx$ and
coherence constraints of the copies matching the graph structure, i.e., $\bx_i = \bx_j$ for all $(i,j) \in \EE$.
Then, problem~\eqref{dual:primal_cost_coupled} becomes
\begin{align}
\label{dual:G_primal}
\begin{split}
  \min_{ \bx_1,\ldots, \bx_N} \: & \: \smallsum_{i=1}^N f_i ( \bx_i)
  \\
  \subj \: & \: \bx_i \in X_i, \hspace{1cm} i \in \until{N}
  \\
  & \: \bx_i = \bx_j, \hspace{1cm} (i,j) \in \EE.
\end{split}
\end{align}
Being the graph $\GG$ connected, the equivalence of problems~\eqref{dual:primal_cost_coupled}
and~\eqref{dual:G_primal} is guaranteed.

Let $\blambda_{ij} \in \real^S$ be the multiplier associated to the constraint 
$\bx_i = \bx_j$, then the Lagrangian of~\eqref{dual:G_primal} is
\begin{align}
\label{dual:ddec_lagrangian}
  \LL (\bx_1, \ldots ,\bx_N, \bLambda) 
  = 
  \smallsum_{i=1}^N f_i( \bx_i) 
  +
  \smallsum_{i=1}^N \smallsum_{j\in \nbrs_i} \blambda_{ij}^\top (\bx_i - \bx_j),
\end{align}
where the variable $\bLambda$ stacks all the $| \EE|$ multipliers $\blambda_{ij}$.

Notice that, being the communication graph undirected, for each term 
$\blambda_{ij}^\top (\bx_i- \bx_j)$ in~\eqref{dual:ddec_lagrangian} there is also
a symmetric counterpart $\blambda_{ji}^\top (\bx_j - \bx_i)$.
Thus, the Lagrangian~\eqref{dual:ddec_lagrangian} can be rearranged so as to isolate the 
primal variables $\bx_i$, $i\in\until{N}$, as
\begin{align*}
  \LL(\bx_1, \ldots, \bx_N, \bLambda) 
  = 
  \smallsum_{i=1}^N \Big( f_i( \bx_i ) + \bx_i^\top \smallsum_{j\in \nbrs_i} ( \blambda_{ij} - \blambda_{ji}) \Big).
\end{align*}
At this point, the dual function of~\eqref{dual:G_primal} is obtained by 
minimizing the Lagrangian $\LL$ with respect to the primal variables, leading 
to a separable function.
Finally, the $\GG$-dual of~\eqref{dual:primal_cost_coupled} is the (standard) 
dual of~\eqref{dual:G_primal}, which is given by
\begin{align}
\label{dual:G_dual}
  \max_{ \bLambda } \: q( \bLambda) 
  = 
  \max_{ \bLambda } \: \smallsum_{i=1}^N q_i ( \{ \blambda_{ij}, \blambda_{ji}\}_{(i,j)\in \EE}),
\end{align}
where the $i$-th term $q_i$ of the dual function $q$ is defined as
\begin{align*}
  q_i ( \{ \blambda_{ij}, \blambda_{ji}\}_{(i,j)\in \EE})
  =
  \min_{ \bx_i\in X_i}
  f_i ( \bx_i) + \bx_i^\top \smallsum_{j\in \nbrs_i} ( \blambda_{ij} - \blambda_{ji}),
\end{align*}
for all $i\in\until{N}$.
  We notice that problem~\eqref{dual:G_dual} exhibits interesting features
  for a distributed computation framework. First, it is an unconstrained optimization
  problem with cost function expressed, similarly to the starting problem, as the sum of local 
  terms $q_i$. However, differently from the original problem~\eqref{dual:G_primal},
  in the $\GG$-dual~\eqref{dual:G_dual} the $i$-th cost function
  depends only on the variables of agent $i$ and of its neighbors, rather
  than on the entire stack of decision vectors. %
  In Section~\ref{sec:dual_decomposition}, we will derive a distributed
  algorithm that exploits the special structure of problem~\eqref{dual:G_dual},
  known in the literature as partitioned optimization (cf. Remark~\ref{setups:partitioned_optimization}).

\section{Distributed Dual Decomposition for Cost-Coupled Problems}
\label{sec:dual_decomposition}

In this section, we review an algorithm, known as distributed dual decomposition, 
that relies on duality to solve cost-coupled problems in a distributed way.
Decomposition techniques based on duality have been introduced 
in~\cite{bertsekas1989parallel,palomar2006tutorial,yang2010distributed}.
Typically, they are used to obtain parallel algorithms to speed-up the computation.
However, the distributed extension of those techniques are only partially
discussed in the mentioned references, while in the following we provide
a comprehensive and constructive analysis for this scenario.

We consider $N$ agents in a network that want to cooperatively solve a cost-coupled
problem~\eqref{dual:primal_cost_coupled} (cf. Section~\ref{sec:setups_cost_coupled})
that satisfies the following regularity properties.
\begin{assumption}
\label{dual:problem_regularity_assumption}
  For all $i\in\until{N}$, each $f_i$ is a convex function and
  each $X_i$ is a compact, convex set. Moreover, there exists a vector
  $\bx$ such that $\bx \in \relint{X_i}$\footnote{
  Given a set $X \subset \real^d$, we denote by $\relint{X}$ its relative interior.
  }, for all $i \in \until{N}$.
  \oprocend
\end{assumption}

The latter part of Assumption~\ref{dual:problem_regularity_assumption}
is known in the literature as Slater's constraint qualification, and is a
sufficient condition to ensure that strong duality holds.

Agent $i$ maintains a primal solution estimate $\bx_i^t$,
and dual solution estimates $\blambda_{ij}^t, j \in \nbrs_i$.
The distributed dual decomposition algorithm is based on a subgradient method applied to the $\GG$-dual 
of~\eqref{dual:primal_cost_coupled} (see Section~\ref{sec:graph_duality}), i.e., 
\begin{align}
\label{dual:ddec_Gdual}
  \max_{\bLambda} \: & \: \smallsum_{i=1}^N q_i (\{ \blambda_{ij}, \blambda_{ji}\}_{j \in\nbrs_i} ).
\end{align}

A subgradient of the dual function $q(\bLambda)$ 
at a given $\bar{\bLambda}$ (stacking all the $\bar{\blambda}_{ij}$) 
can be computed in a distributed way as follows. The component of $\subgrad q$ corresponding
to the variable $\blambda_{ij}$ is equal to (cf. Appendix~\eqref{sec:appendix_duality})
\begin{align*}
  \subgrad_{\blambda_{ij} } q( \bar{ \bLambda } ) 
  =
  \bar{\bx}_i - \bar{\bx}_j,
\end{align*}
where $\bar{\bx}_i$ is computed as
\begin{align*}
  \bar{\bx}_i 
  \in
  \argmin_{\bx_i \in X_i}
  f_i ( \bx_i) + \bx_i^\top \smallsum_{j\in \nbrs_i} ( \bar{\blambda}_{ij} - \bar{\blambda}_{ji}),
\end{align*}
and, consistently, for $\bar{\bx}_j$.
Due to the sparse computation of dual subgradients, a subgradient method applied 
to the $\GG$-dual of~\eqref{dual:primal_cost_coupled} turns out to be a 
distributed algorithm.
Formally, each agent $i$ initializes $\blambda_{ij}^t$ for $j \in \nbrs_i$ to any vector in $\real^d$.
At each iteration $t$, each agent $i$ collects from its neighbors $j\in\nbrs_i$ 
the updated dual variables $\blambda_{ji}^t $ and performs a primal minimization
\begin{align*}
  \bx_i^{t+1} 
  \in
  \argmin_{ \bx_i\in X_i}  \: f_i(\bx_i) + \bx_i^\top \smallsum_{j\in \nbrs_i} (\blambda_{ij}^t - \blambda_{ji}^t).
\end{align*}
Then, agents exchange their updated primal solution estimates and perform a subgradient
method step on the dual variables according to
\begin{align*}
  \blambda_{ij}^{t+1} 
  =
  \blambda_{ij}^t
  +
  \gamma^t \, ( \bx_i^{t+1} - \bx_j^{t+1} ),
  \hspace{0.7cm} 
  j\in \nbrs_i,
\end{align*}
where $\gamma^t$ is the step-size sequence.

Figure~\ref{fig:ddec_graphical} shows the algorithmic flow of the distributed
dual decomposition while the following table
(Algorithm~\ref{alg:distributed_dual_dec}) summarizes the algorithm from the
perspective of each agent $i$.
\begin{figure}[t]
  \centering
  \includegraphics[scale=1]{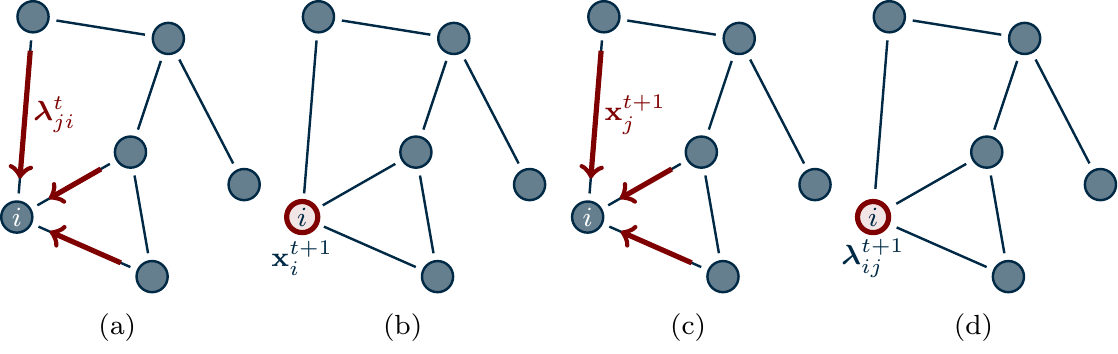}
  \caption{
	  Algorithmic evolution of distributed dual decomposition: in (a) each node
	  receives the dual variables from its neighbors. In (b),
	  the local primal variable is updated according to~\eqref{dual:ddec_algorithm_x}.
	  Then, in (c) the primal variables are broadcast to neighbors
	  to allow in (d) for the dual updates~\eqref{dual:ddec_algorithm_lambda}.
  }
  \label{fig:ddec_graphical}
\end{figure}

\begin{algorithm}[H]
  \begin{algorithmic}[0]
    \Statex \textbf{Initialization}: $\blambda_{ij}^0 $ for all $j\in \nbrs_i$
    \smallskip
    
    \Statex \textbf{Evolution}: for $t=0,1,... $
    \smallskip
    
      \StatexIndent[0.5] \textbf{Gather} $\blambda_{ji}^t$ from neighbors $j \in \nbrs_i$
      
      \StatexIndent[0.5] \textbf{Compute}
      \begin{align}
        \label{dual:ddec_algorithm_x}
        \bx_i^{t+1}
        \in
        \argmin_{ \bx_i \in X_i} 
        \: & \:
        f_i ( \bx_i ) + \bx_i^\top \smallsum_{j\in \nbrs_i} ( \blambda_{ij}^t - \blambda_{ji}^t )
      \end{align}
    
      \StatexIndent[0.5] \textbf{Gather} $\bx_{j}^{t+1}$ from neighbors $j \in \nbrs_i$
      
      \StatexIndent[0.5] \textbf{Update} for all $j\in \nbrs_i$
      \begin{align}
      \label{dual:ddec_algorithm_lambda}
        \blambda_{ij}^{t+1} 
        & 
        =
        \blambda_{ij}^t
        + \gamma^t \, ( \bx_i^{t+1} - \bx_j^{t+1} )
      \end{align}
      
  \end{algorithmic}
  \caption{Distributed Dual Decomposition}
  \label{alg:distributed_dual_dec}
\end{algorithm}

Next, we provide the convergence result for Algorithm~\ref{alg:distributed_dual_dec}.
\begin{theorem}
  Let Assumption~\ref{dual:problem_regularity_assumption} hold.
  Moreover, let the communication graph be undirected and connected and let
  the step-size $\gamma^t$ satisfy Assumption~\ref{primal:stepsize_diminishing}.
	Then, 
	the dual variable sequence $\{\bLambda^t\}_{t\ge 0}$ generated by 
	Algorithm~\ref{alg:distributed_dual_dec} satisfies
	\begin{align*}
	  \lim_{t\to\infty} q( \bLambda^t) = f^\star,
	\end{align*}
	where $f^\star$ is the optimal cost of problem~\eqref{dual:primal_cost_coupled}.
\end{theorem}
\begin{proof}[Proof (Sketch)]
  The proof heavily relies on the constructive derivation we carried out in
  this section. We have proven that the distributed dual decomposition algorithm
  is a subgradient method iteration on the $\GG$-dual~\eqref{dual:ddec_Gdual}.
  Since the primal cost functions $f_i$ are convex and the local sets
  $X_i$ are compact, it is possible to show that the dual function $q$ has bounded 
  subgradients. Thus, by Proposition~\ref{app:subgradient_convergence},
  and since the dual function $q$ is concave, every limit point of $\{\bLambda^t\}_{t\ge 0}$
  is an optimal solution of problem~\eqref{dual:ddec_Gdual}. Therefore, by continuity
  of $q$ and by strong duality, it holds
  \begin{align*}
	  \lim_{t\to\infty} q( \bLambda^t) = f^\star.
	\end{align*}
\end{proof}

Notice that nothing can be said about the convergence of the primal sequence $\{ \bx_i^t \}_{t\ge0}$
generated by Algorithm~\ref{alg:distributed_dual_dec}. 
In fact, due to the lack of strict convexity of the cost functions, there is no guarantee of
feasibilty of the solutions retrieved by the Lagrangian minimization.
This problem has been addressed by introducing averaging mechanisms, i.e., 
let the sequence $\{\widehat{\bx}_i^t\}_{t \ge 0}$ be defined as 
$\widehat{\bx}_i^t = 1/t \sum_{\tau=0}^t \bx_i^\tau$, for all $t$. Then, it holds
\begin{align*}
  \lim_{t\to\infty}
  \smallsum_{i=1}^N f_i (\widehat{\bx}_i^t) 
  & = f^\star,
  \\
  \lim_{t\to\infty}
  \| \widehat{\bx}_i^t - \bx^\star \| & = 0,
  \hspace{1cm} i\in\until{N},
\end{align*}
where $\bx^\star$ and $f^\star$ denote an optimal solution and the optimal 
cost of problem~\eqref{dual:constraint-coupled_problem}, respectively.

\begin{remark}
If each cost function $f_i$ in problem~\eqref{dual:primal_cost_coupled}
is strongly convex then it is possible to improve the result. 
Specifically, under primal strong convexity the dual function $q$ becomes smooth 
(i.e., differentiable with Lipschitz continuous gradient) so that a \emph{gradient} method 
with constant step-size can be applied to solve the dual problem~\eqref{dual:ddec_Gdual}.
Moreover, since strong convexity implies strict convexity,
also primal convergence can be established, i.e., $\lim_{t\to\infty} \| \bx_i^t - \bx^\star\| = 0$ for all $i$
with $\bx^\star$ the optimal solution of~\eqref{dual:primal_cost_coupled}.
This follows since
the Lagrangian minimization admits a unique solution at each iteration $t$.
\oprocend
\end{remark}

As for the rate of convergence of the dual iterates, the algorithm directly inherits the
convergence rate of the standard subgradient method, which is sublinear.
If more regular problems are considered (e.g., strongly convex cost functions),
then the dual function becomes smooth, therefore the linear convergence rate of
gradient method is obtained.

\begin{remark}
  \label{dual:partitioned_remark}
  Distributed dual decomposition can be also applied to partitioned optimization
  problems (cf. Remark~\ref{setups:partitioned_optimization}).  To efficiently
  exploit the partitioned structure of the problem, one can work on copies of
  the relevant portions of the global decision vector. This gives rise to
  tailored distributed dual decomposition algorithms, see, e.g.,
  \cite{carli2013distributed,notarnicola2018partitioned}.  The same procedure
  has been employed for distributed ADMM (cf. the following section) 
  in~\cite{erseghe2012distributed,todescato2015robust,bastianello2018partition}.\oprocend
\end{remark}

In the following section we describe a distributed algorithm that can solve convex optimization
problems and guarantees asymptotic primal feasibility without resorting to averaging mechanisms.

\section{Distributed ADMM for Cost-Coupled Problems}
\label{sec:distributed_ADMM}
In this section we review a distributed algorithm based on the popular Alternating
Direction Method of Multipliers (ADMM, cf. Appendix~\ref{app:ADMM}).
References for the
approach described in this section are, e.g., 
\cite{mateos2010distributed,paul2013network,mota2013dadmm,shi2014linear}

We consider a network of $N$ agents that aim to cooperatively solve a cost-coupled problem
in the form~\eqref{dual:primal_cost_coupled}.
Similarly to distributed dual decomposition, in order to distribute the computation we include sparsity in 
problem~\eqref{dual:primal_cost_coupled} by introducing a set of copies 
of $\bx$ and proper coherence constraints matching the sparsity of the communication graph $\GG$.
That is, problem~\eqref{dual:primal_cost_coupled} can be equivalently stated as
\begin{align} 
\begin{split}
  \min_{\substack{ \bx_1,\ldots, \bx_N \\ \bz_1,\ldots, \bz_N }} \: & \: 
  \smallsum_{i=1}^N f_i(\bx_i)
  \\  
  \subj \: & \: \bx_i \in X_i, \hspace{1cm} i \in\until{N},
  \\
  & \: \bx_i = \bz_j, \hspace{1cm} (i,j) \in \EE,
  \\
  & \: \bx_i = \bz_i, \hspace{1.1cm} i \in \until{N}.
\end{split}
\label{dual:ADMM_N_copies_smart}
\end{align}
This problem reformulation is different from the one used for distributed dual decomposition
and is tailored for the ADMM approach which makes use of the augmented Lagrangian.
Let us introduce $|\EE| + N$ multipliers associated to the coherence constraints.
The augmented Lagrangian is
\begin{align*}
  \LL_\rho(\bX,\bZ,\bLambda)
  = 
  \smallsum_{i=1}^N \bigg ( f_i ( \bx_i ) 
  + 
  \smallsum_{j\in\nbrs_i} \blambda_{ij}^\top (\bx_i - \bz_j )
  &
  +
  \dfrac{\rho}{2} 
  \smallsum_{j\in\nbrs_i} \| \bx_i - \bz_j \|^2
  \\
  +
  \blambda_{ii}^\top (\bx_i - \bz_i )
  & 
  +
  \frac{\rho}{2} \| \bx_i - \bz_i \|^2
  \bigg ),
\end{align*}
where $\bX$, $\bZ$ and $\bLambda$ denote the vectors stacking all the primal variables 
and all the multipliers, respectively.

The ADMM algorithm described in Appendix~\ref{app:ADMM} can applied to 
problem~\eqref{dual:ADMM_N_copies_smart} using the following identifications.
The decision variables $\bx$ and $\bz$ of~\eqref{app:ADMM_problem_Boyd} are
$\bX$ and $\bZ$, respectively.
As for the cost functions, we set
\begin{align*}
  G_1(\bX) = \smallsum_{i=1}^N f_i ( \bx_i ),
  \hspace{0.5cm}
  G_2(\bZ) = 0.
\end{align*}
As for the constraints, $C_1 = X_1\times \cdots \times X_N$ while $C_2 \equiv \real^{N \cdot d}$.
Finally, the linear constraints can be stated as
\begin{align*}
  \underbrace{
	\begin{bmatrix}
	  I_{N \cdot d}
    \\
	  I_{N \cdot d}
  	\end{bmatrix}
  }_{A}
	\begin{bmatrix}
	  \bx_1 \\ \vdots \\ \bx_N
	\end{bmatrix}  
	=
  \underbrace{
	\begin{bmatrix}
	  \text{Adj} \otimes I_d
	  \\
	  I_{N \cdot d}
	\end{bmatrix}  
  }_{B}
	\begin{bmatrix}
	  \bz_1 \\ \vdots \\ \bz_N 
	\end{bmatrix},
\end{align*}
and $c$ equal to zero, where $\text{Adj}$ is the adjacency matrix of $\GG$ (without self-loops)
while $I_{N \cdot d}$ and $I_d$ are $Nd \times Nd$ and $d \times d$ identity matrices, respectively.

\begin{remark}
  An alternative formulation of problem~\eqref{dual:primal_cost_coupled} 
  has been largely used in the literature and it is known as consensus-ADMM 
  formulation (see, e.g.,~\cite{zhu2009distributed}). Formally, %
  the following equivalent formulation of problem~\eqref{dual:primal_cost_coupled}
  is considered,
	\begin{align}
	\begin{split}
	  \min_{\substack{ \bx_1,\ldots, \bx_N \\ \{\bz_{ij} \}_{(i,j)\in\EE} }} \: & \:  
	  \smallsum_{i=1}^N f_i(\bx_i)
	  \\  
	  \subj \: & \: \bx_i \in X_i, \hspace{1cm} i \in\until{N}
	  \\
	  & \: \bx_i = \bz_{ij}, \hspace{1cm} (i,j) \in \EE
	  \\
	  & \: \bx_i = \bz_{ji}, \hspace{1cm} (i,j) \in \EE.
	\end{split}
	\label{dual:ADMM_N_copies_nosmart}
	\end{align}
	The resulting ADMM algorithm is derived by following the same steps
  performed for problem~\eqref{dual:ADMM_N_copies_smart}.  However, notice
  that problem~\eqref{dual:ADMM_N_copies_nosmart} has $|\EE | + N$ variables and $2 \cdot |\EE |$
  coherence constraints, while problem~\eqref{dual:ADMM_N_copies_smart}
  has only $2 \cdot N$ variables and $|\EE | + N$ coherence constraints.\oprocend
\end{remark}

\begin{subequations}
ADMM for problem~\eqref{dual:ADMM_N_copies_smart} turns out to be a fully 
distributed algorithm. Indeed, the primal $\bx$-minimization step reads
\begin{align}
  \bx_i^{t+1} 
  =
  \argmin_{\bx_i\in X_i} \,
  f_i ( \bx_i ) 
  \!+\!
  \Big( \smallsum_{j\in\nbrs_i} \blambda_{ij}^t \!+\! \blambda_{ii}^t \Big)^{\!\top}\! \bx_i
  +\!
  \dfrac{\rho}{2} 
  \smallsum_{j\in\nbrs_i \cup \{ i \} } 
  \!
  \| \bx_i \!-\! \bz_j^t \|^2.
\end{align}
The primal $\bz$-minimization step is
\begin{align}
\label{dual:ADMM_alg_N_agents_z_i}
  \bz_i^{t+1} 
  = 
  \argmin_{\bz_i} \:
  - \Big( \smallsum_{j\in\nbrs_i} \! \blambda_{ji}^t + \blambda_{ii}^t \Big)^{\! \top} \! \bz_i 
  +
  \dfrac{\rho}{2}
  \smallsum_{j\in\nbrs_i \cup \{ i \} } \| \bx_j^{t+1} - \bz_i \|^2.
\end{align}
Finally, the dual ascent step reads
\begin{align}
  \blambda_{ij}^{t+1} 
  & = 
  \blambda_{ij}^t + \rho \, ( \bx_i^{t+1} - \bz_j^{t+1} ),
  \\
  \blambda_{ii}^{t+1} 
  & = 
  \blambda_{ii}^t \, + \rho \, ( \bx_i^{t+1} - \bz_i^{t+1} ),
\end{align}
\end{subequations}
for all $j\in\nbrs_i$ and $i\in\until{N}$.

It is possible to
rephrase the $\bz$-minimization in~\eqref{dual:ADMM_alg_N_agents_z_i}
by noticing that it is an unconstrained quadratic program. The first order necessary condition
of optimality is
\begin{align*}
  - 
  \smallsum_{j\in\nbrs_i} \blambda_{ji}^t - \blambda_{ii}^t 
  - 
  \rho \, 
  \smallsum_{j\in\nbrs_i \cup \{ i \} }\bx_j^{t+1}
  - 
  \rho \, 
  \big( |\nbrs_i| + 1 \big) \bz_i^{t+1}
  = 
  \0.
\end{align*}
Thus, the explicit solution of~\eqref{dual:ADMM_alg_N_agents_z_i} is given by
\begin{align*}
  \bz_i^{t+1} 
  &
  = 
  \frac{ \sum_{j\in\nbrs_i \cup \{ i \} } \bx_j^{t+1} + \bx_i^t  }{|\nbrs_i| + 1}
  +  
  \frac{ \sum_{j\in\nbrs_i} \blambda_{ji}^t + \blambda_{ii}^t  }{\rho \, ( |\nbrs_i| + 1) }.
\end{align*}

Figure~\ref{fig:ADMM_graphical} shows the algorithmic flow of distributed ADMM,
while in Algorithm~\ref{alg:distributed_ADMM_efficient} we summarize the distributed ADMM
algorithm from the perspective of agent $i$.
As for the initialization, each agent $i$ can initialize $\blambda_{ij}^t$ for $j \in \nbrs_i$,
$\blambda_{ii}^t$ and $\bz_i^t$ to arbitrary vectors in $\real^d$.

\begin{figure}[!ht]
  \centering
  \includegraphics[scale=1]{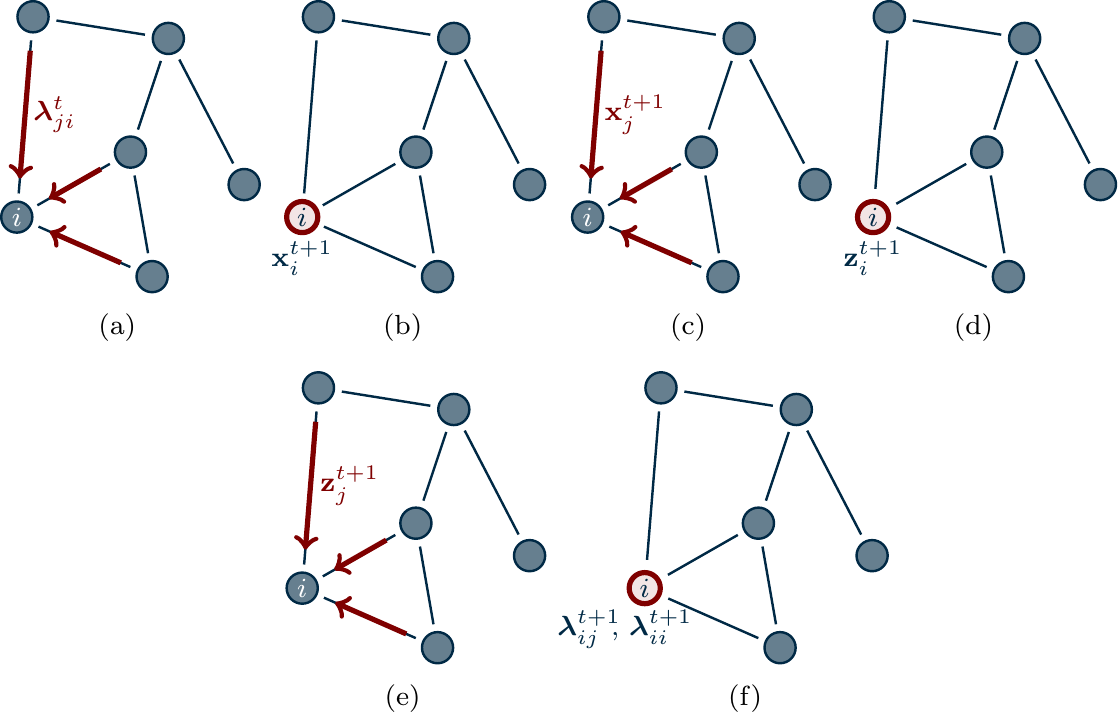}
  \caption{
	  Algorithmic evolution of distributed ADMM: in (a) each node
	  receives the dual variables from its neighbors. In (b),
	  the local primal variable $\bx$ is updated according to~\eqref{dual:dADMM_algorithm_x}.
	  Then, in (c) the variables $\bx$ are broadcast to neighbors
	  to allow for the update of the variables $\bz$, in (d), as in~\eqref{dual:dADMM_algorithm_z}.
	  Finally, in (e) the primal variables $\bz$ are broadcast to neighbors
	  to allow for the dual variables update, in (f), as in~\eqref{dual:dADMM_algorithm_dual}.
  }
  \label{fig:ADMM_graphical}
\end{figure}
\begin{algorithm}[!ht]
  \begin{algorithmic}[0]
    \Statex \textbf{Initialization}: $\blambda_{ij}^0$ for all $j\in \nbrs_i$, $\blambda_{ii}^0$ and $\bz_i^0$
    \smallskip
    
    \Statex \textbf{Evolution}: for $t=0,1,... $
    \smallskip
    
      \StatexIndent[0.5] \textbf{Gather} $\blambda_{ji}^t$ from neighbors $j \in \nbrs_i$
      \StatexIndent[0.5] \textbf{Compute}
      \begin{align}
      \label{dual:dADMM_algorithm_x}
				\begin{split}
				  \bx_i^{t+1} 
				  = 
				  \argmin_{\bx_i\in X_i} 
				  \: 
				  & f_i ( \bx_i ) 
				  +
				  \Big( \smallsum_{j\in\nbrs_i} \blambda_{ij}^t + \blambda_{ii}^t \Big)^{\!\top}\! \bx_i
				  +
				  \dfrac{\rho}{2} 
				  \smallsum_{j\in\nbrs_i \cup \{ i \} } \| \bx_i - \bz_j^t \|^2
				\end{split}
      \end{align}
      
      \StatexIndent[0.5] \textbf{Gather} $\bx_j^{t+1}$ from neighbors $j \in \nbrs_i$
      \StatexIndent[0.5] \textbf{Compute} $\bz_i^{t+1}$ as 
      \begin{align}
      \label{dual:dADMM_algorithm_z}
			  \bz_i^{t+1} 
			  & 
			  = 
			  \frac{ \sum_{j\in\nbrs_i \cup \{ i \} } \bx_j^{t+1} }{|\nbrs_i| + 1}
			  +  
			  \frac{ \sum_{j\in\nbrs_i} \blambda_{ji}^t + \blambda_{ii}^t  }{\rho \, ( |\nbrs_i| + 1) }
      \end{align}
    
      \StatexIndent[0.5] \textbf{Gather} $\bz_j^{t+1}$ from neighbors $j \in \nbrs_i$

      \StatexIndent[0.5] \textbf{Update}
      \begin{align}
      \label{dual:dADMM_algorithm_dual}
      \begin{split}
        \blambda_{ij}^{t+1} 
			  & = \blambda_{ij}^t + \rho \, ( \bx_i^{t+1} - \bz_j^{t+1} ) , \hspace{0.5cm}  j\in \nbrs_i
			  \\
			  \blambda_{ii}^{t+1} 
			  & = \blambda_{ii}^t \, + \rho \, ( \bx_i^{t+1} - \bz_i^{t+1} )
      \end{split}
      \end{align}
      
  \end{algorithmic}
  \caption{Distributed ADMM}
  \label{alg:distributed_ADMM_efficient}
\end{algorithm}

Next, we establish convergence of the distributed ADMM algorithm.
\begin{theorem} 
  Let Assumption~\ref{dual:problem_regularity_assumption} hold
  and let the communication graph be undirected and connected.
  Then, the sequences of local solution estimates $\{ \bx_i^t\}_{t\ge0}$, $i \in\until{N}$,   
  generated by Algorithm~\ref{alg:distributed_ADMM_efficient}
  are asymptotically consensual to an optimal solution $\bx^\star$ of 
  problem~\eqref{dual:primal_cost_coupled}, i.e.,
  \begin{align*}
    \lim_{t\to\infty} \| \bx_i^t - \bx^\star \| = 0.
  \end{align*}
\end{theorem}
\begin{proof}[Proof (Sketch)]
  The proof heavily relies on the constructive derivation we carried out in
  this section. 
  We have shown that Algorithm~\ref{alg:distributed_ADMM_efficient} is
  an istance of the ADMM algorithm (cf.~\eqref{app:ADMM_algorithm} in
  Appendix~\ref{app:ADMM}) applied to
  problem~\eqref{dual:ADMM_N_copies_smart}.
  Thus, by Proposition~\ref{app:prop:ADMM_convergence}, it follows that 
  the primal variable sequence $\{ ( \bx_1^t, \ldots,\bx_N^t)\}_{t\ge0}$
  converges to an optimal (hence feasible) solution of
  problem~\eqref{dual:ADMM_N_copies_smart}.
  Recalling that problem~\eqref{dual:ADMM_N_copies_smart} is an equivalent
  formulation of~\eqref{dual:primal_cost_coupled}, the proof follows.
\end{proof}

\section{Distributed Dual Methods for Constraint-Coupled Problems}
\label{sec:dual_constraint_coupled}

In this section, we consider a constraint-coupled optimization problem (cf.
Section~\ref{sec:setups_constraint_coupled}). We describe how duality
can be exploited to develop distributed optimization algorithms for this problem class.
Notice that the methods discussed in Section~\ref{sec:dual_decomposition}
and Section~\ref{sec:distributed_ADMM} are designed for a different problem 
set-up.

\subsection{Connections between Cost-Coupled and Constraint-Coupled Problems via Duality}
\label{sec:symmetry_dual}

In Section~\ref{sec:Fenchel_Graph_duality}, we have shown that the Fenchel-dual 
problem~\eqref{dual:Fenchel_dual} of a cost-coupled problem is a constraint-coupled problem. 
Next, we show that there exists a more general symmetry between these two set-ups.
In the following, we discuss how duality can be employed to express constraint-coupled
problems as cost-coupled ones. Consider a constraint-coupled problem
\begin{align}
\begin{split}
  \min_{\bx_1,\ldots,\bx_N} \: & \: \smallsum_{i=1}^N f_i(\bx_i)
  \\
  \subj \: & \: \bx_i \in X_i, \hspace{1cm}  i \in \until{N}
  \\
  & \: \smallsum_{i=1}^N \bg_i (\bx_i) \le \0,
\end{split}
\label{dual:constraint-coupled_problem}
\end{align}
where all the quantities have been introduced in Section~\ref{sec:setups_constraint_coupled}.

To derive the dual problem of~\eqref{dual:constraint-coupled_problem},
let us introduce a multiplier $\bmu \in \real^S$ associated to the coupling constraint
$\sum_{i=1}^N \bg_i (\bx_i) \leq \0$. 
Thus, the Lagrangian reads
\begin{align*}
  \LL(\bx_1,\ldots,\bx_N,\bmu) 
  = 
  \smallsum_{i=1}^N 
  \Big ( 
  f_i(\bx_i)
  + \bmu^\top \bg_i (\bx_i) 
  \Big ).
\end{align*}
The dual of problem~\eqref{dual:constraint-coupled_problem} is
\begin{align}
\label{dual:constraint-coupled_dual}
  \max_{\bmu \geq \0 } 
  \: & \: q(\bmu) 
  = 
  \max_{\bmu \geq \0 } \: \: 
  \smallsum_{i=1}^N q_i( \bmu),
\end{align}
where the $i$-th term $q_i$ of the dual function $q$ is defined as
\begin{align}
\label{dual:constraint-coupled_q_i}
  q_i(\bmu) 
  = 
  \min_{\bx_i\in X_i}
  f_i(\bx_i)
  + \bmu^\top \bg_i (\bx_i).
\end{align}
It is easy to see that~\eqref{dual:constraint-coupled_dual} is a
cost-coupled problem.

We consider $N$ agents in a network modeled as a connected, fixed and undirected
graph, which aim to cooperatively solve a constraint-coupled problem~\eqref{dual:constraint-coupled_problem}
satisfying the following assumption.
\begin{assumption}
\label{dual:constraint_coupled_regularity}
  For all $i\in\until{N}$: each function $f_i$ is convex, each constraint $X_i$
  is a non-empty, compact and convex set; each function $\bg_i$ is a component-wise 
  convex function.
  Moreover, there exist 
  $\bar{\bx}_1 \in X_1, \ldots, \bar{\bx}_N \in X_N$ such that
  $\sum_{i=1}^N \bg_i (\bar{\bx}_i ) < \0$.
  \oprocend
\end{assumption}
The latter part of Assumption~\ref{dual:constraint_coupled_regularity} is Slater's
constraint qualification and ensures that strong duality holds.

We recall that each agent $i$ aims to compute only its portion $\bx_i^\star$ of the entire optimal 
solution $(\bx_1^\star,\ldots,\bx_N^\star)$ (cf. Section~\ref{setups:constraint-coupled_problem}).
In the following, we introduce two distributed algorithms that solve
problem~\eqref{dual:constraint-coupled_problem} by means of
problem~\eqref{dual:constraint-coupled_dual}.

\subsection{Distributed Dual Subgradient Algorithm}
\label{sec:dual_distributed_subgradient}

A (centralized) subgradient method (cf. Appendix~\ref{app:subgradient_method}) 
applied to the maximization of the concave 
problem~\eqref{dual:constraint-coupled_dual} reads
\begin{align}
\label{dual:centralized_dual_subg}
\begin{split}
  \bmu^{t+1} 
  & 
  = 
  \PP_{ \bmu \geq \0 }
  \Big (
  \bmu^t
  +
  \gamma^t 
  \subgrad
  q( \bmu^t )
  \Big )
  \\
  & 
  = 
  \PP_{ \bmu \geq \0 }
  \Big (
  \bmu^t
  +
  \gamma^t 
  \smallsum_{i=1}^N
  \subgrad
  q_i (\bmu^t) 
  \Big ).
\end{split}
\end{align}

Notice that, as discussed in Appendix~\ref{sec:appendix_duality},
a subgradient of $q_i$ at $\bmu^t$ can be computed by evaluating
the dualized constraints $\bg_i$ at the minimizer of the Lagrangian, i.e.,
\begin{align*}
  \bx_i^{t+1}
  =
  \argmin_{ \bx_i \in X_i } \: f_i (\bx_i ) + 
  \big( \bmu^t \big)^\top \bg_i (\bx_i ),
\end{align*}
so that $\subgrad q_i (\bmu^t) = \bg_i (\bx_i^{t+1} )$.
The method described by~\eqref{dual:centralized_dual_subg} suggests that the distributed 
subgradient algorithm (cf. Section~\ref{sec:distributed_gradient}) can be applied 
to solve problem~\eqref{dual:constraint-coupled_dual}.

In the following, we describe the distributed dual subgradient algorithm.
Each node $i$ maintains a local dual variable estimate $\bmu_i^t$ that is
iteratively updated according to a distributed subgradient iteration
described by~\eqref{alg:dual_subgr_mu}, and a local primal variable
$\bx_i^t$, computed by minimizing the $i$-th term of the
Lagrangian as in~\eqref{alg:dual_subgr_x}.
Nodes initialize their local dual variables $\bmu_i^t$ to any vector in the positive orthant.
Algorithm~\ref{alg:distributed_dual_subgradient} formally summarizes the
distributed dual subgradient algorithm for a constraint-coupled
optimization problem (from the perspective of agent $i$).
\begin{algorithm}

  \begin{algorithmic}[0]
    \Statex \textbf{Initialization}: $\bmu_i^0 \geq \0$
    \smallskip
    
    \Statex \textbf{Evolution}: for $t=0,1,... $
    \smallskip
    
      \StatexIndent[0.5] \textbf{Gather} $\bmu_j^t$ from neighbors $j\in \nbrs_i$
      \StatexIndent[0.5] \textbf{Compute} %
      \begin{align}
			\begin{split}
			  \bv_i^{t+1} 
			  &
			  =
			  \smallsum_{j \in \nbrs_i }
			  a_{ij} \, \bmu_j^t
			  \\[-0.7ex]
			  \bx_i^{t+1}
			  & 
			  \in 
			  \argmin_{ \bx_i \in X_i } \: f_i (\bx_i ) + 
        \big(\bv_i^{t+1}\big)^\top
			  \bg_i (\bx_i )
      \end{split}
      \label{alg:dual_subgr_x}
      \end{align}

      \StatexIndent[0.5] \textbf{Update} %
			\begin{align}
			\begin{split}
			  \bmu_i^{t+1} 
			  & 
			  = 
			  \PP_{ \bmu \geq \0 }
			  \Big (
          \bv_i^{t+1}
			  +
			  \gamma^t \,
			  \bg_i (\bx_i^{t+1} )
			  \Big )
			\end{split}
			\label{alg:dual_subgr_mu}
			\end{align}
      
  \end{algorithmic}
  \caption{Distributed Dual Subgradient}
  \label{alg:distributed_dual_subgradient}
\end{algorithm}

Being Algorithm~\ref{alg:distributed_dual_subgradient} a distributed
subgradient method (cf. Algorithm~\ref{alg:distributed_subgradient}),
the usual convergence properties (discussed in Chapter~\ref{chap:primal})
apply\footnote{We give the analysis for unconstrained problems,
however the algorithm can be extended to a constrained set-up,
see, e.g.,~\cite{nedic2010constrained}}.
Consider the same network framework as in Section~\ref{sec:distributed_gradient}
and let Assumption~\ref{dual:constraint_coupled_regularity} hold.
We now state the convergence result of the distributed dual subgradient algorithm.
\begin{theorem}
\label{thm:dual_subgradient}
  Let Assumption~\ref{dual:constraint_coupled_regularity}
  hold. Let the communication graph be undirected and connected with weights $a_{ij}$
  satisfying Assumption~\ref{primal:assumption_network_subg} and let
  the step-size $\gamma^t$ satisfy Assumption~\ref{primal:stepsize_diminishing}.
  Then, the sequence of dual variables $\{ \bmu_1^t, \ldots, \bmu_N^t\}_{t\ge0}$ 
  generated by Algorithm~\ref{alg:distributed_dual_subgradient} satisfies
  \begin{align*}
    \lim_{t\to\infty}
    \| \bmu_i^t - \bmu^\star \| = 0, 
    \hspace{1cm} i\in\until{N},
  \end{align*}
  where $\bmu^\star$ is an optimal solution of 
  problem~\eqref{dual:constraint-coupled_dual}, the dual of
  problem~\eqref{dual:constraint-coupled_problem}.
  Moreover, let the sequence $\{\widehat{\bx}_i^t\}_{t \ge 0}$ be defined
  as $\widehat{\bx}_i^t = 1/t \sum_{\tau=0}^t \bx_i^\tau$, for all $t$. Then, it holds
  \begin{align*}
    \lim_{t\to\infty}
    \smallsum_{i=1}^N f_i (\widehat{\bx}_i^t) 
    & = f^\star,
    \\
    \lim_{t\to\infty}
    \| \widehat{\bx}_i^t - \bx^\star \| & = 0,
    \hspace{1cm} i\in\until{N},
  \end{align*}
  where $\bx^\star$ and $f^\star$ denote an optimal solution and the optimal 
  cost of problem~\eqref{dual:constraint-coupled_problem}, respectively.
\oprocend
\end{theorem}
A proof of the statement is provided in~\cite{falsone2017dual} for time-varying %
networks using a proximal minimization perspective.
Notice that Theorem~\ref{thm:dual_subgradient} does not state any convergence
property for the primal variables $\bx_i^t$. To this end, as done in
Section~\ref{sec:dual_decomposition}, it is useful to employ a local running
average (i.e., $\widehat{\bx}_i^t$).
  When the cost function of problem~\eqref{dual:constraint-coupled_problem}
  is strictly convex, problem~\eqref{dual:constraint-coupled_problem} 
  has a unique optimal solution.
  In this scenario, 	convergence of $\bx_i^t$ is guaranteed in any case, 
  so that no primal recovery issues arise and no local running average
  is necessary.

The distributed dual subgradient algorithm enjoys appealing features: \emph{(i)}
local computations at each node involve only the local decision variable and,
thus, scale nicely with respect to the dimension of the decision vector, \emph{(ii)}
privacy is preserved since agents do not communicate, and thus disclose, their estimates 
of the local decision variable, cost or constraints.

\subsection{Relaxation and Successive Distributed Decomposition}
\label{sec:RSDD}

Next, we present a distributed algorithm, named Relaxation and Successive
Distributed Decomposition (RSDD), to solve constraint-coupled problems of the
form~\eqref{dual:constraint-coupled_problem} that has been proposed and analyzed
in~\cite{notarnicola2017ifac, notarnicola2017constraint}.
The main leading ideas of the algorithmic development are: (i) to solve the
(cost-coupled) dual problem~\eqref{dual:constraint-coupled_dual} by means of
distributed dual decomposition, and (ii) to handle infeasibility of local
  problems, occurring during the algorithmic evolution, via a suitable
  relaxation.
The combination of relaxation and duality steps give rise to a simple 
and efficient distributed algorithm that overcomes some limitations of the 
dual distributed subgradient (cf. Section~\ref{sec:dual_distributed_subgradient})
related to primal recovery.

Algorithm~\ref{alg:RSDD} formally states the RSDD distributed algorithm from the 
perspective of node $i$.
\begin{algorithm}

  \begin{algorithmic}[0]
    \Statex \textbf{Initialization}: $\blambda_{ij}^0$ for all $j\in\nbrs_i$

    \Statex \textbf{Evolution}:

      \StatexIndent[0.5] \textbf{Gather} $ \blambda_{ji}^t$ from neighbors $j\in\nbrs_i$%

      \StatexIndent[0.5] \textbf{Compute} $\big( (\bx_i^{t+1}, \rho_i^{t+1} ),\bmu_i^{t+1} \big)$ as a primal-dual 
      optimal solution pair of
      \begin{align}
      \label{dual:RSDD_minimization}
      \begin{split}
        \min_{ \bx_i, \rho_i } \: & \: f_i (\bx_i) + M \rho_i
        \\
        \subj \: & \: \bx_i \in X_i, \: \rho_i \ge 0
        \\[0.5ex]
        & \: \bg_i ( \bx_i ) +
        \smallsum_{ j\in\nbrs_i }  \big( \blambda_{ij}^t - \blambda_{ji}^t \big) \leq \rho_i \1 \\[-1.5ex]
      \end{split}
      \end{align}

      \StatexIndent[0.5] \textbf{Gather} $\bmu_j^{t + 1}$ from neighbors $j\in\nbrs_i$%

      \StatexIndent[0.5] \textbf{Update} for all $j\in\nbrs_i$
      \begin{align}
        \blambda_{ij}^{t+1} & = \blambda_{ij}^t - \gamma^t \big( \bmu_{i}^{t+1} - \bmu_{j}^{t+1} \big)
      \label{dual:RSDD_auxiliary_update}
      \end{align}
  \end{algorithmic}
  \caption{RSDD}
  \label{alg:RSDD}
\end{algorithm}

Informally, the RSDD algorithm 
consists of an iterative two-step procedure.
Each node $i$ stores a set of variables $((\bx_i$, $\rho_i), 
\bmu_i)$, obtained as a primal-dual optimal solution pair of problem
\eqref{dual:RSDD_minimization}. The vector $\bmu_i$ is the multiplier associated to 
the local inequality constraint $\bg_i ( \bx_i ) + \sum_{ j\in\nbrs_i }  ( \blambda_{ij}^t - \blambda_{ji}^t 
) \leq \rho_i \1$.
Notice that problem~\eqref{dual:RSDD_minimization} mimics a local version of the original
problem~\eqref{dual:constraint-coupled_problem}, where the coupling with the other nodes 
is replaced by a local term depending only on
neighboring variables $\blambda_{ij}$ and $\blambda_{ji}$,
$j\in\nbrs_i$. Moreover, this local version of the coupling constraint is also
relaxed, i.e., a positive violation $\rho_{i}\1$ is allowed. %
Finally, instead of minimizing only the local function $f_i$, the (scaled) violation 
$M\rho_i$, $M>0$, enters the cost function as well.
The auxiliary variables $\blambda_{ij}$, $j\in\nbrs_i$, are updated in a second step according 
to a linear law which combines neighboring $\bmu_i$ as shown in~\eqref{dual:RSDD_auxiliary_update}.
Nodes initialize their
variables $\blambda_{ij}^t$, $j\in\nbrs_i$ to arbitrary values.

Similarly to the distributed dual subgradient algorithm,
the RSDD algorithm also enjoys the same appealing features mentioned in
Section~\ref{sec:dual_distributed_subgradient}, i.e., nicely scaling
local computation and information privacy preserving.
Moreover,
a peculiarity of RSDD is that an estimate of a primal optimal solution component is directly computed by 
each agent without any averaging mechanism, which results in a faster algorithm.

Consider the same network framework as in Section~\ref{sec:distributed_gradient}
and let Assumption~\ref{dual:constraint_coupled_regularity} hold.
We now present the convergence
result of RSDD.
\begin{theorem}
  Let Assumption~\ref{dual:constraint_coupled_regularity} hold.
  Let the communication graph be undirected and connected and let
  the step-size $\gamma^t$ satisfy Assumption~\ref{primal:stepsize_diminishing}.
  Moreover, letting $\bmu{}^\star$ be an optimal solution of the dual
  of problem~\eqref{dual:constraint-coupled_problem}, assume
  $M$ be sufficiently large such that $M > \| \bmu{}^\star\|_1$.
  Consider a sequence $\big\{ \bx_i^t, \rho_i^t \big\}_{t\ge 0}$, $i\in \until{N}$, 
  generated by Algorithm~\ref{alg:RSDD}.
  Then, the following holds:
  \begin{enumerate}
    \item[(i)] the sequence $\big \{ \sum_{i=1}^N \big( f_i ( \bx_i^t ) + M \rho_i^t \big) \big\}_{t\ge 0}$ 
    converges to the optimal cost $f^\star$ of~\eqref{dual:constraint-coupled_problem};

    \item[(ii)] every limit point of $\big \{ \bx_i^t \big\}_{t\ge 0}$,
    $ i \in \until{N}$, 
    is a primal optimal (feasible) solution of~\eqref{dual:constraint-coupled_problem}.~\oprocend
  \end{enumerate}
\label{dual:RSDD_convergence}  
\end{theorem}
The proof of Theorem~\ref{dual:RSDD_convergence} can be found 
in~\cite{notarnicola2017constraint}.

In~\cite{camisa2018primal}, Algorithm~\ref{alg:RSDD} has been interpreted as a distributed primal
decomposition method and has been used to solve mixed-integer linear programs
by means a suitable coupling constraint restriction.
The main challenge is due to the presence of local
constraint sets $X_i$ that are mixed-integer polyhedra (i.e., with some
of the components constrained to be integer, see also Section~\ref{eq:ce:sec:milp}).

\begin{remark}
  Another important optimization set-up in smart grid applications arises in 
  so-called Demand Side Management (DSM) programs~\cite{alizadeh2012demand}.
  As an example, a cooperative DSM task has the goal of reducing the hourly and 
  daily variations and peaks of electric demand by optimizing generation, storage 
  and consumption.
  A widely adopted objective in DSM programs is Peak-to-Average Ratio (PAR),
  which gives rise to the following min-max optimization problem
	\begin{align}
	\label{dual:peak_minimization_problem}
	\begin{split}
	  \min_{\bx_{1}, \ldots, \bx_{N}, p} \: & \: p
	  \\
	  \subj \: & \: \bx_{i} \in X_i , \hspace{2.17cm} i\in\until{N},
	  \\
	               & \: \smallsum_{i=1}^N g_{i,s} ( x_{i,s} ) \le p, \hspace{.53cm}  s \in \until{S},
	\end{split}
	\end{align}
	where $p \in \real$ represents the peak value that agents want to shave.
  A duality-based approach similar to the one leading to the RSDD
  distributed algorithm has been proposed and analyzed
  in~\cite{notarnicola2016dsm,notarnicola2016minmax} for solving
  problem~\eqref{dual:peak_minimization_problem}.\oprocend
\end{remark}

\section{Discussion and References}
\label{sec:dual_discussion_extension}

Early popular tutorials on parallel and distributed optimization based on duality are
\cite{palomar2006tutorial,yang2010distributed,boyd2011distributed}.
Distributed algorithms based on the Alternating Direction Method of Multipliers
are proposed in~\cite{mateos2010distributed,paul2013network,mota2013dadmm,chang2015multi,chang2016proximal,teixeira2016admm}.
Convergence rates for ADMM-based algorithms are provided
in~\cite{wei2013on1overk,shi2014linear,jakovetic2015linear,makhdoumi2017convergence}.
A distributed algorithm combining a linearization approach with ADMM
has been proposed in~\cite{ling2015dlm}, while quadratic
approximations have been explored in \cite{mokhtari2016dqm}.
A fast distributed ADMM algorithm for quadratic problems is devised in~\cite{ling2016weighted}.
A more general ADMM framework is considered in \cite{iutzeler2016explicit}, where an explicit
converge rate has been provided.
An application of the distributed ADMM algorithm to an online optimization scenario (i.e., with time-varying cost function)
is analyzed in \cite{ling2014decentralized}.
An asynchronous version of the distributed ADMM algorithm is proposed in \cite{kumar2017asynchronous}.

Primal-dual algorithms for constrained optimization over networks are
given in \cite{zhu2012distributed,latafat2016new}. A primal-dual perturbation approach is explored
in the paper \cite{chang2014distributed}.
An asynchronous version of such algorithm class is provided in \cite{bianchi2016coordinate}.
Augmented Lagrangian algorithms for directed gossip networks are analyzed in~\cite{jakovetic2011cooperative}.
Continuous-time, Lagrangian-based, distributed algorithms are investigated in
\cite{gharesifard2014distributed,kia2015distributed,cherukuri2015distributed,cherukuri2016initialization,
mateos2017distributed,yang2017multi}. A distributed saddle-point algorithm for robust 
linear programs is proposed in~\cite{richert2015robust}.
A saddle-point method for distributed, continuous-time, online optimization 
is proposed in~\cite{lee2016distributed}.
An asynchronous, primal-dual, cloud-based algorithm for distributed convex optimization 
is provided in~\cite{hale2017asynchronous}.
An asynchronous algorithm which allows the presence of local nonconvex
constraints is presented in~\cite{farina2019asymm}.

A dual averaging approach for distributed optimization is proposed in
\cite{duchi2012dual}. A push-sum version for directed networks is analyzed in \cite{tsianos2012push},
while an extension for online optimization is given in \cite{lee2017stochastic}.
A fully distributed dual gradient algorithm to minimize linearly constrained separable convex 
problems, with linear convergence rate, is given in~\cite{necoara2015linear}.
A distributed dual fast gradient algorithm, with sublinear rate, is proposed~\cite{necoara2017fully} for 
linearly constrained separable convex optimization problems.
An asynchronous version of the distributed dual decomposition with composite costs is 
proposed in \cite{notarnicola2017asynchronous}. An extension to a partitioned set-up is 
provided in~\cite{notarnicola2018partitioned}.
A time-varying distributed algorithm based on Fenchel duality is provided in \cite{wu2017fenchel}.
Papers \cite{simonetto2016primal,falsone2017dual} investigate distributed dual subgradient
methods for constraint-coupled optimization. In~\cite{zhang2018consensus} an ADMM approach for
the same set-up is proposed in which multiple consensus steps are needed.
Dual decomposition techniques applied to control problems are proposed 
in~\cite{dinh2013dual,giselsson2013accelerated}. In~\cite{doan2017jacobi} a distributed Jacobi algorithm
for convex optimization problems, arising in distributed model predictive control, is presented.
A fast dual gradient algorithm for network utility maximization is proposed in~\cite{beck2014gradient}.

\section{Numerical Example}
\label{sec:dual_simulations}
In this section, we provide numerical examples of the algorithms presented in this chapter.
Since we considered algorithms for both the cost-coupled set-up and the
constraint-coupled set-up, we analyze the examples in two separate subsections.

As done in Chapter~\ref{chap:primal}, we consider a network of $N = 10$ agents
communicating over a fixed, undirected,
connected graph generated according to an Erd\H{o}s-R\'enyi random
model with parameter $p = 0.2$. 
For the algorithms embedded with consensus iterations, we assume
agents are equipped with a doubly stochastic matrix built according to the 
Metropolis-Hastings rule \cite{xiao2004fast}, i.e.,
\begin{align*}
	a_{ij} = 
	\begin{cases}
	\frac{1}{\max\{ d_i, d_j \} + 1}, & \text{ if } j \neq i \text{ and } (i,j) \in \EE,
	\\
	1-\sum_{j \in \nbrs_i} \frac{1}{\max\{ d_i, d_j \} + 1},  &  \text{ if } j = i,
	\\
	0, & \text{ otherwise. }
	\end{cases}
\end{align*}

\subsection{Cost-coupled Example}
In this subsection, we assume that $N$ agents aim to
cooperatively solve the cost-coupled quadratic program
\begin{align}
\label{dual:simulation_QP}
  \min_{\bx \in \real^5} \: & \:
  \smallsum_{i=1}^N \left ( \bx^\top Q_i \bx + r_i^\top \bx \right ),
\end{align}
where each $Q_i \in \real^{5 \times 5}$ is randomly generated such that
its eigenvalues are drawn uniformly from $[1,10]$.
We compare distributed ADMM (cf. Algorithm~\ref{alg:distributed_ADMM_efficient}),
with $\rho = 0.1$ and distributed dual decomposition (cf. Algorithm~\ref{alg:distributed_dual_dec}),
with diminishing step-size $\gamma^t = (1/t)^{0.7}$.

As for distributed ADMM, in Figure~\ref{dual:fig_ADMM_cost} we show
cost convergence rate, i.e., the evolution of $|\sum_{i=1}^N f_i (\bx_i^t) - f^\star|/|f^\star|$.
\begin{figure}[!htpb]
\centering
  \includegraphics[scale=1]{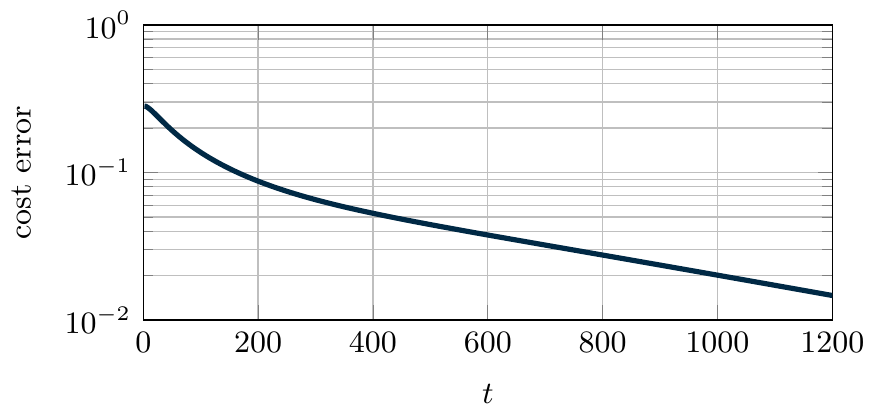}
  \caption{
    Evolution of the cost error for the distributed ADMM algorithm for 
    cost-coupled problems.
  }
  \label{dual:fig_ADMM_cost}
\end{figure}
In Figure~\ref{dual:fig_ADMM_consensus}, we show the consensus error
of the local solution estimates, i.e., $\|\bx_i^t - \bar{\bx}^t\|$ for all $i$,
where $\bar{\bx}^t = 1/N \cdot \sum_{i=1}^N \bx_i^t$. %
\begin{figure}[!htpb]
\centering
  \includegraphics[scale=1]{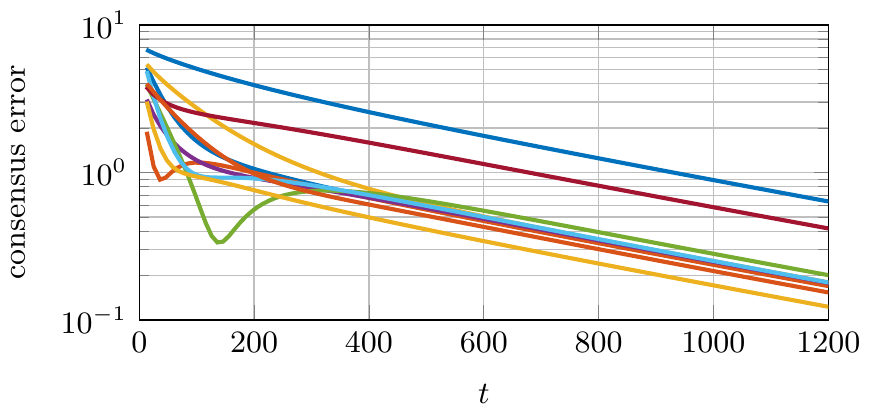}
  \caption{
    Evolution of the consensus error for the distributed ADMM algorithm
    for cost-coupled problems. Each line refers to an agent in the network.
  }
  \label{dual:fig_ADMM_consensus}
\end{figure}

As regards distributed dual decomposition, in Figure~\ref{dual:fig_ddec_cost} 
we show cost convergence. That is, we plot the evolution of primal and dual cost error, i.e., 
$|\sum_{i=1}^N f_i (\bx_i^t) - f^\star|/|f^\star|$ and 
$|q (\bLambda^t) - f^\star|/|f^\star|$. 
As expected for a dual method, dual cost converges faster than primal cost.
\begin{figure}[!htpb]
\centering
  \includegraphics[scale=1]{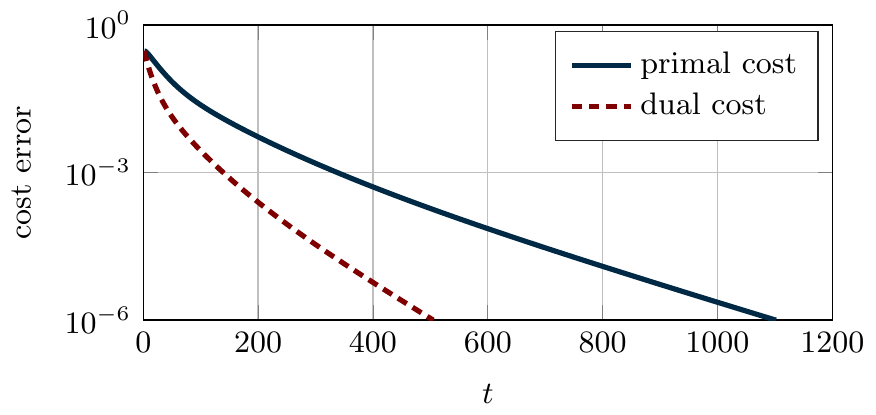}
  \caption{
    Evolution of primal and dual cost errors for the distributed dual decomposition 
    algorithm for cost-coupled problems.
  }
  \label{dual:fig_ddec_cost}
\end{figure}

Finally, in Figure~\ref{dual:fig_ddec_consensus} we show consensus error of 
the local solution estimates, i.e., $\|\bx_i^t - \bar{\bx}^t\|$ for all $i$,
where $\bar{\bx}^t = 1/N \cdot \sum_{i=1}^N \bx_i^t$.
\begin{figure}[!htpb]
\centering
  \includegraphics[scale=1]{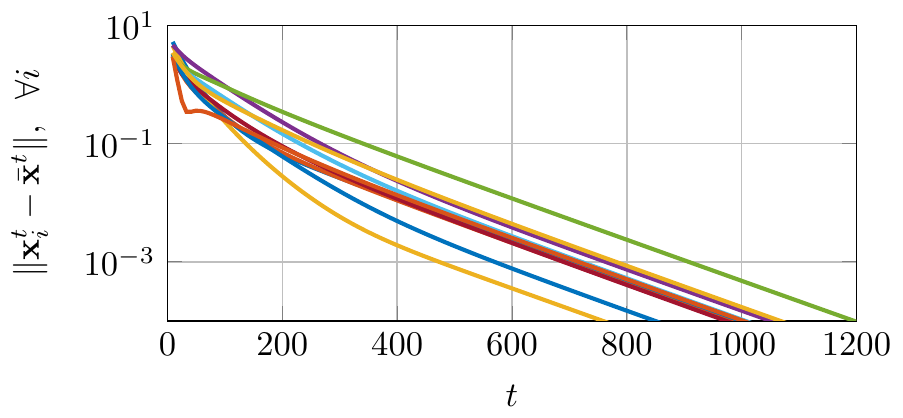}
  \caption{
    Evolution of the consensus error for the distributed dual decomposition 
    for cost-coupled problems. Each line refers to an agent in the network.
  }
  \label{dual:fig_ddec_consensus}
\end{figure}

\subsection{Constraint-coupled Example}
In this subsection, we consider the Microgrid control problem
introduced in Section~\ref{sec:MPC}, where we assume we have
a heterogeneous network of $N = 10$ units with
$4$ generators, $3$ storage devices, $2$ controllable loads 
and $1$ connection to the main grid.
We assume that in the distributed MPC scheme each unit predicts its power 
generation strategy over a horizon of $S = 12$ slots.
The optimization problem to be solved has the form (cf. Section~\ref{sec:MPC})
\begin{align}
\begin{split}
  \min_{\bx_1,\ldots,\bx_N} \: & \: \smallsum_{i = 1}^N f_i(\bx_i)
  \\
  \subj \: 
  & \: 
	  \smallsum_{i \in \GEN} p_{\gen,i}^\tau
	  + \smallsum_{i \in \STOR} p_{\stor,i}^\tau
	  + \smallsum_{i \in \CLOAD} p_{\cload,i}^\tau
	  + p_{\trade}^\tau
	  \ge D^\tau,
	\\
	& \hspace{3.16cm}
	  \forall \: s \in [0,S],
  \\
  & \: \bx_i \in X_i, \hspace{1.68cm} \forall \: i \in \until{N}.
\end{split}
\label{eq:microgrid_control}
\end{align}

We compare RSDD (cf. Algorithm~\ref{alg:RSDD}) and
distributed dual subgradient (cf. Algorithm~\ref{alg:distributed_dual_subgradient}).
For both algorithms, we use the diminishing step-size $\gamma^t = 0.1\cdot (1/t)^{0.7}$.
For RSDD, we set $M = 10 \cdot \|\bmu^\star\|_1$, where $\bmu^\star$ is a dual
optimal solution of the problem~\eqref{eq:microgrid_control} computed by a
centralized solver.

In Figure~\ref{fig:convergence_RSDD} we compare the convergence rate of
RSDD and of distributed dual subgradient.
In particular, for the RSDD algorithm, we plot the difference between the optimal cost $f^\star$ and the sum
of local costs $\sum_{i=1}^N f_i(\bx_{i}^t)$, normalized by $f^\star$.
For the distributed dual subgradient algorithm, we plot the difference between
the optimal cost $f^\star$ and the sum of local costs
$\sum_{i=1}^N f_i(\widehat{\bx}_{i}^t)$, normalized by $f^\star$, where
$\widehat{\bx}_{i}^t$ denotes the $i$-th running average of the local
Lagrangian minimizers $\bx_i^t$.
\begin{figure}[htbp]
\centering
  \includegraphics[scale=1]{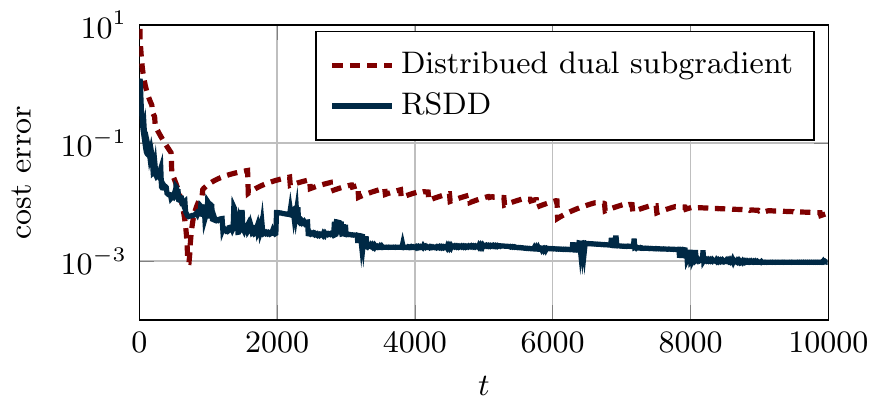}
  \caption{Evolution of the cost error
    $| \sum_{i=1}^N \big( f_i( \bx_{i}^t) + M\rho_i^t \big) - f^\star|/|f^\star|$
    that shows convergence to the optimal cost.
    }
\label{fig:convergence_RSDD}
\end{figure}

For the RSDD algorithm, in Figure~\ref{fig:constraint_violation_RSDD}, we show the algorithmic evolution of
the sum of the penalty parameters $\rho_i^t$ and the maximum violation of the
coupling constraint at each iteration $t$. 
\begin{figure}[htbp]
\centering
  \includegraphics[scale=1]{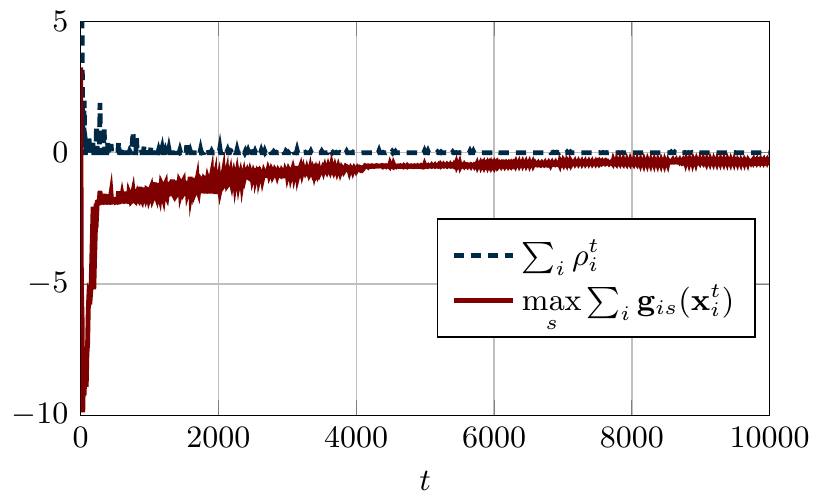}
  \caption{
    Evolution of the maximum violation of coupling constraints
    showing the feasibility of generated primal sequences (red).
    Asymptotically vanishing behavior of the sum of local violations (blue).}
  \label{fig:constraint_violation_RSDD}
\end{figure}

Finally, in Figure~\ref{fig:tracking} we show how $\sum_{j\in\nbrs_i} ( \blambda_{ij}^t - \blambda_{ji}^t )$
compares with the unknown part of the coupling constraint of each agent $i$, 
namely $\sum_{j\neq i} \bg_j  (\bx_j^t)$. Specifically, for all $i$, we plot the quantity
\begin{align*}
  \max_{s \in \until{S}}
  \Big(
    \smallsum_{h \neq i} \bg_{hs}(\bx_h^t) 
    -
    \smallsum_{j \in \nbrs_i} ( \blambda_{ij}^t - \blambda_{ji}^t )_s
  \Big).
\end{align*}
The picture highlights that $\sum_{j\in\nbrs_i} ( \blambda_{ij}^t - \blambda_{ji}^t )$ 
acts as a ``tracker'' of  the maximum of the contribution in the coupling constraint 
due to all the other agents $j\neq i$ in the network.
\begin{figure}[htbp]
  \centering
  \includegraphics[scale=1]{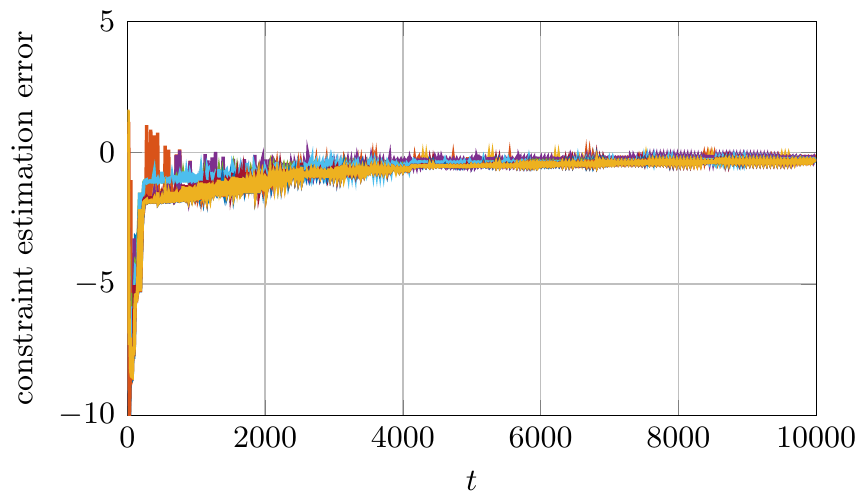}
  \caption{
    Evolution of the error between the unknown part of the coupling constraint and
    the local term $\sum_{j\in\nbrs_i} ( \blambda_{ij}^t - \blambda_{ji}^t )$, for all $i$,
    showing an asymptotic tracking property of the auxiliary variables.}
  \label{fig:tracking}
\end{figure}
 
\fi

\iftrue

\chapter{Constraint Exchange Methods}
\label{chap:constraint_exchange}

In this chapter, we present distributed optimization algorithms 
based on the exchange of constraints among agents. 
These algorithms are structurally different from the ones described in
Chapters~\ref{chap:primal} and~\ref{chap:dual}, since the information exchanged
by agents (encoding the local solution estimate) amounts to constraints rather
than decision variables.
We start by introducing the so-called \ConstrCons/ algorithm for
convex and abstract programs.\footnote{Abstract programs are a generalization of linear programs,
see, e.g., \cite{agarwal2000randomized,notarstefano2011distributed}.} Following the same approach as in
the previous chapter, we present and analyze the algorithm for a simplified optimization
set-up, namely linear programs, and then discuss how it in fact applies to
general convex and abstract programs. Then, we present other methods based
on the constraint exchange approach that generalize \ConstrCons/. 
Finally, we provide a numerical example to show the main
characteristics of the presented methods.

\section{Constraints Consensus applied to Linear Programs}
\label{ce:sec:LP}
In this section, we present and analyze a simplified version, applied to Linear
Programs (LPs), of the \ConstrCons/ algorithm~\cite{notarstefano2011distributed}.
First, we give some intuition on the algorithm together with its formal
description. Then, we provide a convergence analysis, and we briefly
mention a variant of the algorithm in which agents exchange ``columns'' instead
of constraints.

\subsection{Algorithm description}
Consider a network of $N$ agents that aim to solve the 
linear program
\begin{equation}
\begin{split}
  \min_\bx \: & \: c^\top \bx
  \\
  \subj \: & \: a_i^\top \bx \leq b_i, \hspace{0.5cm} i \in \until{N},
\end{split}
\label{ce:eq:LP}
\end{equation}
where $\bx \in \real^d$ is the optimization variable, $c \in \real^d$
is the cost vector, and $a_i \in \real^d$ and $b_i \in \real$, $i \in \until{N}$.
Notice that problem~\eqref{ce:eq:LP} is an instance of the common cost 
set-up described in Section~\ref{sec:setups_common_cost}.
For ease of presentation, we suppose that each agent $i$ knows only
the constraint $a_i^\top \bx \leq b_i$, and we say that this is the initial
constraint of agent $i$.
Also, we make the standing assumption that the number of agents
is greater than the dimension of the variable, i.e., $N > d$.

To convey the idea underlying the \ConstrCons/ algorithm, let us
elaborate on optimization problems in the form of~\eqref{ce:eq:LP}.
It is known from linear programming theory (cf. Appendix~\ref{sec:LP})
that the feasible set of problem~\eqref{ce:eq:LP} is polyhedral
and that, if $\bx^\star$ is an optimal vertex (i.e., an optimal solution attained
at a vertex of the feasible set), then there exists a \emph{basis},
consisting of exactly $d$ linearly independent inequality constraints
$a_{\ell_1}^\top \bx \le b_{\ell_1}, \ldots, a_{\ell_d}^\top \bx \le b_{\ell_d}$,
for some indices $\{\ell_1, \ldots, \ell_d\} \subseteq \until{N}$. Such a basis
allows for the computation of $\bx^\star$ as the (unique) optimal vertex
of the linear program
\begin{equation}
\begin{split}
  \min_\bx \: & \: c^\top \bx
  \\
  \subj \: & \: a_{\ell_h}^\top \bx \leq b_{\ell_h}, \hspace{0.5cm} h \in \until{d},
\end{split}
\label{ce:eq:LP_relaxed_basis}
\end{equation}
obtained as a relaxation of problem~\eqref{ce:eq:LP} by considering the
constraints in the basis only. 
Roughly speaking, in the \ConstrCons/ algorithm, each agent
iteratively solves a relaxation of problem~\eqref{ce:eq:LP}, with constraints
given by its initial constraint and constraints collected from neighbors,
and computes an optimal solution with its corresponding basis.
Then, the basis is broadcast to neighbors and the process is repeated
until convergence.
A natural question arising at this point is how to handle problems with
multiple optimal solutions. 
For such problems, in order to guarantee convergence of the scheme,
it is necessary for agents to select a common solution.
A possible approach to guarantee agent agreement is to employ
a lexicographic criterion (see Appendix~\ref{sec:LP} for a formal description),
i.e., agents compute the lexicographically minimal
optimal solution, termed \emph{lex-optimal solution}, of the LPs
through an appropriate local lexicographic solver.
Thus, in the remainder of this section, we will stick to the following
definition of basis.
\begin{definition}
\label{ce:def:basis_LP}
  Let $\bx^\star$ be the lex-optimal solution of a linear program in the
  form~\eqref{ce:eq:LP}.
  A collection of $d$ inequality constraints $a_{\ell_1}^\top \bx \le b_{\ell_1}, \ldots, a_{\ell_d}^\top \bx \le b_{\ell_d}$,
  for some indices $\{\ell_1, \ldots, \ell_d\} \subseteq \until{N}$,
  is called a \emph{basis} of~\eqref{ce:eq:LP} if $\bx^\star$ is the lex-optimal solution of
  \begin{align*}
    \min_\bx \: & \: c^\top \bx
    \\
    \subj \: & \: a_{\ell_h}^\top \bx \leq b_{\ell_h}, \hspace{0.5cm} h \in \until{d}.
\eqoprocend
  \end{align*}
\end{definition}
This definition is specifically tailored for linear programs. In fact, it can be
obtained as a special version of a more general definition of basis that holds
for so-called \emph{abstract programs}, see, e.g., \cite{agarwal2000randomized,notarstefano2011distributed}. 
From now on, we compactly denote a
basis as $(P, q)$, where $P \in \real^{d \times d}$ is the matrix obtained
by stacking the row vectors $a_{\ell_h}^\top$ and $q \in \real^d$ is the vector
obtained by stacking the scalars $b_{\ell_h}$, i.e.,
\begin{align*}
  P & = 
  \begin{bmatrix}
    a_{\ell_1}^\top \\ \vdots \\ a_{\ell_d}^\top
  \end{bmatrix},
  \hspace{0.5cm}
  q =
  \begin{bmatrix}
    b_{\ell_1} \\ \vdots \\ b_{\ell_d}
  \end{bmatrix}.
\end{align*}
Notice that, even if the lex-optimal solution of a LP is unique, there might be
several bases associated to the problem. In
Figure~\ref{ce:fig:LP_basis}, an example scenario in $\real^2$ is graphically
represented.
\begin{figure}[htpb]
	\centering
	\includegraphics[scale=1.1]{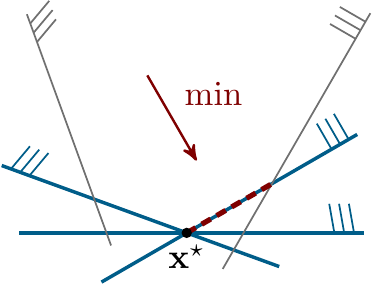}
	\caption{Example of an instance of problem~\eqref{ce:eq:LP}, where
	  the feasible side for each inequality constraint is denoted by three bars.
	  The LP admits several optimal solutions (indicated by a dashed red line),
	  but $\bx^\star$ is the lex-optimal solution. Notice that several bases can be chosen
	  (i.e., the constraint on which the optimal solutions lie together with either
	  of the other two constraints in blue).
	}
\label{ce:fig:LP_basis}
\end{figure}

Next, we describe the \ConstrCons/ algorithm applied
to problem~\eqref{ce:eq:LP}.
We assume that the agents communicate according to a jointly strongly connected
(time-varying) directed graph $\GG^t$ (cf. Section~\ref{sec:network_comm_models}),
and we denote by $\nbrs_i^t$ the in-neighbors of agent $i$ at communication round%
\footnote{
  In a synchronous algorithm the term iteration is more suited. Since the \ConstrCons/ 
  algorithm can be implemented also in an asynchronous setting, we prefer to use 
  this terminology.} $t$.
Each agent $i$ maintains a local solution estimate $\bx_i^t$ and a local basis $(P_i^t,q_i^t)$.
It is initialized to $(a_i^\top, b_i)$ and is incrementally filled as the agent
collects information from neighbors during the algorithm evolution.
At each communication round $t$, agent $i$ first gathers the bases from its
neighbors, then it constructs a (small) local LP with constraints given by the aggregation of:
\emph{(i)} the old basis, \emph{(ii)} the collected bases from neighbors, and
\emph{(iii)} its initial constraint.
Then, the agent finds a basis for the local LP to
update its state. Finally, the updated basis is broadcast to neighbors.
Notice that the local LP can be unbounded. Thus, an artificial (sufficiently large)
bounding box, denoted as $-M \1 \le \bx \le M \1$, with $M > 0$, is added to ensure
that the algorithm is well posed at each communication round, so that the bounding
box can becomes part of the local bases.
If $M$ is sufficiently large, the lex-optimal solution of problem~\eqref{ce:eq:LP}
is contained in the bounding box and the bounding box will eventually
leave the local bases.
Algorithm~\ref{ce:alg:constraints_consensus} formally summarizes the \ConstrCons/
algorithm applied to linear programs from the perspective of node $i$.
\begin{algorithm}
  \begin{algorithmic}[0]
  
    \Statex \textbf{Initialization}: $(P_i^0, q_i^0) = (a_i^\top, b_i)$
    \medskip

    \Statex \textbf{Evolution}: for $t=0,1,... $
    \smallskip
    
      \StatexIndent[0.5] \textbf{Gather} $(P_j^t, q_j^t)$ from neighbors $j \in \nbrs_i^t$

      \StatexIndent[0.5] \textbf{Compute} $\bx_i^{t+1}$ as the lex-optimal solution of
      \begin{align}
      \begin{split}
        \min_{\bx} \: & \: c^\top \bx
        \\
        \subj \: & \: a_i^\top \bx \le b_i
        \\
        & \: P_i^t \bx \le q_i^t
        \\
        & \: P_j^t \bx \le q_j^t, \hspace{0.5cm} j \in \nbrs_i^t
        \\
        & \: -M \1 \le \bx \le M \1
      \end{split}
      \label{ce:eq:alg_constr_cons_local_pb}
      \end{align}
      
      \StatexIndent[0.5] \textbf{Update} $(P_i^{t+1}, q_i^{t+1})$ as a basis of~\eqref{ce:eq:alg_constr_cons_local_pb}

  \end{algorithmic}
  \caption{\ConstrCons/ applied to LPs}
  \label{ce:alg:constraints_consensus}
\end{algorithm}

In Section~\ref{ce:sec:analysis}, we analyze the convergence of
Algorithm~\ref{ce:alg:constraints_consensus}.

Let us now
highlight the differences of the constraint exchange approach with respect
to the other approaches discussed in this survey.
First, note that in primal methods (cf. Chapter~\ref{chap:primal}),
consensus of the agents on a common optimal solution is enforced
by means of consensus iterations that steer the local quantities to
a common value, whereas in Algorithm~\ref{ce:alg:constraints_consensus},
consensus follows because eventually the lex-optimal solution
of the local problems~\eqref{ce:eq:alg_constr_cons_local_pb}
is the same.
Second, the communication network assumptions of constraint exchange
methods are generally very weak. For instance,
Algorithm~\ref{ce:alg:constraints_consensus} only
requires joint strong connectivity, which allows for an asynchronous
implementation of the algorithm (cf. Section~\ref{sec:network_comm_models}),
and allows for unreliable communication links
(e.g., subject to packet loss).
Also, it is worth mentioning that if the network consists of a large number of
agents with relatively small in-degree,
the local optimization problem~\eqref{ce:eq:alg_constr_cons_local_pb}
solved at each iteration is much smaller than the original
problem~\eqref{ce:eq:LP}, so that the algorithm scales nicely with the network size.
This is also corroborated by the fact that the communication is bounded:
each exchanged basis always consists of $d$ constraints, except in the
early stages of the algorithm in which less than $d$ constraints are available.
Finally, note that Algorithm~\ref{ce:alg:constraints_consensus} does
not require global tuning parameters (e.g., the step-size).

\subsection{Convergence Analysis}
\label{ce:sec:analysis}
In this subsection, we analyze the convergence of
Algorithm~\ref{ce:alg:constraints_consensus}.
The proof reported in this survey is different from the one in~\cite{notarstefano2011distributed},
which was devised for general abstract programs. Here we present a new
proof inspired by the arguments used in~\cite{testa2019distributed}.
Let us make the following assumption on problem~\eqref{ce:eq:LP}.
\begin{assumption}
\label{ce:ass:LP}
  Problem~\eqref{ce:eq:LP} is feasible and the lex-optimal solution exists.%
  \footnote{For a discussion on the existence of the lex-optimal solution, see Appendix~\ref{sec:LP}.}
  \oprocend
\end{assumption}

In the following, we prove that Algorithm~\ref{ce:alg:constraints_consensus} enjoys
finite-time convergence. The line of proof relies on three facts, namely
\emph{(i)} (finite-time) convergence of the solution estimates computed by each agent
(Lemma~\ref{ce:lemma:local_convergence}), \emph{(ii)} consensus of the
solution estimates at convergence (Lemma~\ref{ce:lemma:consensus}), \emph{(iii)}
optimality of the consensual solution estimates.

In the next lemma we prove that the quantities
computed by each agent converge in finite time.
\begin{lemma}[Local convergence]
  \label{ce:lemma:local_convergence}
  Let Assumption~\ref{ce:ass:LP} hold.
  Then, for all $i \in \until{N}$,
  \begin{enumerate}
    \item[\emph{(i)}] the cost sequence $\{c^\top \bx_i^t\}_{t \ge 0}$ is monotonically non-decreasing and
      converges in finite time, i.e., there exist $T_i > 0$ and $\bar{J}_i \in \real$
      such that
      \begin{align*}
        c^\top \bx_i^t = \bar{J}_i,
        \hspace{1cm}
        \text{for all } t \ge T_i;
      \end{align*}
      
    \item[\emph{(ii)}] the solution estimate sequence $\{\bx_i^t\}_{t \ge 0}$ converges in finite time to
      a vector satisfying the initial constraint of agent $i$, i.e., there exist
      $T_i^\prime > 0$ and $\bar{\bx}_i$ such that
      \begin{align*}
        &\bx_i^t = \bar{\bx}_i, 
        \hspace{1cm}
        \text{for all } t \ge T_i^\prime,
        \\
        &a_i^\top \bar{\bx}_i \le b_i.
      \end{align*}
  \end{enumerate}
\end{lemma}
\begin{proof}
  For the sake of analysis, let us denote by $J_i^t \triangleq c^\top \bx_i^t$ the
  cost associated to $\bx_i^t$.
  To prove \emph{(i)}, we consider problem~\eqref{ce:eq:alg_constr_cons_local_pb}
  at consecutive communication rounds, say $t$ and $t+1$.
  The lex-optimal solution of problem~\eqref{ce:eq:alg_constr_cons_local_pb}
  is $\bx_i^{t+1}$, with cost $J_i^{t+1} = c^\top \bx_i^{t+1}$ and
  $(P_i^{t+1}, q_i^{t+1})$ an associated basis.
  Thus, $\bx_i^{t+1}$ is the lex-optimal solution of
  \begin{align}
  \begin{split}
    \min_\bx \: & \: c^\top \bx
    \\
    \subj \: & \: P_i^{t+1} \bx \le q_i^{t+1},
  \end{split}
  \label{ce:eq:lemma_local_prob}
  \end{align}
  with optimal cost $J_i^{t+1}$.
  At the successive communication round $t+1$,
  the lex-optimal solution of the local problem~\eqref{ce:eq:alg_constr_cons_local_pb}
  does not violate any constraint of problem~\eqref{ce:eq:lemma_local_prob}.
  Thus, it holds $J_i^{t+2} \ge J_i^{t+1}$. Therefore, we conclude that the cost sequence is
  monotonically non-decreasing, i.e., for all $t \ge 0$,
  \begin{equation*}
    J_i^{t+1} \ge J_i^{t},
    \hspace{1cm}
    i \in \until{N}.
  \end{equation*}
  Also, because of the bounding box, the feasible set of
  problem~\eqref{ce:eq:alg_constr_cons_local_pb} is bounded, so that
  $\{ J_i^t \}_{t \ge 0}$ converges. %
  Finally, since there is a finite number of constraints in the network, $J_i^t$
  can only assume a finite number of values (corresponding to all the possible
  combinations of constraints). Thus, $\{ J_i^t \}_{t \ge 0}$ converges
  in finite time, i.e., there exist $T_i > 0$ and $\bar{J}_i \in \real$ such that
  \begin{align*}
    c^\top \bx_i^t = J_i^t = \bar{J}_i,
    \hspace{1cm}
    \text{for all } t \ge T_i.
  \end{align*}
  
  To prove \emph{(ii)}, let us consider the sequence of the first component
  of $\bx_i^t$ for $t \ge T_i$, i.e., $\{\bx_{i,1}^t\}_{t \ge T_i}$.
  First, notice that the cost associated to such sequence is identically equal to $\bar{J}_i$,
  i.e., $c^\top \bx_i^t = \bar{J}_i$ for all $t \ge T_i$.
  In the following, we apply ideas similar to \emph{(i)}, namely
  we consider problem~\eqref{ce:eq:alg_constr_cons_local_pb}
  at consecutive communication rounds, say $t$ and $t+1$, with $t \ge T_i$.
  The lex-optimal solution of problem~\eqref{ce:eq:alg_constr_cons_local_pb}
  is $\bx_i^{t+1}$, with first component $\bx_{i,1}^{t+1}$ and
  $(P_i^{t+1}, q_i^{t+1})$ an associated basis.
  Thus, $\bx_i^{t+1}$ is the lex-optimal solution of
  \begin{align}
  \begin{split}
    \min_\bx \: & \: c^\top \bx
    \\
    \subj \: & \: P_i^{t+1} \bx \le q_i^{t+1}.
  \end{split}
  \label{ce:eq:lemma_local_prob_x}
  \end{align}
  At the successive communication round $t+1$, the optimal cost stays
  equal to $\bar{J}_i$ and
  the lex-optimal solution of the local problem~\eqref{ce:eq:alg_constr_cons_local_pb}
  does not violate any constraint of problem~\eqref{ce:eq:lemma_local_prob_x}.
  Thus, since the local lexicographic solver selects the optimal solution
  with minimal first component, it follows that $\bx_{i,1}^{t+2} \ge \bx_{i,1}^{t+1}$.
  Therefore, we conclude that the sequence $\{\bx_{i,1}^t\}_{t \ge T_i}$
  is monotonically non-decreasing, i.e., for all $t \ge T_i$,
  \begin{equation*}
    \bx_{i,1}^{t+1} \ge \bx_{i,1}^{t},
    \hspace{1cm}
    i \in \until{N}.
  \end{equation*}
  Also, because of the bounding box, the feasible set of
  problem~\eqref{ce:eq:alg_constr_cons_local_pb} is bounded, so that
  $\{ \bx_{i,1}^t \}_{t \ge 0}$ converges. %
  Finally, since there is a finite number of constraints in the network, $\bx_{i,1}^t$
  can only assume a finite number of values (corresponding to all the possible
  combinations of constraints). Thus, $\{ \bx_{i,1}^t \}_{t \ge 0}$ converges
  in finite time, i.e., there exist $T_i^\prime > 0$ and $\bar{\bx}_{i,1} \in \real$ such that
  \begin{align*}
    \bx_{i,1}^t = \bar{\bx}_{i,1},
    \hspace{1cm}
    \text{for all } t \ge T_i^\prime.
  \end{align*}
  By repeating the same arguments for each of the subsequent
  components of $\bx_i^t$ for $t \ge T_i^\prime$, we are able to
  conclude that $\{ \bx_i^t \}_{t \ge 0}$ converges in finite time to some $\bar{\bx}_i$,
  which by construction satisfies $a_i^\top \bar{\bx}_i \le b_i$.
\end{proof}

In the following lemma, we prove that the solution estimates to which
agents converge are consensual.
\begin{lemma}[Consensus]
  \label{ce:lemma:consensus}
  Let the communication graph be jointly strongly connected.
  Moreover, assume that the sequences computed by agents have converged, i.e.,
  there exists $T_0 > 0$ such that for all $i \in \until{N}$ it holds
  \begin{align*}
    c^\top \bx_i^t = \bar{J}_i \: \text{ and } \: \bx_i^t = \bar{\bx}_i,
    \hspace{1cm}
    \text{for all } t \ge T_0,
  \end{align*}
  for some $\bar{J}_i \in \real$ and $\bar{\bx}_i \in \real^d$.
  Then, it holds
  \begin{align*}
    \bar{J}_i = \bar{J}_j \: \text{ and } \: \bar{\bx}_i = \bar{\bx}_j,
    \hspace{1cm}
    \text{for all } i,j \in \until{N}.
  \end{align*}
\end{lemma}
\begin{proof}
  For the sake of analysis, let us denote by $J_i^t \triangleq c^\top \bx_i^t$ the
  cost associated to $\bx_i^t$.
  By contradiction, assume that there exist two different agents $i$ and $j$ such that
  $\bar{J}_i \ne \bar{J}_j$. Without loss of generality, let $\bar{J}_j > \bar{J}_i$.
  
  By finite-time convergence of the cost sequences, there exists $T_0 > 0$ such that
  $J_j^t = \bar{J}_j > \bar{J}_i = J_i^t$ for all $t \ge T_0$.
  Moreover, since the communication graph is jointly strongly connected,
  for all $t \ge T_0$ and each pair of agents $(i,j)$, there exists a
  sequence of time instants $\tau_1, \ldots, \tau_{k}$, with
  $t \leq \tau_1 < \ldots < \tau_{k}$, and a sequence of nodes
  $\nu_1, \ldots, \nu_{k-1}$, such that the directed edges
  $(j,\nu_1),(\nu_1,\nu_2),\ldots, (\nu_{k-1},i)$ belong to the digraph at times
  $\tau_1, \ldots, \tau_{k}$ (cf. \cite{notarstefano2011distributed}).
  
  At communication round $\tau_1$, agent $\nu_1$ computes 
  $\bx_{\nu_1}^{\tau_1 + 1}$ by minimizing
  $c^\top \bx$ over a subset of the basis associated to $\bx_j^{\tau_1}$
  (by construction), so that $J_{\nu_1}^{\tau_1+1} \geq J_j^{\tau_1}$.
  Similarly, at communication round $\tau_2$, agent $\nu_2$
  computes $\bx_{\nu_2}^{\tau_2 + 1}$ by minimizing
  $c^\top \bx$ over a subset of the basis associated to $\bx_{\nu_1}^{\tau_2}$.
  Thus, it holds
  \begin{align*}
    J_{\nu_2}^{\tau_2+1} \geq J_{\nu_1}^{\tau_2}.
  \end{align*}  
  Since the cost sequences have converged, 
  it follows that $\bar{J}_{\nu_1} = J_{\nu_1}^{\tau_2} = J_{\nu_1}^{\tau_1+1}$.
  Thus, it holds
  \begin{align*}
    J_{\nu_2}^{\tau_2+1} \ge J_j^{\tau_1}.
  \end{align*}
  The argument can be iterated to conclude that $J_i^{\tau_k+1} \geq J_j^{\tau_1}$.
  Therefore, for all $t > T_0$ there exists $\theta_{ij} > 0$ such that
  \begin{align*}
    \bar{J}_i = J_i^{t+\theta_{ij}} \geq J_j^t = \bar{J}_j,
  \end{align*}
  contradicting the assumption $\bar{J}_j > \bar{J}_i$.
  Thus, $\bar{J}_1 = \ldots = \bar{J}_N$, which concludes the first part of the proof.
  To prove consensus of the solutions, we note that for all $t \ge T_0$,
  $c^\top \bx_1^t = \ldots = c^\top \bx_N^t$.
  Then, it is possible to apply arguments similar to the first part
  to each component of the solution vector (in lexicographic order,
  see proof of Lemma~\ref{ce:lemma:local_convergence} \emph{(ii)}).
\end{proof}
With Lemma~\ref{ce:lemma:local_convergence} and
Lemma~\ref{ce:lemma:consensus} at reach, we are now ready to prove
the convergence of Algorithm~\ref{ce:alg:constraints_consensus}.

\begin{theorem}
\label{ce:thm:convergence}
  Let Assumption~\ref{ce:ass:LP} hold and let the communication graph
  be jointly strongly connected. Moreover, let $\bx^\star$ be the
  lex-optimal solution of problem~\eqref{ce:eq:LP} and assume $M > 0$
  is sufficiently large.
  Consider the sequences $\{ \bx_i^t \}_{t \ge 0}, i \in \until{N}$, generated
  by Algorithm~\ref{ce:alg:constraints_consensus}.
  Then, for all $i \in \until{N}$, the following holds:
  \begin{enumerate}
    \item the cost sequence $\{ c^\top \bx_i^t \}_{t \ge 0}$
      converges in finite time to the optimal cost $J^\star$ of~\eqref{ce:eq:LP};
    \item the solution sequence $\{ \bx_i^t \}_{t \ge 0}$ converges in finite time to $\bx^\star$.
  \end{enumerate}
\end{theorem}
\begin{proof}
  For the sake of analysis, let us denote by $J_i^t \triangleq c^\top \bx_i^t$ the
  cost associated to $\bx_i^t$.
  By Lemma~\ref{ce:lemma:local_convergence}, the cost sequences $\{J_i^t\}_{t \ge 0}$ and the
  solution sequences $\{\bx_i^t\}_{t \ge 0}$ converge in finite time to $\bar{J}_i$ and $\bar{\bx}_i$
  respectively, and by construction it holds
  \begin{align*}
    a_i^\top \bar{\bx}_i \le b_i,
    \hspace{1cm}
    \text{for all } i \in \until{N}.
  \end{align*}
  By Lemma~\ref{ce:lemma:consensus}, there exist a common scalar $\bar{J} \in \real$
  and a common vector $\bar{\bx}$ such that $\bar{J}_i = \bar{J}$ and
  $\bar{\bx}_i = \bar{\bx}$ for all $i \in \until{N}$. Therefore, $\bar{\bx}$ is
  feasible for problem~\eqref{ce:eq:LP}, since $a_i^\top \bar{\bx} \le b_i$ for all $i$.
  To prove that $\bar{J} = J^\star$, we first note that $\bar{J} \le J^\star$, since each agent
  builds up the local LP as a relaxation (i.e., with a lower number of constraints) of the original
  problem~\eqref{ce:eq:LP}, and the bounding box is sufficiently large
  (thus, we can assume that $M > \|\bx^\star\|_\infty$).
  On the other hand, since $\bar{\bx}$ is feasible for problem~\eqref{ce:eq:LP},
  then $J^\star \le c^\top \bar{\bx} = \bar{J}$,
  thus implying $\bar{J} = J^\star$.
  
  Since we have shown that $\bar{\bx}$ is feasible and cost-optimal,
  so that $\bar{\bx}$ is an optimal solution of~\eqref{ce:eq:LP}, we only have
  to show that it is the lexicographic minimum among all the minima (i.e., $\bar{\bx} = \bx^\star$).
  By contradiction, suppose it is not. Then, $\bx^\star \lexprec \bar{\bx}$, where
  the symbol $\lexprec$ means that $\bx^\star$ is lexicographically smaller
  than $\bar{\bx}$ (cf. Appendix~\ref{sec:LP}).
  Now, since $\bar{\bx}$ is
  computed by each agent as the lex-optimal solution to the local problem,
  there exists a basis $(\bar{P}, \bar{q})$, made up of constraints of problem~\eqref{ce:eq:LP},
  such that $\bar{\bx}$ is the lex-optimal solution to
  \begin{align}
  \begin{split}
    \min_\bx \: & \: c^\top \bx
    \\
    \subj \: & \: \bar{p}_h^\top \bx \le \bar{q}_h, \hspace{0.5cm} h \in \until{d},
  \end{split}
  \label{ce:eq:proof_LP_basis}
  \end{align}
  where $\bar{p}_h^\top \in \real^{1 \times d}$ denotes the $h$-th row of
  $\bar{P}$ and $\bar{q}_h \in \real$ denotes the $h$-th entry of $\bar{q}$.
  But this means that $\bx^\star$ must be infeasible for problem~\eqref{ce:eq:proof_LP_basis},
  otherwise the lex-optimal solution of~\eqref{ce:eq:proof_LP_basis}
  would be $\bx^\star$ instead of $\bar{\bx}$.
  Therefore, one of the constraints in~\eqref{ce:eq:proof_LP_basis} is violated
  by $\bx^\star$, i.e., there exists $h \in \until{d}$ such that $P_h^\top \bx^\star > q_h$.
  But since the constraints in~\eqref{ce:eq:proof_LP_basis} are drawn from problem~\eqref{ce:eq:LP},
  this contradicts the fact that $\bx^\star$ is feasible for the original LP~\eqref{ce:eq:LP}.
  Thus, $\bar{\bx} = \bx^\star$ and the proof follows.
\end{proof}

A few remarks on the \ConstrCons/ algorithm are in order.
In the algorithm analysis we did not prove that
the local bases are consensual at convergence. Indeed, agents may
compute different bases associated to the lex-optimal solution.
A sufficient condition for consensus of bases is the so-called \emph{non-degeneracy}
of problem~\eqref{ce:eq:LP} (see also~\cite{notarstefano2011distributed}).
Finally, a remarkable property of the algorithm is that a fully distributed
halting condition can be obtained. Indeed, if the communication graph
is fixed, each agent can halt the execution of the algorithm as soon as the
locally computed solution stays constant for $2 \diam(\GG) + 1$
communication rounds \cite[Theorem IV.4]{notarstefano2011distributed}.
If the communication graph is time-varying and $T$-strongly connected
(cf. Section~\ref{sec:network_comm_models}),
it can be seen that each agent can halt the execution of the algorithm as soon as the
locally computed solution stays constant for $2NT + 1$ communication
rounds.

\subsection{Distributed Simplex}
\label{ce:sec:distributed_simplex}
In this section, we briefly mention a variant of the \ConstrCons/ algorithm
applied to LPs, namely the \DistrSimplex/ algorithm \cite{burger2012distributed}.
We consider a network of $N$ agents that aim to cooperatively solve
linear programs in the so-called standard form, i.e.,
\begin{align}
\begin{split}
  \min_{\bx} \: & \: c^\top \bx
  \\
  \subj \: & \: A \bx = b,
  \\
  & \: \bx \ge 0,
\end{split}
\label{ce:eq:LP_standard_form}
\end{align}
where $A \in \real^{d\times N}$, $b \in \real^d$ and $c \in \real^N$ are the
problem data and $\bx \in \real^N$ is the decision vector.
A \emph{column} of problem~\eqref{ce:eq:LP_standard_form} is defined as the vector
\begin{align*}
  h_i
  \triangleq
  \begin{bmatrix}
    c_i
    \\
    a_i
  \end{bmatrix} \in \real^{1+d},
\end{align*}
where $c_i \in \real$ is the $i$-th entry of the vector $c$ and $a_i \in \real^N$
is the $i$-th column of the matrix $A$.

From a centralized perspective, in the
classical simplex method, a set of columns (which for problems in standard form are treated as a basis),
is iteratively updated until an optimal solution of problem~\eqref{ce:eq:LP_standard_form} is found.
At each iteration, a \emph{leaving} column exits
the basis and is replaced by an \emph{entering} column.
The \DistrSimplex/ algorithm extends the (centralized) simplex method.
Agents are assumed to initially know only a
subset of the problem columns. Informally, at every communication round,
each agent builds up a (small) local LP with a subset of the problem columns
(namely, the old basis and the bases collected from neighbors).
Then, the local LP is solved, a basis associated to the optimal solution is found
and is sent to neighbors.
It can be shown that the evolution of the \DistrSimplex/ algorithm applied to
problem~\eqref{ce:eq:LP_standard_form} is tightly linked to the evolution of the
\ConstrCons/ algorithm applied to the dual of
problem~\eqref{ce:eq:LP_standard_form} (see \cite[Proposition
5.3]{burger2012distributed}).

\section{Constraints Consensus for Convex and Abstract Programs}
\label{ce:sec:CP}
In this section, we describe the \ConstrCons/ algorithm for more general
set-ups than problem~\eqref{ce:eq:LP}.
Formally, assume $N$ agents aim to cooperatively solve the convex program
\begin{equation}
\label{ce:eq:convex_problem}
\begin{split}
  \min_{\bx} \: & \: c^\top \bx
  \\
  \subj \: & \: \bx \in \bigcap_{i=1}^N X_i,
\end{split}
\end{equation}
where $c \in \real^d$ is the cost vector and the sets $X_i$ are subsets
of $\real^d$, for all $i \in \until{N}$.
Problem~\eqref{ce:eq:convex_problem} is in the common-cost
form (cf. Section~\ref{sec:setups_common_cost}), and we suppose that, for all $i$,
the set $X_i$ is known by agent $i$ only and that the cost vector $c$ is globally
known.
Notice that the linear cost function in problem~\eqref{ce:eq:convex_problem}
results in no loss of generality, as discussed in Remark~\ref{ce:remark:nonlinear_cost}.
We make the following assumption.
\begin{assumption}
\label{ce:ass:CP}
  Problem~\eqref{ce:eq:convex_problem} is feasible and the sets $X_i$ are
  convex and compact, for all $i \in \until{N}$.
  \oprocend
\end{assumption}

The \ConstrCons/ algorithm applied to problem~\eqref{ce:eq:convex_problem}
can be formalized by extending the concept of basis (cf. Definition~\ref{ce:def:basis_LP})
so as to consider the (possible) nonlinear nature of the local constraints $X_i$.
Formally, let $\bx^\star$ be the lex-optimal solution of problem~\eqref{ce:eq:convex_problem}.
Then, the collection of $\delta$ constraints $X_{\ell_1}, \ldots, X_{\ell_\delta}$, for some indices
$\{\ell_1, \ldots, \ell_\delta\} \subseteq \until{N}$, are a basis of
\eqref{ce:eq:convex_problem} if $\bx^\star$ is the lex-optimal solution of
\begin{align*}
\begin{split}
  \min_\bx \: & \: c^\top \bx
  \\
  \subj \: & \: \bx \in X_{\ell_h}, \hspace{0.5cm} h \in \until{\delta},
\end{split}
\end{align*}
and if the collection of $\delta$ constraints is minimal (i.e., removing a constraint from
the previous problem implies that the lex-optimal solution changes).
We compactly denote the basis as the set $B = \bigcap_{h=1}^\delta X_{\ell_h}$.
For feasible convex problems in the form~\eqref{ce:eq:convex_problem}, it holds
$\delta \le d$, whereas for linear programs, it holds $\delta = d$ (cf.
Definition~\ref{ce:def:basis_LP}).
The maximum $\delta$ for a given problem is called the \emph{combinatorial
  dimension} of the problem.
A more comprehensive discussion can be found
in~\cite{agarwal2000randomized,notarstefano2011distributed}.

Next, we describe the \ConstrCons/ algorithm applied to convex programs
in Algorithm~\ref{ce:alg:constraints_consensus_convex}, from the perspective
of node $i$. Each agent $i$ maintains a local solution estimate $\bx_i^t$ and
a local basis $B_i^t$, initialized to $X_i$.
The algorithm looks similar to Algorithm~\ref{ce:alg:constraints_consensus},
where the main difference is that general convex constraints are considered,
instead of linear ones. 

\begin{algorithm}
  \begin{algorithmic}[0]
  
    \Statex \textbf{Initialization}: $B_i^0 = X_i$
    \medskip

    \Statex \textbf{Evolution}: for $t=0,1,... $
    \smallskip
    
      \StatexIndent[0.5] \textbf{Gather} $B_j^t$ from neighbors $j \in \nbrs_i^t$

      \StatexIndent[0.5] \textbf{Compute} $\bx_i^{t+1}$ as the lex-optimal solution of
      \begin{align}
      \begin{split}
        \min_{\bx} \: & \: c^\top \bx
        \\
        \subj \: & \: \bx \in X_i
        \\
        & \: \bx \in B_i^t
        \\
        & \: \bx \in B_j^t, \hspace{0.5cm} j \in \nbrs_i^t
      \end{split}
      \label{ce:eq:local_prob_convex}
      \end{align}
      
      \StatexIndent[0.5] \textbf{Update} $B_i^{t+1}$ as a basis of~\eqref{ce:eq:local_prob_convex}

  \end{algorithmic}
  \caption{Constraints Consensus applied to convex problems}
  \label{ce:alg:constraints_consensus_convex}
\end{algorithm}

Note that, as in Algorithm~\ref{ce:alg:constraints_consensus}, we ask
processors to use a lexicographic solver to handle possible non-uniqueness
of the optimal solution.
Algorithm~\ref{ce:alg:constraints_consensus_convex} enjoys the same
convergence properties of Algorithm~\ref{ce:alg:constraints_consensus},
formalized next.
\begin{theorem}
  \label{ce:thm:convergence_convex}
  Let Assumption~\ref{ce:ass:CP} hold and let the communication graph be jointly
  strongly connected.
  Moreover, let $\bx^\star$ be the lex-optimal solution of problem~\eqref{ce:eq:convex_problem}.
  Consider the sequences $\{ \bx_i^t \}_{t \ge 0}, i \in \until{N}$, generated
  by Algorithm~\ref{ce:alg:constraints_consensus_convex}.
  Then, for all $i \in \until{N}$, the following holds:
  \begin{enumerate}
    \item the cost sequence $\{ c^\top \bx_i^t \}_{t \ge 0}$
      converges in finite time to the optimal cost $J^\star$ of~\eqref{ce:eq:convex_problem};
    \item the solution sequence $\{ \bx_i^t \}_{t \ge 0}$ converges in finite time to $\bx^\star$.~\oprocend
  \end{enumerate}
\end{theorem}
Theorem~\ref{ce:thm:convergence_convex} can be proven by using arguments similar
to the ones in Theorem~\ref{ce:thm:convergence}, thus we omit the proof.

We highlight that, in practice, Algorithm~\ref{ce:alg:constraints_consensus_convex}
can be implemented when
the constraint sets $X_i$ are easy to communicate (e.g., when all of them
have the same parametric form and they only differ for the parameters).
In more difficult set-ups, polyhedral approximations of the local sets $X_i$
can be communicated instead (cf. Section~\ref{ce:sec:cutting_plane_consensus}).

\begin{remark}
\label{ce:remark:nonlinear_cost}
Algorithm~\ref{ce:alg:constraints_consensus_convex} can be properly adapted
to handle problems with nonlinear cost in the form
\begin{equation}
\begin{split}
  \min_{\bx} \: & \: f(\bx)
  \\
  \subj \: & \: \bx \in \bigcap_{i=1}^N X_i,
\end{split}
\label{ce:eq:convex_problem_nonlinearcost}
\end{equation}
with $\map{f}{\real^d}{\real}$ a convex cost function. By
resorting to the epigraph form of \eqref{ce:eq:convex_problem_nonlinearcost},
which is in the form~\eqref{ce:eq:convex_problem}, it can be shown that
Algorithm~\ref{ce:alg:constraints_consensus_convex} can be implemented by simply
replacing the linear function in the local problem~\eqref{ce:eq:local_prob_convex}
with the nonlinear one and by increasing the
maximum number of sets in the bases to $d+1$.  \oprocend
\end{remark}

The \ConstrCons/ algorithm can handle more general problems than~\eqref{ce:eq:convex_problem}.
Indeed, in \cite{notarstefano2011distributed}, the algorithm has been formulated
for general abstract programs (or LP-type problems),
which include, as a special case, problems~\eqref{ce:eq:LP}
and~\eqref{ce:eq:convex_problem}.
We do not give the technical details of abstract programs, but we only mention
that they are a generalization of linear programs, which capture numerous
geometric optimization problems such as, e.g., computation of the smallest enclosing
ball of a set of points.
When the combinatorial dimension of the problem is known, the distributed
algorithm \cite{notarstefano2011distributed} can be applied directly.
Otherwise, if the \emph{Helly number} of the problem is known,
one can use the results in \cite{amenta1994helly} to compute the
combinatorial dimension of the problem.

\section{Extensions}
\label{ce:sec:extensions}
In this section, we discuss extensions of the Constraints
Consensus algorithm.

\subsection{Cutting-plane Consensus}
\label{ce:sec:cutting_plane_consensus}
Let us consider again the convex program~\eqref{ce:eq:convex_problem}.
The \emph{Cutting-plane Consensus} algorithm \cite{burger2014polyhedral}
is an extension of Algorithm~\ref{ce:alg:constraints_consensus_convex},
in which outer approximations of the local constraint sets $X_i$ are communicated
(instead of the sets $X_i$ themselves).
There are several situations in which this approach is desirable, such as, e.g.,
\emph{(i)} when privacy must be preserved (so that agents do not want to
share their own constraint with the other nodes), \emph{(ii)} when it is expensive
to send $X_i$, \emph{(iii)} when there are infinitely many local constraints
(e.g., robust, semi-infinite programming).

The Cutting-plane Consensus algorithm is based on a successive refinement of
polyhedral approximations of the local sets $X_i$. In particular, agents repeatedly
solve linear programs of the form
\begin{align}
\begin{split}
  \min_\bx \: & \: c^\top \bx
  \\
  \subj \: & \: A \bx \le b,
\end{split}
\label{ce:eq:CPC_approximate_LP}
\end{align}
where the feasible set $\{\bx \in \real^d \mid A \bx \le b\}$ is a
(polyhedral) outer-approximation of $\bigcap_{i=1}^N X_i$.
It is constructed by generating and exchanging a particular
type of constraints, called \emph{cutting planes}.\footnote{A cutting plane is a
half space $h \triangleq \{ \bx \in \real^d \mid a^\top \bx \le b \}$
separating a query point $\bx_q \in \real^d$ from a set $X$, i.e., such that
$X \subset h$ and $\bx_q \notin h$.}

The evolution of the Cutting-plane Consensus algorithm can be roughly
summarized as follows.
Each agent $i$ first solves problem~\eqref{ce:eq:CPC_approximate_LP}
and finds an optimal solution $\bx_q$. Then, it checks whether
$\bx_q$ belongs to its own constraint set $X_i$. If so, it sends to neighbors
a basis associated to $\bx_q$ (in terms of the approximated constraints). If not, it generates
a new cutting plane, it computes an optimal solution of the new approximate
problem and sends a basis to neighbors.

Differently from the \ConstrCons/ algorithm, the Cutting-plane
Consensus algorithm does not enjoy finite-time convergence, but instead it
converges asymptotically. Also, we point out that the tie-break rule
used in \cite{burger2014polyhedral} (in case problem~\eqref{ce:eq:CPC_approximate_LP}
has multiple optimal solutions) consists of the minimal 2-norm solution,
instead of the lex-optimal solution.

\subsection{Distributed Mixed-Integer Linear Programming via Cut Generation and Constraint Exchange}
\label{eq:ce:sec:milp}
Mixed-integer linear programs (MILPs) are linear programs in which some
of the variables are constrained to be integer, i.e.,
\begin{align}
\begin{split}
  \min_\bx \: & \: c^\top \bx
  \\
  \subj \: & \: a_i^\top \bx \le b_i, \hspace{1cm} i \in \until{N},
  \\
  & \: \bx \in \integer^{d_Z} \times \real^{d_R},
\end{split}
\label{ce:eq:MILP}
\end{align}
where $d_Z$ and $d_R$ are the dimensions of the integer and real variables,
$d = d_Z +d_R$, $c \in \real^d$ and $a_i \in \real^d$, $b_i \in \real$ for all $i \in \until{N}$.

It is well known that MILPs are NP-hard problems, which makes problem~\eqref{ce:eq:MILP}
difficult to solve. In \cite{testa2017finite} and in
\cite{testa2019distributed} distributed algorithms are proposed, with
finite-time convergence, for the solution of problem~\eqref{ce:eq:MILP}. They
are based on a constraint exchange approach as in \ConstrCons/, but
  appropriate additional constraints (cutting planes, cf. also
  Section~\ref{ce:sec:cutting_plane_consensus}) are generated throughout the
  algorithm evolution.

Let $P \triangleq \{ \bx \in \real^d \mid a_i^\top \bx \le b_i \text{ for all } i\}$ denote
the polyhedron described by the inequality constraints of problem~\eqref{ce:eq:MILP}
and let $P_I \triangleq P \cap (\integer^{d_Z} \times \real^{d_R})$ denote the feasible set
of problem~\eqref{ce:eq:MILP}. An important feature of problem~\eqref{ce:eq:MILP} is that
it has the same optimal cost of the linear program
\begin{align}
\begin{split}
  \min_\bx \: & \: c^\top \bx
  \\
  \subj \: & \: \bx \in \text{conv}(P_I),
\end{split}
\label{ce:eq:MILP_conv}
\end{align}
where $\text{conv}(P_I)$ is the convex hull of $P_I$. Moreover, the optimal solution
set of problem~\eqref{ce:eq:MILP} is contained in the optimal solution set
of~\eqref{ce:eq:MILP_conv}.
In order to solve the original MILP~\eqref{ce:eq:MILP}, the algorithms in
\cite{testa2019distributed} produce successive approximations of $\text{conv}(P_I)$
by generating two types of cutting planes:
\emph{(i)} mixed-integer Gomory cuts and \emph{(ii)} cost-based cuts.
We do not provide the technical details on the algorithms, but we only point out
that, as in \ConstrCons/, the algorithms work under asynchronous and
unreliable communication and enjoy finite-time convergence.

\subsection{Other extensions}

In this subsection, we briefly mention other extensions of the algorithms
presented in this chapter.

Robust optimization is the field of optimization that considers
problems in which the problem data is uncertain. Typical approaches to tackle
an uncertain problem consider the worst case of the uncertain
parameters, giving rise to a \emph{semi-infinite} optimization problem, i.e., with an
infinite number of constraints.
In~\cite{chamanbaz2017randomized}, a distributed robust
optimization algorithm is proposed, which is a \emph{randomized} extension of the
\ConstrCons/ algorithm, to solve linear programs where the problem data is
subject to uncertainty.
The algorithm relies on a verification step (based on a random sampling of each
agent of its local uncertain constraint set), and on the deterministic solution
of a local version of the global semi-infinite problem.

In \cite{notarstefano2015core}, the authors considered a \emph{big-data} quadratic
programming set-up emerging in several learning problems for cyber-physical
networks, where the big-data keyword is due to the very high dimension of the
optimization variable and of the training samples.
For this class of big-data quadratic optimization problems, they proposed a
distributed algorithm, obtained as an extension of the \ConstrCons/ algorithm,
which solves the problem up to an arbitrary tolerance $\epsilon$.
The algorithm is based on the notion of core-set used in geometric optimization to
approximate the value function of a given set of points with a smaller subset of points.
From an optimization point of view, a subset of active constraints is identified,
whose number depends on the tolerance $\epsilon$. The resulting approximate solution
is such that an $\epsilon$-relaxation of the constraints guarantees no constraint violation.

Submodular optimization is a special class of combinatorial optimization
(in which the cost function is actually a \emph{set function}) arising in
several machine learning problems, but also in cooperative control of complex systems.
In~\cite{testa2018submodular}, a submodular minimization
problem is considered. Agents can evaluate the cost function only for those sets
including the agent itself. Then, by relying on a proper linear programming
reformulation of the submodular problem (involving a huge number of variables),
it is possible to devise a distributed algorithm based on a \emph{column generation}
approach, in which columns are generated through a local \emph{greedy} algorithm.

\section{Numerical Example}
\label{ce:sec:simulations}
In this section, we provide a numerical example of the \ConstrCons/
algorithm to highlight its main features. 

We consider a network of $N=30$ agents communicating over a fixed, directed,
strongly connected graph generated according to an Erd\H{o}s-R\'enyi random
model with parameter $p = 0.1$.

We focus on the soft-margin SVM problem introduced in Section~\ref{sec:setups_SVM},
where we consider a two-dimensional space, i.e., $d = 2$,
and we recall that each agent $i$ is assigned one
training sample $(p_i, \ell_i) \in \real^2 \times \{-1, 1\}$.
We suppose that the training samples are randomly picked from
two bivariate gaussian distributions with covariance matrix
equal to the identity matrix. A number of $15$ agents are assigned to the first
distribution, which has zero mean and is associated to the label $\ell_i = 1$,
while the remaining agents are assigned to the second distribution,
associated to the label $\ell_i = -1$ and with mean equal to $[3,2]^\top$.

The goal for agents is to agree on an optimal solution of
problem~\eqref{eq:soft_svm_problem}, which we recall here
\begin{align*}
\begin{split}
  \min_{w, b, \bxi} \:  & \: \frac{1}{2} w^\top w + C \smallsum_{i=1}^N \xi_i
  \\
  \subj \: & \: \ell_i ( w^\top p_i + b ) \ge 1 - \xi_i, \hspace{0.5cm} i \in\until{N},
  \\
  & \: \bxi \ge 0,
\end{split}
\end{align*}
where the parameter $C$ is set to $100$. In the following
we also denote the vector stacking all the optimization variables with $\bx$.
As discussed in Remark~\ref{ce:remark:nonlinear_cost}, in order to solve
problem~\eqref{eq:soft_svm_problem} with the \ConstrCons/
algorithm, we implement the local optimization problems in
Algorithm~\ref{ce:alg:constraints_consensus_convex} with the cost
function $f(\bx) = 1/2 w^\top w + C \sum_{i=1}^N \xi_i$ and we allow up to
$d+1$ constraints in the bases.
To solve the $\lexmin$ optimization in~\eqref{ce:alg:constraints_consensus_convex},
we solve a total of $d+1$ problems as follows. First, we obtain the optimal cost $f^\star$ of
the problem. Then we add to the problem the constraint $f(x) = f^\star$ (in order to force the optimal cost)
and we minimize the first component of the decision variable. We continue this procedure until we obtain 
the lex-optimal solution.
Moreover, artificial box constraints $-M \1 \le w, b, \bxi \le M\1$,
with $M = 10$ (which we verified to be sufficiently large for this problem),
are added to problem~\eqref{eq:soft_svm_problem}
in order to satisfy Assumption~\ref{ce:ass:CP}.

In our simulation, agents reached consensus on the lex-optimal solution
of problem~\eqref{eq:soft_svm_problem} in $10$ communication rounds,
as expected from the finite-time result of Theorem~\ref{ce:thm:convergence_convex}.
In Figure~\ref{ce:fig_cost_error} we show the convergence rate of
Algorithm~\ref{ce:alg:constraints_consensus_convex}.
In particular, we plot the difference between the cost of the solution estimates
and the optimal cost $J^\star$ of problem~\eqref{eq:soft_svm_problem}, i.e.,
$f(\bx_i^t) - J^\star$, for all $i$. Note that all the lines eventually approach
zero.
\begin{figure}[!htpb]
\centering
  \includegraphics[scale=1]{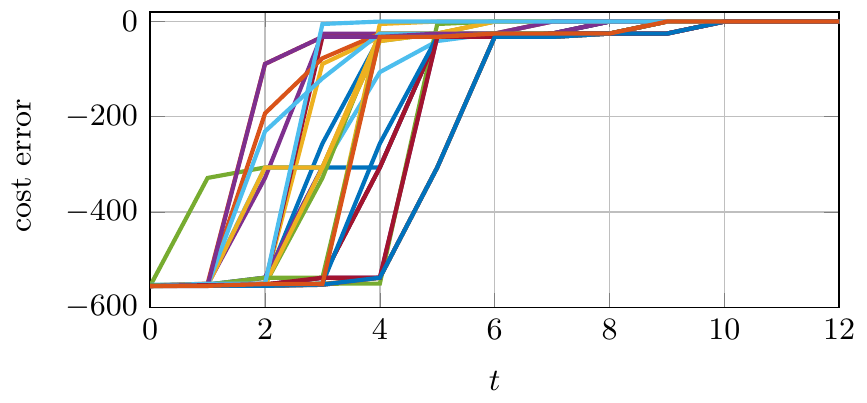}
\caption{
  Evolution of the cost error $f(\bx_i^t) - J^\star$ of local solution estimates $\bx_i^t$
  for the \ConstrCons/ algorithm. Each line refers to an agent in the network.}
  \label{ce:fig_cost_error}
\end{figure}

In Figure~\ref{ce:fig_max_constraint} we show the maximum constraint value
associated to the local solution estimates, i.e.,
for all $i$ we plot the quantity
\begin{align*}
  \max_{j \in \until{N}} \: \big[ 1 - \ell_j((w_i^t)^\top p_j + b_i^t) \big].
\end{align*}
Notice that the algorithm evolves in an outer-approximation fashion, that is,
the solution estimates are infeasible for problem~\eqref{eq:soft_svm_problem}
until the optimal solution is found. This can also be seen by noting in Figure~\ref{ce:fig_cost_error}
that the costs
associated to the intermediate solution estimates are lower than the optimal cost
of the problem.
\begin{figure}[!htpb]
\centering
  \includegraphics[scale=1]{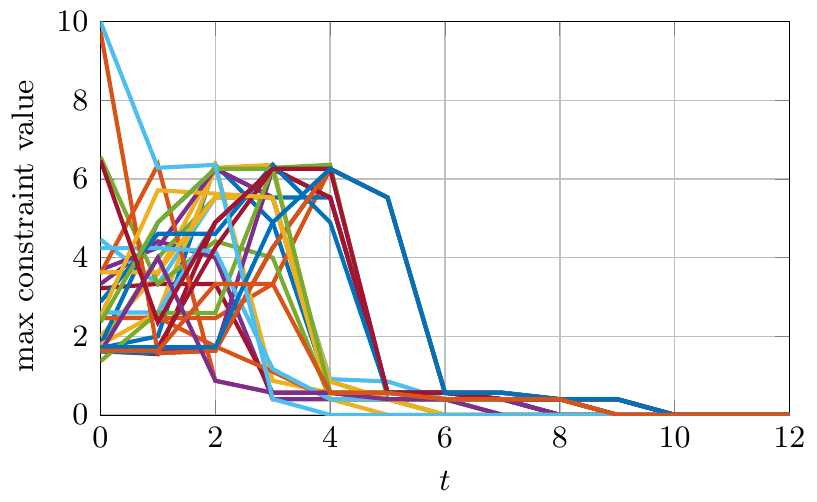}
\caption{
  Evolution of the maximum value of the constraints for the solution estimate
  $\bx_i^t$ computed by each agent in the \ConstrCons/ algorithm.
  Each line refers to an agent in the network.}
  \label{ce:fig_max_constraint}
\end{figure}

In Figure~\ref{ce:fig_solution_error} we show the distance of the local
solution estimates from the lex-optimal solution $\bx^\star$ of problem~\eqref{eq:soft_svm_problem},
i.e., $\|\bx_i^t - \bx^\star\|$, for all $i$.
\begin{figure}[!htpb]
\centering
  \includegraphics[scale=1]{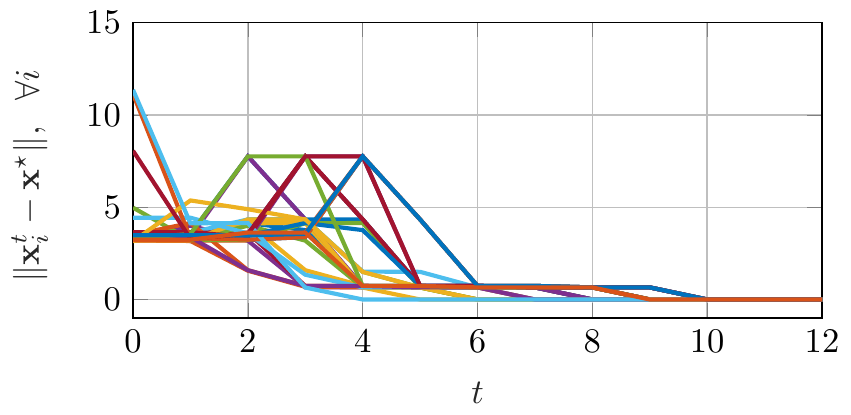}
\caption{
  Evolution of the distance $\|\bx_i^t - \bx^\star\|$ of the local solution estimates $\bx_i^t$
  from the lex-optimal solution $\bx^\star$ for the \ConstrCons/ algorithm. Each line refers to an agent in the network.}
  \label{ce:fig_solution_error}
\end{figure}

\fi

\iftrue

\chapter*{Concluding Remarks}
\addcontentsline{toc}{chapter}{Concluding Remarks}

In this survey, we considered a distributed optimization framework arising in
modern cyber-physical networks, in which computing units have only a partial
knowledge of a global optimization problem and must solve it through local
computation and communication without any central coordinator.
First, we introduced main optimization set-ups addressed in distributed
optimization (i.e., cost-coupled, common-cost, and constraint-coupled), and
motivated them with relevant estimation, learning, decision and control
applications arising in smart networks.
Then, we reviewed three main approaches to design distributed optimization
algorithms, namely (primal) consensus-based, duality-based and
constraint-exchange methods, and provided a theoretical analysis under
simplified communication assumptions and/or problem set-ups.
To highlight the behavior of the presented algorithms, the theoretical
results are also equipped with numerical examples.

\fi

\iftrue

\appendix
\chapter{Centralized Optimization Methods}
\label{app:basics}

\section{Gradient Method}
\label{app:gradient_method}

Consider the following unconstrained optimization problem
\begin{align}
  \min_{\bx \in \real^d} \: & \: f (\bx ),
  \label{app:gradient_problem}
\end{align}
where $\map{f}{\real^d}{\real}$.
The gradient method is an iterative algorithm given by
\begin{align}
  \bx^{t+1} = \bx^t - \gamma^t \nabla f( \bx^t ),
\label{app:gradient_iteration}
\end{align}
where $t\ge 0$ denotes the iteration counter and $\gamma^t $ is the
step-size.
The following result states the convergence of the gradient method for constant 
step-size.
\begin{proposition}[{\cite[Proposition 1.2.3]{bertsekas1999nonlinear}}]
\label{app:prop:gradient_method}
  Assume that $f$ is a $\CC^1$ function with Lipschitz continuous gradient $\nabla f$ with constant $L$.
  Let the step-size be constant, i.e., $\gamma^t = \gamma$, for all $t\ge0$, and such that
  $0 < \gamma < 2/L$. 
  Then, every limit point of the sequence $\{ \bx^t \}_{t\ge0}$ generated by the 
  gradient method~\eqref{app:gradient_iteration}, is a stationary point
  of problem~\eqref{app:gradient_problem}, i.e., there exists a subset of 
  indices $\KK \subseteq \natural$ such that
  \begin{align*}
  \lim_{ \KK \ni t \to \infty } 
  \| \bx^t - \bar{\bx} \| = 0,
  \end{align*}
  where $\bar{\bx}$ is a stationary point of~\eqref{app:gradient_problem}.\oprocend
\end{proposition}

The previous result can be extended in several ways, e.g., with different step-size rules and
adapted to constrained problems. We refer the interested reader to 
\cite{bertsekas1999nonlinear} and references therein.

\section{Subgradient Method}
\label{app:subgradient_method}

Consider the following constrained optimization problem
\begin{align}
  \min_{\bx \in X} \: & \: f (\bx ),
  \label{app:subgradient_problem}
\end{align}
with $\map{f}{\real^d}{\real}$ a convex function and $X \subseteq \real^d$
a closed, convex set.

A vector $\widetilde{\nabla} f( \bx ) \in\real^d$ is called a subgradient of the 
convex function $f$ at $\bx \in\real^d$ if
\begin{align*}
  f(\by) 
  \ge 
  f(\bx) + \widetilde{\nabla} f( \bx ) ^\top  (\by - \bx)
\end{align*}
for all $\by \in \real^d$.
The (projected) subgradient method is the iterative algorithm
given by
\begin{align}
  \bx^{t+1} = \PP_X \Big( \bx^t - \gamma^t \widetilde{\nabla} f( \bx^t ) \Big) \: ,
\label{app:subgradient_iteration}
\end{align}
where $t\ge 0$ denotes the iteration counter, $\gamma^t $ is the
step-size, $\widetilde{\nabla} f( \bx^t )$ denotes a subgradient of $f$ at
$\bx^t $, and $\PP_X(\, \cdot\, )$ is the Euclidean projection onto $X$.

\begin{assumption}[Diminishing Step-size]
\label{app:stepsize_assumption}
The step-size sequence $\{ \gamma^t \}_{t\ge 0}$ is such that 
$\gamma^t \ge 0$ and satisfies
\begin{align*}
  \lim_{t\to\infty} \gamma^t = 0,
  \hspace{0.6cm}
  \smallsum_{t=0}^\infty \gamma^t = \infty,
  \hspace{0.6cm}
  \smallsum_{t=0}^\infty ( \gamma^t) ^2 < \infty.
  \eqoprocend
\end{align*}
\end{assumption}

The following proposition formally states the convergence of the subgradient 
method~\eqref{app:subgradient_iteration}. 
\begin{proposition}[{\cite[Proposition 3.2.6]{bertsekas2015convex}}]
\label{app:subgradient_convergence}
  Assume that all the subgradients of $f$ are bounded at each
  $\bx \in X$. Moreover, assume the optimal solution set of 
  problem~\eqref{app:subgradient_problem} is not empty.
  Let the step-size $\gamma^t$ satisfy Assumption~\ref{app:stepsize_assumption}. 
  Then, the sequence $\{ \bx^t \}_{t\ge0}$ generated by the 
  subgradient method~\eqref{app:subgradient_iteration} converges 
  to an optimal solution $\bx^\star$ of problem~\eqref{app:subgradient_problem}, i.e.,
  \begin{align*}
  \lim_{ t \to \infty } 
  \| \bx^t - \bx^\star\| = 0, \hspace{1cm} \lim_{ t \to \infty } \| f(\bx^t) - f^\star \| = 0.
  \eqoprocend
  \end{align*}
\end{proposition}

\section{Lagrangian Duality and Dual Subgradient Method}
\label{sec:appendix_duality}

Consider a constrained optimization problem, addressed as primal problem,
having the form
\begin{align}
\begin{split}
  \min_{ \bx \in X } \:& \: f(\bx)
  \\
  \subj \: & \: \bg (\bx) \le \0,
\end{split}
\label{eq:appendix_primal}
\end{align}
where $X \subseteq \real^d$ is a convex, compact set,
$\map{f}{\real^d}{\real}$ is a convex function and $\map{\bg}{\real^d}{\real^S}$
is such that each component $\map{\bg_s}{\real^d}{\real}$,
$s \in \until{S}$, is a convex (scalar) function.

The following optimization problem
\begin{align}
\begin{split}
  \max_{\bmu} \:& \: q(\bmu)
  \\
  \subj \: & \: \bmu \ge \0
\end{split}
\label{eq:appendix_dual}
\end{align}
is called the dual of problem~\eqref{eq:appendix_primal}, where
$\map{q}{\real^S}{\real}$ is obtained by minimizing with respect to $\bx \in X$
the Lagrangian function $\LL (\bx,\bmu) = f(\bx) + \bmu^\top \bg(\bx)$, i.e.,
$q(\bmu) = \min_{\bx \in X} \LL(\bx,\bmu)$.
It can be shown that the domain of $q$ (i.e., the set of $\bmu$ such that $q(\bmu) > -\infty$)
is convex and that $q$ is concave on its domain.
A vector $\bar{\bmu} \in\real^S$ is said to be a Lagrange multiplier if 
it holds $\bar{\bmu} \geq \0$ and
\begin{align*}
  \inf_{\bx \in X} \LL (\bx,\bar{\bmu}) 
  = \inf_{\bx\in X \: : \: \bg(\bx) \le \0} \, f(\bx).
\end{align*}

It can be shown that the following inequality holds \cite{bertsekas1999nonlinear}
\begin{align}
  \inf_{\bx \in X } \sup_{\bmu \geq \0} \LL(\bx, \bmu) \ge 
  \sup_{\bmu\geq \0} \inf_{\bx \in X } \LL(\bx,\bmu),
  \label{eq:appendix_weak_duality}
\end{align}
which is called weak duality.
When in~\eqref{eq:appendix_weak_duality} the equality holds, then we say
that strong duality holds and, thus, solving the primal
problem~\eqref{eq:appendix_primal} is equivalent to solving its dual
formulation~\eqref{eq:appendix_dual}. In this case the right-hand-side
problem in~\eqref{eq:appendix_weak_duality} is referred to as
\emph{saddle-point problem} of~\eqref{eq:appendix_primal}.

\begin{definition}
  A pair $(\bx^\star , \bmu^\star)$ is called a primal-dual optimal solution of
  problem~\eqref{eq:appendix_primal} if $\bx^\star \in X$ and
  $\bmu^\star \geq \0$, and $(\bx^\star , \bmu^\star)$ is a saddle point of the
  Lagrangian, i.e.,
  \begin{align*}
    \LL (\bx^\star , \bmu ) \le \LL (\bx^\star,\bmu^\star) \le \LL (\bx,\bmu^\star)
  \end{align*}
  for all $\bx \in X$ and $\bmu\geq \0$.\oprocend
\label{def:primal_dual_pair}
\end{definition}

Given the dual function $q$, an important property is as follows.
A subgradient of $-q$ at a given $\bar{\bmu}$ can be 
efficiently computed as $g ( \bar{\bx} )$,
where 
$\bar{\bx} = \argmin_{\bx \in X } \: f(\bx) + \bar{\bmu}^\top g(\bx)$
(see \cite[Section~6]{bertsekas1999nonlinear} for further details).
Then, a subgradient method to 
solve the dual problem~\eqref{eq:appendix_dual} reads
\begin{align*}
\begin{split}
  \bx^{t+1} & = \argmin_{\bx \in X } \: f(\bx) + ( \bmu^t)^\top g(\bx)
  \\
  \bmu^{t+1} & = \PP_{ \bmu \ge 0}  \Big( \bmu^t + \gamma^t g(\bx^{t+1}) \Big),
\end{split}
\end{align*}
where $\gamma^t$ is a suitable step-size and $\bmu^0 \ge 0$ is arbitrary.

\section{ADMM Algorithm}
\label{app:ADMM}

In this section, we review the Alternating Direction Method of Multipliers (ADMM) 
following~\cite[Section~3.4]{bertsekas1989parallel}.
Consider the following optimization problem
\begin{align}
\begin{split}
  \min_{\bx \in \real^d} \: & \:  G_1 (\bx) + G_2 (A \bx)
  \\
  \subj \: & \: \bx \in C_1, \: A\bx \in C_2,
\end{split}
\label{app:ADMM_problem_original}
\end{align}
where $\map{G_1}{\real^d}{\real}$ and $\map{G_2}{\real^S}{\real}$ are convex functions,
$A$ is a ${S \times d}$ matrix, and
$C_1 \subseteq \real^d$ and $C_2 \subseteq \real^S$ are nonempty, closed convex sets.
We assume that the optimal solution set $X^\star$ of problem~\eqref{app:ADMM_problem_original} 
is nonempty. Furthermore, either $C_1$ is bounded or else $A^\top A$ is invertible.

Problem~\eqref{app:ADMM_problem_original} can be equivalently rewritten as
\begin{align}
\label{app:ADMM_problem}
\begin{split}
  \min_{\bx \in \real^d, \bz\in\real^S} \: & \:  G_1 (\bx) + G_2 (\bz)
  \\
  \subj \: 
  & \: A \bx = \bz,
  \\
  & \: \bx \in C_1, \: \bz \in C_2.
\end{split}
\end{align}

Let $\blambda\in \real^S$ be a multiplier associated to the equality constraint $A \bx = \bz$
and introduce the \emph{augmented} Lagrangian of problem~\eqref{app:ADMM_problem}
\begin{align*}
  \LL_\rho ( \bx, \bz, \blambda) 
  = 
  G_1 (\bx) + G_2 (\bz) + \blambda^\top ( A \bx - \bz) + \frac{\rho}{2} \| A \bx - \bz \|^2
\end{align*}
where $\rho > 0$ is a penalty parameter.
The ADMM algorithm is an iterative procedure
in which at each iteration $t\ge0$, the following steps are performed
\begin{subequations}
\label{app:ADMM_algorithm}
\begin{align}
  \label{app:ADMM_x}
  \bx^{t+1} & = \argmin_{\bx \in C_1} \, \LL_\rho (\bx, \bz^t,\blambda^t)
  \\
  \label{app:ADMM_z}
  \bz^{t+1} & = \argmin_{\bz \in C_2} \, \LL_\rho (\bx^{t+1} , \bz,\blambda^t)
  \\
  \label{app:ADMM_lambda}
  \blambda^{t+1} & = \blambda^t + \rho \, ( A \bx^{t+1} - \bz^{t+1}),
\end{align}
\end{subequations}
where the initialization of the variables $\bz^0$ and $\blambda^0$ can be arbitrary.

The ADMM algorithm is very similar to dual ascent and to the Method of Multipliers (MM): it consists of 
an $\bx$-minimization, a $\bz$-minimization, and a dual variable update. As in the method of 
multipliers, the dual variable update uses a step-size equal to the augmented Lagrangian parameter $\rho$.
In the MM, the augmented Lagrangian $\LL_\rho$ is minimized jointly with respect to the two primal 
variables. In ADMM, on the other hand, $\bx$ and $\bz$ are updated in an alternating or sequential 
fashion, which accounts for the term \emph{alternating direction}.

\begin{proposition}[{\cite[Proposition~4.2]{bertsekas1989parallel}}]
\label{app:prop:ADMM_convergence}
  Consider a sequence $\{ \bx^t, \bz^t, \blambda^t \}_{t\ge0}$ generated by the ADMM 
  algorithm~\eqref{app:ADMM_algorithm}. Then, the generated sequence is bounded and
  every limit point of $\{ \bx^t\}_{t\ge0}$ is an optimal solution of problem~\eqref{app:ADMM_problem_original}.
  Furthermore, the sequence $\{\blambda^t \}_{t\ge0}$ converges to an optimal solution
  of the dual of problem~\eqref{app:ADMM_problem_original}.\oprocend
\end{proposition}

In~\cite{boyd2011distributed} a more general problem set-up for ADMM is considered.
Specifically, let us consider a two-variable problem defined as
  \begin{align}
  \label{app:ADMM_problem_Boyd}
  \begin{split}
    \min_{\bx \in \real^d, \bz\in\real^S} \: & \:  G_1 (\bx) + G_2 (\bz)
    \\
    \subj \: 
    & \: A \bx + B \bz + c = 0
    \\
    & \: \bx \in C_1, \: \bz \in C_2.
  \end{split}
  \end{align}
  with $A\in\real^{p\times d}$, $B\in\real^{p\times S} $ and $c\in\real^{p\times 1}$. 
  Then, the ADMM algorithm applied to problem~\eqref{app:ADMM_problem_Boyd} 
  reads as follows
  \begin{subequations}
  	\label{app:ADMM_algorithm_Boyd}
	\begin{align}
	  \bx^{t+1} & = \argmin_{\bx \in C_1} \, \LL_\rho (\bx, \bz^t,\blambda^t)
	  \\
	  \bz^{t+1} & = \argmin_{\bz \in C_2} \, \LL_\rho (\bx^{t+1} , \bz,\blambda^t)
	  \\
	  \blambda^{t+1} & = \blambda^t + \rho \, ( A \bx^{t+1} + B \bz^{t+1} + c),
	\end{align}
  \end{subequations}
	where the augmented Lagrangian is defined as
	\begin{align*}
	  \LL_\rho (\bx, \bz,\blambda) 
	  & =
    G_1 (\bx) + G_2 (\bz) + \blambda^\top ( A \bx + B \bz + c) 
    + \frac{\rho}{2} \| A \bx + B \bz + c \|^2.
	\end{align*}

\chapter{Consensus Over Networks}
\label{sec:consensus_appendix}

Consensus and distributed averaging are fundamental building blocks in 
distributed optimization.

We introduce the consensus problem for a group of $N$ agents that considers conditions under 
which, using a certain message-passing protocol, the local variables of each agent
converge to the same value. There exist several
results related to the convergence of local variables to a common
value using various information exchange protocols among agents. 

\section{Average Consensus over Static Networks}
\label{sec:consensus_static_network}
One of the most used models for consensus is based on the following 
discrete-time iteration: to generate an estimate at iteration $t + 1$, 
agent $i$ forms a convex combination of its current estimate $\bz_i^t$ with the estimates 
received from other agents as
\begin{align}
\label{app:average_consensus}
  \bz_i^{t+1} = \smallsum_{j\in\nbrs_i} a_{ij} \, \bz_j^t,
\end{align}
where $a_{ij}$ denotes a (positive) weight that agent $i$ assigns to each
neighbor $j$, and we recall that $\nbrs_i$ is the set of neighbors of agent $i$ in
the (static) undirected communication graph.
The weights $a_{ij}$ are set to zero if $i$ and $j$ are not neighbors in the communication graph $\GG$ and are 
doubly stochastic, i.e., they satisfy $\sum_{j=1}^N a_{ij} = 1$, for all $i\in\until{N}$,
and $\sum_{i=1}^N a_{ij} = 1$, for all $j\in\until{N}$.

The consensus algorithm can be written in an aggregate form by stacking all the agents'
estimates in a single variable which evolves according to
\begin{align}
\label{app:average_consensus_aggregated}
  \bz^{t+1} 
  =
  \begin{bmatrix}
  \bz_1^{t+1} 
  \\
  \vdots
  \\
  \bz_N^{t+1} 
  \end{bmatrix}
  = A \bz^t,
\end{align}
where $A$ is a matrix whose $(i,j)$-th entry is $a_{ij}$ for all $i,j\in\until{N}$.

A useful property of doubly stochastic matrices is the following. Given $A$ 
doubly stochastic, it holds
\begin{align*}
  \| A \bz - \avgz \|
  \le
  \sigma_A 
  \| \bz - \avgz \|,
\end{align*}
where $\avgz \triangleq \frac{1}{N} \sum_{i=1}^N \bz_i$ and $\sigma_A$ is the 
spectral radius of $A - \1\1^\top/N$. It can be proven (see \cite{xiao2004fast})
that if the graph is connected and
$A$ is doubly stochastic, then $\sigma_A \in (0,1)$, and specifically
$\sigma_A = \max\{ |\lambda_2|, |\lambda_N| \}$,
where $\lambda_h$ denotes the $h$-th largest eigenvalue of $A$.

\begin{theorem}
  Let $\GG$ be a connected graph and let $a_{ij}$, $i,j\in\until{N}$
  be doubly stochastic weights matching the graph.
  Then, the sequences $\{\bz_i^t\}_{t\ge0}$, $i\in\until{N}$,
  generated by~\eqref{app:average_consensus} satisfy
  \begin{align*}
    \lim_{t\to\infty} \| \bz_i^t - \avgz^0 \| = 0,
  \end{align*}
  for all $i\in\until{N}$, where $\avgz^0 = \frac{1}{N} \sum_{i=1}^N \bz_i^0$.~\oprocend
\end{theorem}

Several extensions of the basic consensus algorithm~\eqref{app:average_consensus} exist. 
For instance, one can consider time-varying networks that have some long-term 
connectivity properties. The consensus algorithm needs to be adapted to 
accommodate the time-varying network by considering time-varying weights 
$a_{ij}^t$. Also, it is possible to design a consensus algorithm that works under 
delays and is robust to packet losses. See~\cite{hadjicostis2018distributed} for a recent
survey on this topic.
Next, we describe another extension in which the consensus algorithm is tailored for directed 
networks.

\section{Push-sum Consensus over Directed Networks}
\label{sec:consensus_push-sum}
In this section we describe how the average consensus algorithm can be adapted 
to work on directed networks. This algorithm is known as push-sum algorithm 
and has been introduced in~\cite{benezit2010weighted}.

In directed networks is not always possible to construct a doubly stochastic 
matrix $A$, while a column stochastic matrix is always available. We use
$B$ to denote a column stochastic matrix, i.e., such that $\1^\top B = \1^\top$.
Formally, the push-sum consensus reads
\begin{subequations}
\label{app:push_sum_consensus}
\begin{align}
  \phi_i^{t+1} & = \smallsum_{j\in\nbrs_i} b_{ij} \, \phi_j^t
  \\
  \bs_i^{t+1} & = \smallsum_{j\in\nbrs_i} b_{ij} \, \bs_j^t
  \\
  \bz_i^{t+1} & = \frac{ \bs_i^{t+1}}{\phi_i^{t+1}},
\end{align}
\end{subequations}
with the initial values $\phi_i^0 = 1$ for all $i\in\until{N}$.

The convergence of this scheme has been proven in~\cite{benezit2010weighted}, 
i.e., the sequences $\{\bz_i^t\}_{t\ge0}$, $i\in\until{N}$, generated 
by~\eqref{app:push_sum_consensus} satisfy
\begin{align*}
  \lim_{t\to\infty}
  \| \bz_i^t - \avgz^0 \| = 0,
\end{align*}
for all $i\in\until{N}$, where $\avgz^0 = \frac{1}{N} \sum_{i=1}^N \bz_i^0$.

\section{Dynamic Average Consensus Algorithm}
\label{sec:dynamic_average_consensus}

In this section, we present a distributed algorithm to achieve dynamic average consensus
that has been proposed in~\cite{zhu2010discrete}. See also \cite{kia2018tutorial} for
a very recent tutorial. 

We consider a network of $N$ agents in which each agent $i$ 
is able to measure a local discrete-time signal $\{ \br_i^t \}_{t \ge 0}$.
The goal is to design a distributed algorithm that enables agents 
to eventually track the average of their signal $\br_i^t$, $i \in\until{N}$, 
by means of local communication only.

The dynamic consensus algorithm proposed in~\cite{zhu2010discrete} consists in a
consensus-based procedure in which each agent maintains a local estimate $\bz_i^t$ of the average.
The local estimate is iteratively updated according to
\begin{align}
\label{app:running_consensus_alg}
  \bz_i^{t + 1} 
  = 
  \smallsum_{j\in\nbrs_i} a_{ij} \, \bz_j^t + \big( \br_i^{t+1} - \br_i^t \big),
\end{align}
where $a_{ij}$ are entries of a doubly stochastic matrix.

If the input signals $\br_i^t$ asymptotically converge to a constant value, 
then the dynamic average consensus algorithm in~\eqref{app:running_consensus_alg} 
is guaranteed to converge, 
i.e., for all $i\in\until{N}$, it holds
\begin{align*}
  \lim_{t\to\infty} \| \bz_i^t - \bar{\br}^t \| = 0,
\end{align*}
where $\bar{\br}^t = \frac{1}{N} \sum_{i=1}^N \br_i^t$ for all $t\ge 0$.

The interested reader can find a rigorous treatment and a more
comprehensive discussion on this class of algorithms in~\cite{zhu2010discrete,kia2018tutorial}.

\chapter{Linear Programming}
\label{sec:LP}

A Linear Program (LP) is an optimization problem with linear cost function
and linear constraints:
\begin{align}
\begin{split}
  \min_\bx \: & \: c^\top \bx
  \\
  \subj \: & \: a_k^\top \bx \le b_k, \hspace{0.5cm} k \in \until{K},
\end{split}
\label{LP:eq:LP_general}
\end{align}
where $c \in \real^d$ is the cost vector and $a_k \in \real^d$ and
$b_k \in \real$ describe $K$ inequality constraints.
In the subsequent discussion, we assume that $d \le K$.
The feasible set $\XX$ of problem~\eqref{LP:eq:LP_general} is
the set of vectors satisfying all the constraints, i.e.,
\begin{align*}
  \XX \triangleq \{ \bx \in \real^d \mid a_k^\top \bx \le b_k \text{ for all } k \in \until{K} \}.
\end{align*}
Note that $\XX$ is a polyhedron, for which the following definition of
vertex can be given.
\begin{definition}
	A vector $\tilde{\bx} \in \real^d$ is a vertex of $\XX$ if there exists some
	$c \in \real^d$ 	such that $c^\top \tilde{\bx} < c^\top \bx$ for all
	$\bx \in \XX$ with $\bx \ne \tilde{\bx}$.
	\oprocend
\end{definition}
If problem~\eqref{LP:eq:LP_general} admits an optimal solution, it can be shown that
there exists an optimal vertex, i.e., a vertex which is an optimal solution of the problem
(see, e.g., \cite[Theorem 2.7]{bertsimas1997introduction}).
Let $\bx^\star$ be an optimal vertex of problem~\eqref{LP:eq:LP_general}.
Then, it is a standard result in linear programming theory that there exists an index set
$\{\ell_1, \ldots, \ell_d\} \subset \until{K}$, with cardinality $d$,
such that $\bx^\star$ is the unique optimal vertex of the problem
\begin{align*}
\begin{split}
  \min_\bx \: & \: c^\top \bx
  \\
  \subj \: & \: a_{\ell_h}^\top \bx \le b_{\ell_h}, \hspace{0.5cm} h \in \until{d},
\end{split}
\end{align*}
which is a relaxed version of problem~\eqref{LP:eq:LP_general} in which
only $d$ constraints are considered. In addition, the vectors $a_{\ell_h},
h \in \until{d}$ are linearly independent, so that they form a basis of $\real^d$.
By analogy, the constraints $a_{\ell_h}^\top \bx \le b_{\ell_h}, h \in \until{d}$
are called a \emph{basis} of the point $\bx^\star$. Due to the optimality of $\bx^\star$, we call it also a
basis of problem~\eqref{LP:eq:LP_general}.
To compactly denote such basis, we introduce a matrix $P \in \real^{d \times d}$,
obtained by stacking the row vectors $a_{\ell_h}^\top$, and a vector $q \in \real^d$,
obtained by stacking the scalars $b_{\ell_h}$, i.e.,
\begin{align*}
  P = \begin{bmatrix}
    a_{\ell_1}^\top
    \\
    \vdots
    \\
    a_{\ell_d}^\top
  \end{bmatrix},
  \hspace{0.5cm}
  q = \begin{bmatrix}
    b_{\ell_1}
    \\
    \vdots
    \\
    b_{\ell_d}
  \end{bmatrix}.
\end{align*}
Then, $\bx^\star = P^{-1} q$, and we say that the tuple $(P, q)$ is a basis of~\eqref{LP:eq:LP_general}.

If problem~\eqref{LP:eq:LP_general} has multiple optimal solutions,
we say that the LP is \emph{dual degenerate}.
In presence of dual degeneracy, it is not trivial to guarantee convergence
of distributed algorithms to the same optimal solution.
In order to overcome this issue, it is possible to rely on
the lexicographic ordering of vectors. We now give some definitions.
\begin{definition}
  A vector $\bv \in \real^n$ is said to be \emph{lexicographically positive}
  (or \emph{lex-positive}) if $\bv \ne \0$ and the first non-zero component
  of $\bv$ is positive. In symbols:
  \begin{align*}
    \bu \lexsucc \0.
  \end{align*}
  
  A vector $\bu \in \real^n$ is said to be \emph{lexicographically larger} (resp.
  \emph{smaller}) than another vector $\bv \in \real^n$ if $\bu - \bv$ is lex-positive
  (resp. $\bv - \bu$ is lex-positive), or, equivalently, if $\bu \ne \bv$ and the first nonzero
  component of $\bu - \bv$ is positive (resp., negative). In symbols:
  \begin{align*}
    \bu \lexsucc \bv
    \hspace{0.5cm}
    \text{or}
    \hspace{0.5cm}
    \bu \lexprec \bv.
  \end{align*}
  
  Given a set of vectors $\{ \bv_1, \ldots, \bv_r \}$, the lexicographic minimum
  is the element $\bv_i$ such that $\bv_j \lexsucc \bv_i$ for all $j \ne i$.
  In symbols:
  \begin{align*}
    \bv_i = \lexmin \{ \bv_1, \ldots, \bv_r \}.
    \eqoprocend
  \end{align*}
\end{definition}

Now, consider the optimal solution set of problem~\eqref{LP:eq:LP_general}, i.e.,
$\XX^\star \triangleq \{ \bx \in \XX \mid c^\top \bx \le c^\top \bx^\prime
\text{ for all } \bx^\prime \in \XX \} \subseteq \XX$,
where $\XX$ is the feasible set of problem~\eqref{LP:eq:LP_general}.
Among all the optimal solutions in $\XX^\star$, it is possible to
compute the lexicographically minimal one, i.e., $\lexmin (\SS^\star)$.
It turns out that finding $\lexmin (\SS^\star)$ is equivalent to finding the (unique)
optimal solution to a modified (non dual-degenerate) version of the original
problem~\eqref{LP:eq:LP_general}, where the cost vector $c$ is perturbed
to $c^\prime = c + \Delta$, with $\Delta$ a lexicographic perturbation vector:
\begin{align*}
  \Delta^\top = [ \Delta_0 \:\: \Delta_0^2 \:\: \ldots \:\: \Delta_0^d ],
\end{align*}
for a sufficiently small $\Delta_0 > 0$ (see \cite{jones2007lexicographic}).
Therefore, the lex-optimal solution of problem~\eqref{LP:eq:LP_general}
is the \emph{unique} optimal solution of the problem with perturbed cost
\begin{align}
\begin{split}
  \min_\bx \: & \: (c + \Delta)^\top \bx
  \\
  \subj \: & \: a_k^\top \bx \le b_k, \hspace{0.5cm} k \in \until{K}.
\end{split}
\label{LP:eq:LP_perturbed}
\end{align}
Thus, the lex-optimal solution
of problem~\eqref{LP:eq:LP_general} exists if and only if
problem~\eqref{LP:eq:LP_perturbed} admits an optimal solution.
Moreover, the optimal solution of~\eqref{LP:eq:LP_perturbed} is attained
at a vertex of~\eqref{LP:eq:LP_general}, therefore it is an
optimal vertex of problem~\eqref{LP:eq:LP_general}.

\bigskip \bigskip

\noindent\textbf{Acknowledgements}

This work is part of a project that has received funding from the European 
Research Council (ERC) under the European Union's Horizon 2020 research 
and innovation programme (grant agreement No 638992 - OPT4SMART)

\vspace{2.cm}

\centering
\begin{tikzpicture}
  \node[anchor=east] at(0,0) {\includegraphics[height=2.3cm]{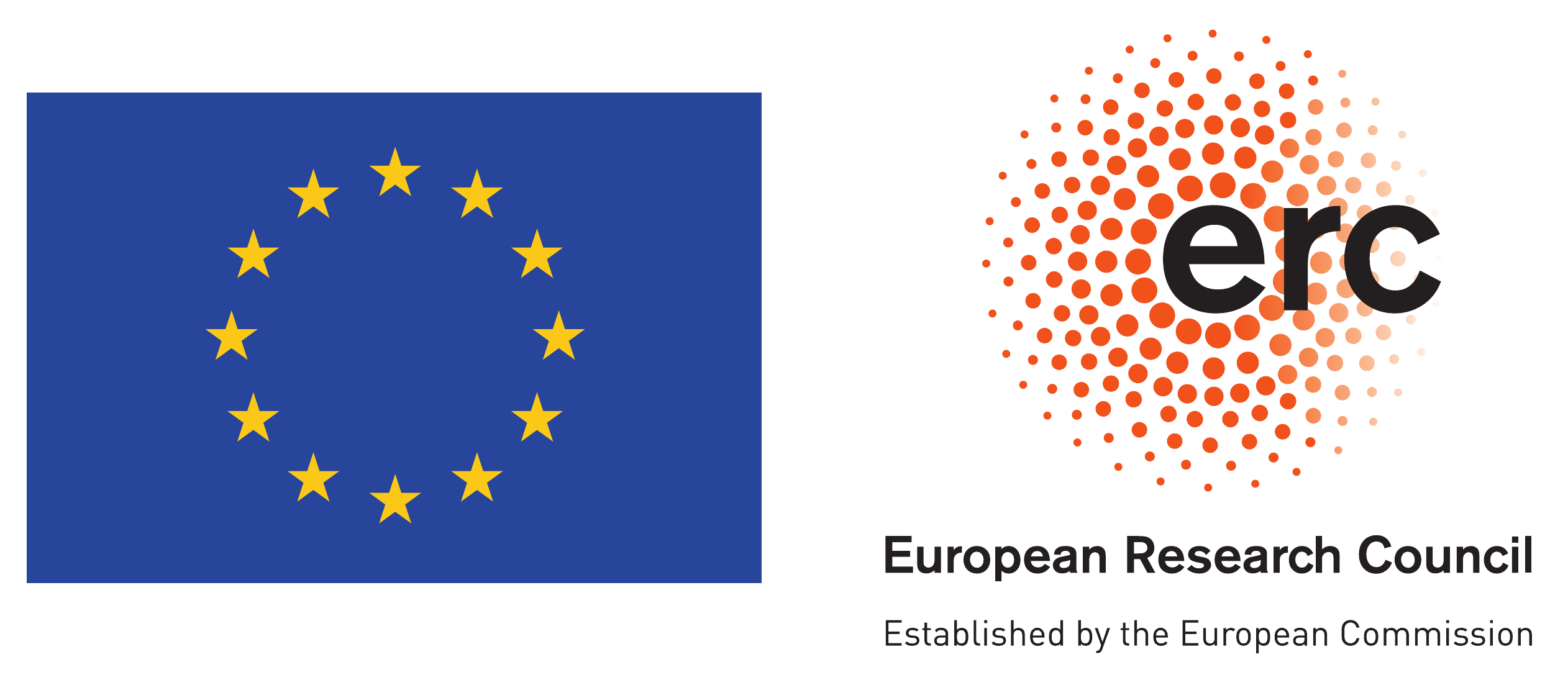}};
  \node[anchor=west] at(1,0.2) {\includegraphics[height=1.8cm]{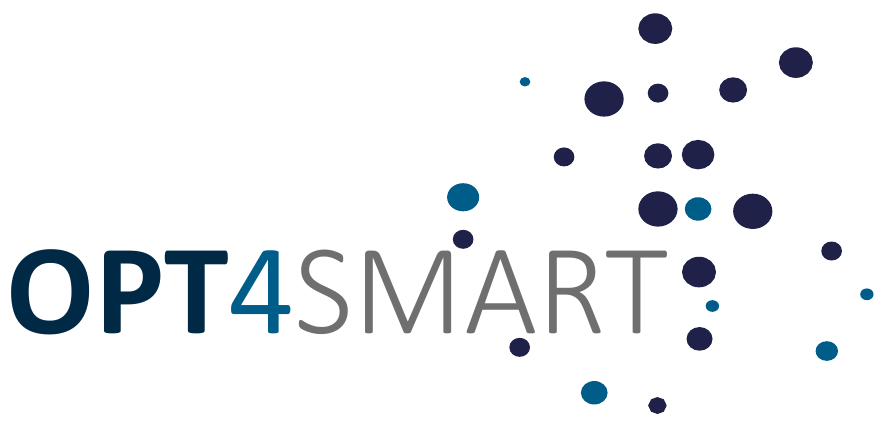}};
\end{tikzpicture}

\fi

\clearpage

\bibliographystyle{IEEEtran}
\bibliography{distributed_tutorial}

\end{document}